\documentclass{IEEEtran} 
\usepackage{tikz}
\usepackage{amsmath,amssymb,bm,amsthm}
\usepackage{graphicx}
\usepackage{caption}
\usepackage{xcolor,colortbl}
\usepackage{soul} 
\usepackage{diagbox}  
\usepackage{fancyhdr}  
\usepackage{hyperref}  
\usepackage{cleveref}
\usepackage{cite}
\usepackage{enumerate}
\usepackage{url}
\usepackage{balance}   
\allowdisplaybreaks[2]  
\usepackage{pdflscape}
\usepackage{multirow}
\usepackage{hhline}
\usepackage{array}
\newcolumntype{C}[1]{>{\centering\let\newline\\\arraybackslash\hspace{0pt}}m{#1}} 

\usepackage{fancyhdr}  
\begin{document}
	\title{The Explicit Coding Rate Region of \\ Symmetric Multilevel Diversity Coding}   
	\author{Tao Guo and Raymond W. Yeung\thanks{This paper was presented in part at 2018 IEEE International Symposium on Information Theory (ISIT) \cite{guo-yeung-isit2018}.
			
	Tao Guo was with the Institute of Network Coding and the Department of Information Engineering, The Chinese University of Hong Kong, Hong Kong SAR, China. He is now with the Department of Electrical and Computer Engineering, Texas A\&M University, College Station, TX, USA. (\mbox{e-mail:} \href{mailto:guotao@tamu.edu}{guotao@tamu.edu})
			
	Raymond~Yeung is with the Institute of Network Coding and the Department of Information Engineering, The Chinese University of Hong Kong, Hong Kong SAR, China. He also holds an adjunct position at the School of Science and Engineering, The Chinese University of Hong Kong, Shenzhen, China. (\mbox{e-mail:} \href{mailto:whyeung@ie.cuhk.edu.hk}{whyeung@ie.cuhk.edu.hk})
	
	Their work was funded in part by the University Grant Committee of the Hong Kong SAR, China (Project No.\ AoE/E-02/08), the National Natural Science Foundation of China (Grant No. 61471215), and the Science, Technology and Innovation Commission of Shenzhen Municipality (Project No. JSGG20160301170514984).}
	}
	\maketitle
	\newcommand{\reffig}[1]{Fig. \ref{#1}}
	\newcommand{\cA}{\mathcal{A}}
	\newcommand{\cB}{\mathcal{B}}
	\newcommand{\cD}{\mathcal{D}}
	\newcommand{\cF}{\mathcal{F}}
	\newcommand{\cG}{\mathcal{G}}
	\newcommand{\cI}{\mathcal{I}}
	\newcommand{\cJ}{\mathcal{J}}
	\newcommand{\cK}{\mathcal{K}}
	\newcommand{\cL}{\mathcal{L}}
	\newcommand{\cN}{\mathcal{N}}
	\newcommand{\cO}{\mathcal{O}}
	\newcommand{\cR}{\mathcal{R}}
	\newcommand{\cT}{\mathcal{T}}
	\newcommand{\cU}{\mathcal{U}}
	\newcommand{\cX}{\mathcal{X}}
	
	\theoremstyle{plain}
	\newtheorem{theorem}{Theorem}
	\newtheorem{lemma}{Lemma}
	\newtheorem{corollary}{Corollary}[lemma]
	\newtheorem{proposition}{Proposition}
	\newtheorem{conjecture}{Conjecture}
	\newtheorem{claim}{Claim}
	
	\theoremstyle{remark}
	\newtheorem{remark}{Remark}
	
	\theoremstyle{definition}
	\newtheorem{definition}{Definition}
	\newtheorem{example}{Example}
	
	
	\begin{abstract}
		It is well known that {\em superposition coding}, namely separately encoding the independent sources,
		is optimal for symmetric multilevel diversity coding (SMDC) (Yeung-Zhang 1999) for any $L\geq 2$, where $L$ is the number of levels of the coding system.	
		However, the characterization of the coding rate region therein involves uncountably many linear inequalities 
		and the constant term (i.e., the lower bound) in each inequality is given in terms of the solution of a linear optimization problem. 
		Thus this implicit characterization of the coding rate region does not enable the determination of the achievability of a given rate tuple. 
		In principle, the achievability of a given rate tuple can be determined by direct computation, but the complexity is prohibitive even for $L=5$.
		In this paper, for any fixed $L$, we obtain in closed form a finite set of linear inequalities for characterizing the coding rate region.
		We further show by the symmetry of the problem that only a much smaller subset of this finite set of inequalities needs to be verified in determining the achievability of a given rate tuple. Yet, the cardinality of this smaller set grows at least exponentially fast with $L$. We also present a subset entropy inequality, which together with our explicit characterization of the coding rate region, is sufficient for proving the optimality of superposition coding. 
	\end{abstract}
	\section{Introduction}  \label{section-introduction}
	{\it Multilevel diversity coding} was introduced by Yeung \cite{yeung95}. 
	In a multilevel diversity coding system, an information source is encoded by a number of encoders. 
	There are a set of decoders, which are partitioned into multiple {\it levels}. 
	The reconstructions of the source by decoders within the same level are identical.
	
	Here, we confine our discussion to multilevel diversity coding systems with 
	symmetrical connectivity between the encoders and decoders, referred to as {\it symmetrical multilevel diversity coding} (SMDC) \cite{roche-yeung-hau-1995,yeung97,yeung99}. 
	The SMDC system finds applications in distributed data storage \cite{roche92,roche88isit}, 
	secret sharing \cite{Shamir79,blakley79-tss,blakley85-ramp-secret-sharing}, 
	and robust network communication \cite{robust-1,robust-2}. 
	It is a special case of multi-source network coding \cite{Yeung-Zhang-satellite-99,yeung00-NC,songlihua-03}. 
	This problem can also be regarded as a lossless counterpart of the multiple descriptions problem \cite{gamal82,zhang-berger87,VKG03,gamal11book}. 
	The SMDC coding strategy in turn is used for approximating the multiple descriptions rate region in \cite{tianchao09-app-MDregion-SMDC,guo-yeung2016}. 
	
	In the SMDC problem, there are $L$ ($L\geq 2$) independent discrete memoryless sources $\{ X_l(t): t = 1, 2, \cdots \}$, $l = 1, 2, \cdots, L$, where for each $l$, $X_l(t)$ are independent and identically distributed copies of a generic random variable $X_l$. The importance of the sources is in the order $X_1(t)>X_2(t)>\cdots>X_L(t)$. The sources are encoded by $L$ encoders. There are totally $2^L-1$ decoders, each of which has access to a distinct subset of the encoders. A decoder which can access any $\alpha$ encoders, called a Level $\alpha$ decoder, is required to reconstruct the first $\alpha$ sources. Such a system is called a symmetric $L$-level diversity coding system.
	The system is symmetric in the sense that the problem is unchanged by permuting the $L$ encoders, which is evident from the reconstruction requirements of the decoders.
	
	The SMDC problem was treated for $L=3$ in \cite {yeung97} and in full generality by Yeung and Zhang \cite{yeung99},
	where a coding method called {\em superposition coding} (to be formally defined in Section~II.B) was proved to be optimal.  In this method, 
	the independent sources $\{ X_l(t) \}$, $l = 1, 2, \cdots, L$ are encoded separately.
	Albanese {\it et al.} studied the {\it priority encoding transmission} (PET) problem in \cite{PET-96}, which is almost the same as SMDC. 
	In \cite{PET-96}, they proposed a coding scheme using the same idea as superposition coding 
	and further obtained a sum-rate lower bound which is also given in \cite{yeung97}. 
	The problem has subsequently been generalized in different directions. 
	The {\it secure} communication setting was considered by Balasubramanian {\it et al.} \cite{liutie13-sSMDC} and Jiang {\it et al.} \cite{liutie14-s-all}.
	In \cite{liutie14-s-all}, they also extended the original SMDC setting by introducing an {\it all-access} encoder which is accessible by all the decoders. In both of the above settings, superposition coding is shown to be optimal.
	Xiao {\it et~al.}~\cite{xiaozhiqing-D-MDC15} studied the problem of {\it distributed} multilevel diversity coding where each source is decomposed into $L$ components, each of which is accessed by one distinct encoder. 
	Tian and Liu  \cite{tianchao-liutie-SMDC-Re16} considered the problem with {\it regeneration}, where the storage versus repair-bandwidth tradeoff was investigated. 
	Mohajer {\it et al.} \cite{A-MDC-tianchao} considered the {\it asymmetric} multilevel diversity coding problem and proved that superposition coding is in general suboptimal. 
	Li {\it et~al.}~\cite{Congduan-16} studied the multilevel diversity coding problem with at most 3 sources and 4 encoders in a systematic way 
	and obtained the exact rate region of each of the over 7,000 instances with the aid of computation. 
	In the current paper, we focus on some fundamental issues pertaining to the original SMDC problem discussed in \cite{yeung97,yeung99}.
	
	It was proved in \cite{yeung97} that superposition coding is optimal for $L=3$, and the corresponding coding rate region, 
	referred to as the superposition coding rate region, can be explicitly characterized by 10 linear
	inequalities in the coding rates of the 3 encoders. Thus, the achievability of any given rate triple can be determined by 
	verifying these 10 inequalities. 
	
	However, the optimality proof in \cite{yeung97} is not readily generalizable to a general $L$.  
	Here is an outline of the proof in~\cite{yeung97}. The superposition coding rate region
	is first characterized by the aforementioned 10 inequalities.  This involves the determination of the 
	extreme points of the region.  Then the necessity of 
	these 10 inequalities are established by means of conventional techniques for proving converse
	coding theorems.  
	The difficulty for generalizing the proof to a general $L$ is two-fold:
	\begin{enumerate}
		\item
		It is observed through computation that both the number of linear inequalities 
		needed for characterizing the superposition coding rate region 
		and the number of extreme points of this region grow with $L$.
		As such, it is impossible to 
		determine all of them for a general $L$.
		\item
		For a fixed $L$, once the superposition coding rate region is characterized by a finite set of linear inequalities,
		their necessity needs to be proved.
		With conventional techniques, this needs to be done for each inequality in a way that depends on the 
		coefficients of coding rates.
		It is observed through computation that the number of these inequalities
		grows with $L$.  Therefore,  for a general $L$,
		it is not possible to prove the necessity of all of these inequalities.
	\end{enumerate}
	For a fixed $L$, the extreme points of the superposition coding rate region and 
	the set of linear inequalities characterizing the region can in principle be found by computation.
	However, the complexity grows very quickly with $L$ and becomes prohibitive even for $L=5$.
	On a notebook computer, by using the Fourier-Motzkin elimination algorithm \cite{FME-book}, we were able to 
	compute all the linear inequalities needed for characterizing the superposition coding rate region for $L=4$ in less than 2 minutes.
	However, the computation involved for $L=5$ is already unmanageable.
	
	
	In \cite{yeung99}, the optimality of superposition coding was established for a general $L$ by means of 
	a highly sophisticated method that does not involve any explicit characterization of the coding rate region.
	Instead of a fixed $L$, the problem is tackled for a general $L$.  As $L$ is not fixed, the number of 
	linear inequalities needed for the characterization of the superposition coding rate region may be unbounded.  
	To get around the problem, the coding rate region is characterized
	by an uncountable collection of linear inequalities, where for each inequality, 
	the coefficients associated with the rates are arbitrary nonnegative real numbers with at least one of them
	being nonzero.  The constant terms (i.e., the lower bounds) in these inequalities are given implicitly 
	in terms of the solution of a {\em common} linear optimization problem with the coefficients associated with the rates
	as parameters.  In other words, although the coding rate region 
	is characterized by uncountably many linear inequalities, they have a common form and the necessity of 
	these inequalities can be established in a unified manner.
	
	Although the optimality of superposition coding for a general $L$ has been established in \cite{yeung99},
	this result does not yield an explicit characterization of the coding rate region for any fixed $L$.
	In particular, it does not enable the determination of the achievability of a given rate tuple, even for a fixed $L$,
	for the following two reasons.
	First, the characterization of the coding rate region in \cite{yeung99} involves uncountably many inequalities.
	Second, each inequality in the characterization is implicit, and can be made explicit only by solving a linear
	optimization problem.
	
	In the present paper, we develop fundamental results pertaining to SMDC.
	Our main contributions are summarized as follows:
	\begin{enumerate}
		\item
		We obtain an explicit characterization of the coding rate region for a general $L$. This is done by first solving in closed form 
		the linear optimization problem in \cite{yeung99} that gives an implicit characterization of the coding rate region.
		Then among all the uncountably many inequalities involved in 
		characterizing the coding rate region, we identify a finite subset that is sufficient for characterizing the coding rate region.
		It is further proved that there is no redundancy in this finite set of inequalities.
		Thus for a fixed $L$, the achievability of any given rate tuple can be determined.
		\item
		By taking advantage of the symmetry of the problem, we show that in determining the achievability of 
		a given rate tuple, it suffices to verify a much smaller subset of the set of inequalities identified in 1). Yet, the cardinality of this smaller set of inequalities grows at least exponentially fast
		with $L$. This reveals the inherent complexity of the problem.
		\item
		A subset entropy inequality, which plays a key role in the converse proof in \cite{yeung99}, 
		requires a painstaking and extremely technical proof.  
		We present a weaker version of this subset entropy inequality whose proof is considerably simpler. With our explicit characterization of the coding rate region, this weaker version of the subset entropy inequality is sufficient for proving the optimality of superposition coding.
	\end{enumerate}
	
	The rest of the paper is organized as follows. We first formulate the problem and state some existing results in \Cref{section-formulation}. In \Cref{section-coefficients}, we present a closed-form solution of the linear optimization problem in \cite{yeung99} and establish some basic properties of the solution. In \Cref{section-hyperplane}, we identify a finite set of inequalities that characterizes the superposition coding rate region and show that this set contains no redundancy. In \Cref{section-check-achievability}, we further identify a subset of inequalities we need to verify in determining the achievability of a given rate tuple. We also provide a lower bound and an upper bound on the cardinality of this set. In \Cref{section-subset-inequality}, we present a weaker version of the subset entropy inequality in \cite{yeung99}. We conclude the paper in \Cref{section-conclusion}. Some essential proofs can be found in the appendices.
	
	\section{Problem Formulation and Existing Results}  \label{section-formulation}
	\subsection{Problem Formulation}
	An $L$-level SMDC system, $L\geq2$,  is depicted in \reffig{fig-SMDC-L}.
	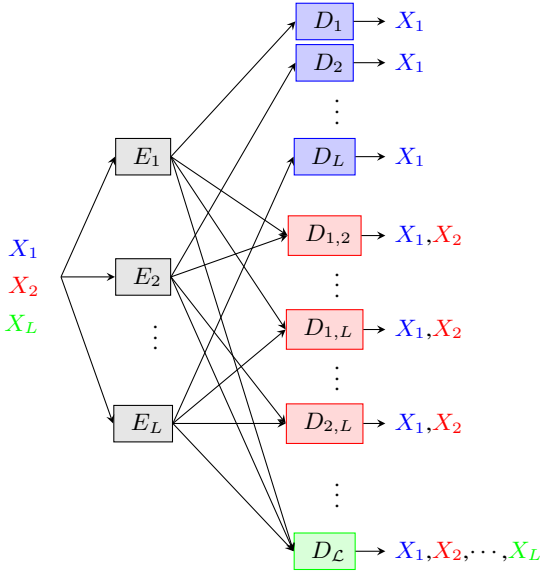
\begin{figure}[!ht]
		\centering
		\begin{tikzpicture}
		\tikzstyle{ann} = [draw=none,fill=none,right] 
		\tikzstyle{every node}=[font=\small]
		\node[left](source1) at (-3.3,0.5) {\color{blue}$X_1$};
		\node[left](source2) at (-3.3,0) {\color{red}$X_2$};
		\node[left](source3) at (-3.3,-0.5) {\color{green}$X_L$};
		
		\matrix(a)[row sep={0.05cm,between borders},column sep=1.5cm] {
			&[-0.8cm,between borders]& \node(decoder1)[draw=blue,rectangle,fill=blue!20] {~$D_1$~};&[-1.2cm,between borders] \node(recover1)[ann]{\color{blue}$X_1$};\\
			&& \node(decoder2)[draw=blue,rectangle,fill=blue!20] {~$D_2$~};& \node(recover2)[ann]{\color{blue}$X_1$};\\[0.5cm,between origins]
			&& \node[ann] {$\vdots$};\\[0.65cm,between origins]
			&\node(encoder1)[draw=black,rectangle,fill=gray!20] {~$E_1$~};& \node(decoderL)[draw=blue,rectangle,fill=blue!20] {~$D_L$~};& \node(recoverL)[ann]{\color{blue}$X_1$};\\[1.0cm,between origins]
			&& \node(decoder12)[draw=red,rectangle,fill=red!15] {~$D_{1,2}$~};& \node(recover12)[ann]{{\color{blue}$X_1$},{\color{red}$X_2$}};\\[0.5cm,between origins]
			\node(source)[draw=none,fill=none] {}; & \node(encoder2)[draw=black,rectangle,fill=gray!20] {~$E_2$~}; & \node[ann] {$\vdots$};\\[0.65cm,between origins]
			& \node[ann] {$\vdots$};& \node(decoder1L)[draw=red,rectangle,fill=red!15] {~$D_{1,L}$~};& \node(recover1L)[ann]{{\color{blue}$X_1$},{\color{red}$X_2$}};\\[0.5cm,between origins]
			&& \node[ann] {$\vdots$};\\[0.65cm,between origins]
			&\node(encoderL)[draw=black,rectangle,fill=gray!20] {~$E_L$~};& \node(decoder2L)[draw=red,rectangle,fill=red!15] {~$D_{2,L}$~};& \node(recover2L)[ann]{{\color{blue}$X_1$},{\color{red}$X_2$}};\\[0.8cm,between origins]
			&& \node[ann] {$\vdots$};\\[0.8cm,between origins]
			&& \node(decodercL)[draw=green,rectangle,fill=green!15] {~$D_{\cL}$~};& \node(recovercL)[ann]{{\color{blue}$X_1$},{\color{red}$X_2$},$\cdots$,{\color{green}$X_L$}};\\
		};
		\path [-stealth] (source.east) edge node [near end,above] {} (encoder1.west)
		(source.east) edge node [near end,sloped,above] {} (encoder2)
		(source.east) edge node [near end,above] {} (encoderL.west)
		(encoder1.east) edge node {}(decoder1.west)
		(encoder1.east) edge node {}(decoder12.west)
		(encoder1.east) edge node {}(decoder1L.west)
		(encoder1.east) edge node {}(decodercL.west)
		(encoder2.east) edge node {}(decoder2.west)
		(encoder2.east) edge node {}(decoder12.west)
		(encoder2.east) edge node {}(decoder2L.west)
		(encoder2.east) edge node {}(decodercL.west)
		(encoderL.east) edge node {}(decoderL.west)
		(encoderL.east) edge node {}(decoder1L.west)
		(encoderL.east) edge node {}(decoder2L.west)
		(encoderL.east) edge node {}(decodercL.west)
		(decoder1.east) edge node {}(recover1.west)
		(decoder2.east) edge node {}(recover2.west)
		(decoderL.east) edge node {}(recoverL.west)
		(decoder12.east) edge node {}(recover12.west)
		(decoder1L.east) edge node {}(recover1L.west)
		(decoder2L.east) edge node {}(recover2L.west)
		(decodercL.east) edge node {}(recovercL.west);
		\end{tikzpicture}\centering
		\caption{The symmetrical multilevel diversity coding system.}\label{fig-SMDC-L}
	\end{figure}
	The problem is defined as follows. Let $\cL=\{1,2,\cdots,L\}$. 
	Let $t$ be the time index and $\left\{\big(X_1(t), X_2(t),\cdots, X_L(t)\big): t=1,2,\cdots\right\}$ be 
	a collection of $L$ independent discrete memoryless information sources with an $L$-tuple of generic random variables $(X_1$, $X_2$, $\cdots,X_L)$ 
	taking values in $\cX_1\times\cX_2\times\cdots\times\cX_L$, where $\cX_i,i\in\cL$ are finite alphabets. 
	There are $L$ encoders, indexed by $\cL$, each of which can access all the $L$ information sources. 
	There are also $2^L-1$ decoders. For each $\cU\subseteq\cL$ such that $\cU\neq \emptyset$, 
	Decoder-$\cU$ can access the subset of encoders indexed by $\cU$. 
	Without loss of generality, assume the elements in $\cU$ are in an ascending order. 
	For $1\leq \alpha\leq L$ and $\cU$ such that $|\cU|=\alpha$, Decoder-$\cU$ can reconstruct 
	the first $\alpha$ sources $\{X_1(t), X_2(t), \cdots, X_{\alpha}(t)\}$ perfectly asymptotically, which will be defined later. 
	
	An $(n,M_1,M_2,\cdots,M_L)$ code is defined by the encoding functions
	\begin{equation*}
	E_l:\prod_{i=1}^L\cX_i^n\rightarrow\{1,2,\cdots,M_l\}, \text{ for }l\in\cL 
	\end{equation*}
	and decoding functions
	\begin{equation*}
	D_{\cU}:\prod_{l\in\cU}\{1,2,\cdots,M_l\}\rightarrow\prod_{i=1}^{|\cU|}\cX_i^n,\text{ for }\cU\subseteq\cL\text{ and }\cU\neq\emptyset. 
	\end{equation*}
	For $1\leq \alpha\leq L$, let $\bm{X}_{\alpha}=(X_{\alpha}(1),X_{\alpha}(2),\cdots,X_{\alpha}(n))$. Let $W_l=E_l(\bm{X}_1,\bm{X}_2,\cdots,\bm{X}_L)$ be the output of Encoder-$l$
	 and $W_{\cU}=(W_i:i\in\cU)$ for $\cU\subseteq\cL$.\footnote{Here $E_l(\bm{X}_1,\bm{X}_2,\cdots,\bm{X}_L)$ is a function of random vectors and hence $W_l$ is a random variable.  The reader should not confuse $E_l$ with the expectation of a random variable.} A nonnegative rate tuple $(R_1,R_2,\cdots,R_L)$ is \textit{achievable} if for any $\epsilon>0$, there exists for sufficiently large $n$ an $(n,M_1,M_2,\cdots,M_L)$ code such that
	\begin{equation*}
	\frac{1}{n}\log M_l\leq R_l+\epsilon,\forall~ l\in\cL, 
	\end{equation*}
	and
	\begin{equation*}
	\Pr\{D_{\cU}(W_{\cU})\neq (\bm{X}_1,\bm{X}_2,\cdots,\bm{X}_{\alpha})\}\leq \epsilon,   \label{recover constraint}
	\end{equation*}
	for all $\alpha=1,2,\cdots,L$ and $\cU\subseteq\cL$ such that $|\cU|=\alpha$. The achievable rate region $\cR$ is defined as the collection of all achievable rate tuples.
	
	\subsection{Existing Results}
	We adopt the terminologies and notations in \cite{yeung99}. Let $\cR_{\text{sup}}$ be the rate region induced by superposition coding. Then $\cR_{\text{sup}}$ is the set of nonnegative rate tuples $\bm{R}=(R_1,R_2,\cdots,R_L)$ such that 
	\begin{eqnarray}
	R_l=\sum_{\alpha=1}^{L}r_l^{\alpha}, \text{ for }l\in \cL  \label{def-R-sup1}
	\end{eqnarray}
	for some $r_l^{\alpha}\geq 0,~1\leq \alpha\leq L$, satisfying
	\begin{equation}
	\sum_{l\in\cU}r_l^{\alpha}\geq H(X_{\alpha}),\text{ for all }\cU\subseteq\cL \text{ and }|\cU|=\alpha.  \label{def-R-sup2}
	\end{equation}
	For an elaborative discussion on superposition coding for the 3-level SMDC system, we refer the reader to \cite{yeung97}.
	
	For a fixed $L$, based on \eqref{def-R-sup1} and \eqref{def-R-sup2}, one can apply the Fourier-Motzkin algorithm to eliminate $r_l^{\alpha}$ for $l,\alpha\in\cL$. The output is a set of linear inequalities involving $R_l,l\in\cL$ that gives an explicit characterization of $\cR_{\sup}$. However, as mentioned in \Cref{section-introduction}, the computation involved for $L\geq 5$ is unmanageable.
	
	Let $\bm{\lambda}=(\lambda_1,\lambda_2,\cdots,\lambda_L)$ and
	\begin{equation}
	\mathbb{R}_+^L=\{\bm{\lambda}:~\bm{\lambda}\neq\bm{0} \text{ and } \lambda_i\in\mathbb{R},\lambda_i\geq 0\text{ for }i\in\cL\}.   \label{definition-R-set}
	\end{equation}
	Let $\Omega_L^{\alpha}=\left\{\bm{v}\in\{0,1\}^L:|\bm{v}|=\alpha\right\}$, where $|\bm{v}|$ is the Hamming weight of a vector $\bm{v}=(v_1,v_2,\cdots,v_L)$. 
	Note that there is a one-to-one correspondence between a vector $\bm{v}\in\{0,1\}^L$ and Decoder-$\cU$, where $\cU=\{i:v_i=1\}$. 
	For any $\bm{v}\in\Omega_L^{\alpha}$, let $c_\alpha(\bm{v})$ be any nonnegative real number. 
	For any $\bm{\lambda}\in\mathbb{R}_+^L$ and $\alpha\in\cL$, let $f_{\alpha}(\bm{\lambda})$ be the optimal solution to the following optimization problem:
	\begin{eqnarray}
	f_{\alpha}(\bm{\lambda})\triangleq&\max&\sum_{\bm{v}\in\Omega_L^{\alpha}}c_{\alpha}(\bm{v})  \label{optimization-1} \\
	&\text{s.t.}&\sum_{\bm{v}\in\Omega_L^{\alpha}}c_{\alpha}(\bm{v}) \bm{v}\leq \bm{\lambda}   \label{optimization-2}\\
	&&c_{\alpha}(\bm{v})\geq 0,\forall \bm{v}\in\Omega_L^{\alpha}.  \label{optimization-3}
	\end{eqnarray}
	Note that the functions $f_{\alpha}(\cdot)$ and $c_{\alpha}(\cdot)$ above depend on $L$, but for simplicity we omit this dependency in the notations. Thus, if the length of $\bm{\lambda}$ is given, then $f_{\alpha}(\bm{\lambda})$ can be defined accordingly. A set $\left\{c_{\alpha}(\bm{v}): \bm{v}\in\Omega_L^{\alpha}\right\}$ is called an $\alpha$-\textit{resolution} for $\bm{\lambda}$ if \eqref{optimization-2} and \eqref{optimization-3} are satisfied and it will be abbreviated as $\{c_{\alpha}(\bm{v})\}$ if there is no ambiguity. Furthermore, an $\alpha$-resolution is called \textit{optimal} if it achieves the optimal value $f_{\alpha}(\bm{\lambda})$. 
	
	\begin{remark}
		Here is an intuitive explanation of $f_\alpha(\bm{\lambda})$: Consider a set of items from $L$ different types indexed by $\cL$, 
		where the number of items of type $i$ ($i\in\cL$) is $\lambda_i$. 
		An $\alpha$-type group is defined as a group of $\alpha$ items of different types. 
		The goal is to cluster the items into $\alpha$-type groups so that the total number of such groups is maximized. 
		This maximum is defined as $f_\alpha(\bm{\lambda})$. 
	\end{remark}
	
	Let $\cR_{\text{h}}$ be the collection of nonnegative rate tuples $\bm{R}$ such that
	\begin{equation}
	\sum_{l=1}^L\lambda_lR_l\geq \sum_{\alpha=1}^{L}f_{\alpha}(\bm{\lambda})H(X_{\alpha}), \text{ for all } \bm{\lambda}\in\mathbb{R}_+^L.  \label{rate-contraints}
	\end{equation}
	It was proved in \cite{yeung99} that the superposition region $\cR_{\text{sup}}$ can be alternatively characterized by $\cR_{\text{h}}$. This means that in addition to being the optimal value of the optimization problem in \eqref{optimization-1}, for every fixed $\bm{\lambda}\in\mathbb{R}_+^L$, $f_{\alpha}(\bm{\lambda})$ also gives a tightest possible linear outer bound on $\cR_{\sup}$ via \eqref{rate-contraints}. It was further proved in \cite{yeung99} that $\cR_{\text{h}}$ is an outer bound on $\cR$. Then
	\begin{equation*}
	\cR_{\sup}\subseteq\cR\subseteq\cR_{\text{h}}
	\end{equation*}
	which implies
	\begin{equation}
	\cR=\cR_{\text{h}}=\cR_{\sup},  \label{Rsup=Rh}
	\end{equation}
	i.e., superposition coding is optimal.
	
	The following lemma is a direct consequence of Lemma 4 and 7 in \cite{yeung99}. It will be used in the proof of our main result in the next section.
	\begin{lemma}\label{lemma-perfect-resolution}
		Assume $\lambda_1\geq \lambda_2\geq \cdots\geq \lambda_L$. For $\alpha\geq 2$, if $\lambda_1\leq \frac{\lambda_2+\lambda_3+\cdots+\lambda_L}{\alpha-1}$, then $f_{\alpha}(\bm{\lambda})=\frac{1}{\alpha}\sum_{i=1}^L\lambda_i.$
	\end{lemma}
	
	\section{Optimal $\alpha$-resolution}  \label{section-coefficients}
	For any $\bm{\lambda}\in\mathbb{R}^L_+$ and any permutation $\omega$ on $\{1,2,\cdots,L\}$, with an abuse of notation, we denote $\left(\lambda_{\omega(1)}\right.,$ $\lambda_{\omega(2)},\cdots,\left.\lambda_{\omega(L)}\right)$ by $\omega(\bm{\lambda})$. For each $\alpha\in\cL$, due to the symmetry of the system, it is intuitive that the values of $f_{\alpha}(\omega(\bm{\lambda}))$ are the same for all $\omega$. This important property of $f_{\alpha}(\bm{\lambda})$ is formally proved in the following lemma.
	\begin{lemma}\label{lemma-indep-order}
		$f_{\alpha}\big(\omega(\bm{\lambda})\big)=f_{\alpha}(\bm{\lambda})$ for any $\alpha\in\cL$.
	\end{lemma}
	\begin{proof}
		For any $\alpha\in\cL$, let $\{c_{\alpha}(\bm{v}):\bm{v}\in\Omega_L^{\alpha}\}$ be an optimal $\alpha$-resolution for $\bm{\lambda}$. Then we have 
		\begin{eqnarray}
		\sum_{\bm{v}\in\Omega_L^{\alpha}}c_{\alpha}(\bm{v}) \bm{v}&\leq& \bm{\lambda},   \label{pf-lemma-indep-order-optimization-1}\\
		c_{\alpha}(\bm{v})&\geq& 0,\forall \bm{v}\in\Omega_L^{\alpha}, \nonumber 
		\end{eqnarray}
		and
		\begin{equation*}
		f_{\alpha}(\bm{\lambda})=\sum_{\bm{v}\in\Omega_L^{\alpha}}c_{\alpha}(\bm{v}). 
		\end{equation*}
		Let $\sum_{\bm{v}\in\Omega_L^{\alpha}}c_{\alpha}(\bm{v}) \bm{v}=\tilde{\bm{\lambda}}$. Then by \eqref{pf-lemma-indep-order-optimization-1}, we have $\tilde{\bm{\lambda}}\leq \bm{\lambda}$. For any permutation $\omega$ on $\{1,2,\cdots,L\}$, we can check that 
		\begin{equation}
		\sum_{\bm{v}\in\Omega_L^{\alpha}}c_{\alpha}(\bm{v})~\omega(\bm{v})=\omega\bigg(\sum_{\bm{v}\in\Omega_L^{\alpha}}c_{\alpha}(\bm{v})\bm{v}\bigg)=\omega(\tilde{\bm{\lambda}})\leq \omega(\bm{\lambda}).   \label{pf-lemma-indep-order-sum-w(v)}
		\end{equation}
		For any $\bm{v}\in\Omega_L^{\alpha}$, let 
		\begin{equation*}
		c_{\alpha}'\big(\omega(\bm{v})\big)=c_{\alpha}(\bm{v}).
		\end{equation*}
		It is immediate that for all $\bm{v}\in\Omega_L^{\alpha}$,
		\begin{equation}
		c_{\alpha}'\big(\omega(\bm{v})\big)\geq 0.   \label{pf-lemma-indep-order-check-c'-1}
		\end{equation}
		Since $\omega$ is a one-to-one mapping from $\Omega_L^{\alpha}$ to $\Omega_L^{\alpha}$, we have $\bm{v}\in\Omega_L^{\alpha}$ if and only if $\omega(\bm{v})\in\Omega_L^{\alpha}$ for any $\omega$. Thus,
		\begin{eqnarray}
		\sum_{\omega(\bm{v})\in\Omega_L^{\alpha}}c_{\alpha}'\big(\omega(\bm{v})\big)~\omega(\bm{v})&=&\sum_{\bm{v}\in\Omega_L^{\alpha}}c_{\alpha}'\big(\omega(\bm{v})\big)~\omega(\bm{v})  \nonumber \\
		&=&\sum_{\bm{v}\in\Omega_L^{\alpha}}c_{\alpha}(\bm{v})~\omega(\bm{v})  \nonumber \\
		&\leq& \omega\big(\bm{\lambda}\big),   \label{pf-lemma-indep-order-check-c'-2}
		\end{eqnarray}
		where the inequality follows from \eqref{pf-lemma-indep-order-sum-w(v)}. By \eqref{pf-lemma-indep-order-check-c'-1} and \eqref{pf-lemma-indep-order-check-c'-2}, we see that $\{c_{\alpha}'\big(\omega(\bm{v})\big):\bm{v}\in\Omega_L^{\alpha}\}$ is an $\alpha$- resolution for $\omega(\bm{\lambda})$. In light of the definition of $f_{\alpha}(\bm{\lambda})$ in \eqref{optimization-1}-\eqref{optimization-3}, we have 
		\begin{eqnarray}
		f_{\alpha}\big(\omega(\bm{\lambda})\big)&\geq& \sum_{\omega(\bm{v})\in\Omega_L^{\alpha}}c_{\alpha}'\big(\omega(\bm{v})\big)  \nonumber \\
		&=&\sum_{\bm{v}\in\Omega_L^{\alpha}}c_{\alpha}'\big(\omega(\bm{v})\big)  \nonumber \\ 
		&=&\sum_{\bm{v}\in\Omega_L^{\alpha}}c_{\alpha}(\bm{v})  \nonumber \\
		&=&f_{\alpha}(\bm{\lambda}), \nonumber
		\end{eqnarray}
		and so 
		\begin{equation}
		f_{\alpha}\big(\omega(\bm{\lambda})\big)\geq f_{\alpha}(\bm{\lambda}).   \label{pf-indep-order-1}
		\end{equation}
		
		Let $\omega^{-1}$ be the inverse permutation of $\omega$. By the same argument, we can obtain 
		\begin{equation}
		f_{\alpha}\big(\omega^{-1}\big(\omega(\bm{\lambda})\big)\big)\geq f_{\alpha}\big(\omega(\bm{\lambda})\big).  \label{pf-indep-order-2}
		\end{equation}
		Since $\omega^{-1}\left(\omega(\bm{\lambda})\right)=\bm{\lambda}$, we see from \eqref{pf-indep-order-1} and \eqref{pf-indep-order-2} that
		\begin{equation*}
		f_{\alpha}(\bm{\lambda})=f_{\alpha}\big(\omega(\bm{\lambda})\big)~~\text{ for all }\alpha\in\cL.
		\end{equation*}
		The lemma is proved.
	\end{proof}
	
	\medskip
	If a vector $\bm{\lambda}$ satisfies
	\begin{equation}
	\lambda_1\geq \lambda_2\geq \cdots\geq \lambda_L,  \label{assumption-lambda-1}
	\end{equation}
	we call $\bm{\lambda}$ an \textit{ordered vector}. Throughout this section, except for \Cref{lemma-concave-f}, in light of \Cref{lemma-indep-order}, we assume without loss of generality that $\bm{\lambda}$ is an ordered vector. For any $\alpha\in\cL$, it is easy to see that 
	\begin{equation}
	f_{\alpha}(\mu\bm{\lambda})=\mu f_{\alpha}(\bm{\lambda})  \label{scale-f-lambda}
	\end{equation}
	for all $\mu\in\mathbb{R}$ such that $\mu>0$. 
	In view of \eqref{rate-contraints} and \eqref{scale-f-lambda}, we will consider only $\bm{\lambda}$'s whose minimum nonzero element is equal to 1. Then there exists a $\zeta\in\cL$ such that 
	\begin{equation*}
	\lambda_1\geq \lambda_2\geq \cdots\geq \lambda_{\zeta}=1
	\end{equation*}
	and $\lambda_i=0$ for all $i=\zeta+1,\zeta+2,\cdots,L$.
	
	Fix $\bm{\lambda}$, it is easy to see that 
	\begin{equation}
	f_1(\bm{\lambda})=\sum_{i=1}^L\lambda_i,  \label{f(1)}
	\end{equation}
	and
	\begin{equation*}
	f_{\zeta}(\bm{\lambda})=1,
	\end{equation*}
	and for $\alpha\geq \zeta+1$,
	\begin{equation*}
	f_{\alpha}(\bm{\lambda})=0.
	\end{equation*}
	For other cases, determining the value of $f_{\alpha}(\bm{\lambda})$ is highly nontrivial.
	
	For $\alpha\in\cL$ and $\beta=0,1,\cdots,\alpha-1$, let
	\begin{equation}
	g_{\alpha,\bm{\lambda}}(\beta)=\frac{1}{\alpha-\beta}\sum_{i=\beta+1}^L\lambda_i.  \label{definition-g}
	\end{equation}
	Let $\beta^*_{\alpha}$ be a value of $\beta$ (not necessarily unique) that achieves the minimum $\min_{\beta\in\{0,1,\cdots,\alpha-1\}}g_{\alpha,\bm{\lambda}}(\beta)$, i.e., 
	\begin{equation}
	g_{\alpha,\bm{\lambda}}(\beta^*_{\alpha})=\min_{\beta\in\{0,1,\cdots,\alpha-1\}}g_{\alpha,\bm{\lambda}}(\beta).  \label{definition-beta*}
	\end{equation}
	The following theorem, a main result of the current paper, gives a closed-form solution for $f_{\alpha}(\bm{\lambda})$.
	
	\begin{theorem}\label{thm-opt-f}
		For any $\alpha\in\cL$, $f_{\alpha}(\bm{\lambda})=g_{\alpha,\bm{\lambda}}(\beta^*_{\alpha})$.
	\end{theorem}
	
	\begin{proof}
		Fix an $\alpha\in\cL$, and denote $\beta^*_{\alpha}$ by $\beta^*$ for simplicity. We prove the theorem by proving i) $f_{\alpha}(\bm{\lambda})\leq~g_{\alpha,\bm{\lambda}}(\beta^*)$; ii) there exists a solution for the optimization problem \eqref{optimization-1} that achieves $g_{\alpha,\bm{\lambda}}(\beta^*)$, so that $f_{\alpha}(\bm{\lambda})\geq~g_{\alpha,\bm{\lambda}}(\beta^*)$.
		
		\begin{enumerate}[i)]
			\item $f_{\alpha}(\bm{\lambda})\leq g_{\alpha,\bm{\lambda}}(\beta^*)$.
			
			For $0\leq \beta \leq\alpha-1$, let $\bm{e}_{\beta}$ be an $L$-vector with the first $\beta$ components being 0 and the last $L-\beta$ components being 1. For any $\bm{v}\in\Omega_L^{\alpha}$, since $\sum_{i=1}^{\beta}v_i\leq \beta$, we have 
			\begin{equation*}
			\bm{v}\cdot\bm{e}_{\beta}=\sum_{i=\beta+1}^Lv_i\geq \alpha-\beta.
			\end{equation*}
			Then for any solutions $\{c_{\alpha}(\bm{v})\}$ to the optimization problem in \eqref{optimization-1}, we have 
			\begin{eqnarray}
			\sum_{i=\beta+1}^L\lambda_i&=&\bm{\lambda}\cdot\bm{e}_{\beta}  \nonumber \\
			&\geq &\left(\sum_{\bm{v}\in\Omega_L^{\alpha}}c_{\alpha}(\bm{v})\bm{v}\right)\cdot\bm{e}_{\beta}  \nonumber \\
			&=&\sum_{\bm{v}\in\Omega_L^{\alpha}}c_{\alpha}(\bm{v})(\bm{v}\cdot\bm{e}_{\beta})  \nonumber \\
			&\geq &\sum_{\bm{v}\in\Omega_L^{\alpha}}c_{\alpha}(\bm{v})(\alpha-\beta)  \nonumber \\
			&=&(\alpha-\beta)\sum_{\bm{v}\in\Omega_L^{\alpha}}c_{\alpha}(\bm{v}). \nonumber 
			\end{eqnarray}
			This implies that for all $0\leq \beta\leq \alpha-1$, 
			\begin{equation*}
			f_{\alpha}(\bm{\lambda})\leq \frac{1}{\alpha-\beta}\sum_{i=\beta+1}^L\lambda_i=g_{\alpha,\bm{\lambda}}(\beta).
			\end{equation*}
			Thus, we have
			\begin{equation*}
			f_{\alpha}(\bm{\lambda})\leq g_{\alpha,\bm{\lambda}}(\beta^*).  
			\end{equation*}
			
			\item $f_{\alpha}(\bm{\lambda})\geq g_{\alpha,\bm{\lambda}}(\beta^*)$.
			
			We now show that there exists a solution that achieves $g_{\alpha,\bm{\lambda}}(\beta^*)$. For any $\alpha\in\cL$ and $\beta^*\in\{0,1,\cdots,\alpha-2\}$, by \eqref{definition-beta*}, we have
			\begin{equation*}
			\frac{1}{\alpha-\beta^*}\sum_{i=\beta^*+1}^L\lambda_i\leq \frac{1}{\alpha-(\beta^*+1)}\sum_{i=\beta^*+2}^L\lambda_i,
			\end{equation*}
			which is equivalent to 
			\begin{equation}
			\lambda_{\beta^*+1}\leq \frac{1}{(\alpha-\beta^*)-1}\sum_{i=\beta^*+2}^L\lambda_i.  \label{pf-thm-opt-f-lambda}
			\end{equation}
			Denote the $(L-\beta^*)$-vector $(\lambda_{\beta^*+1}, \lambda_{\beta^*+2}, \cdots, \lambda_{L})$ by $\bm{\lambda}'$. In view of \eqref{pf-thm-opt-f-lambda}, by \Cref{lemma-perfect-resolution}, \eqref{definition-g}, and \eqref{definition-beta*}, we have
			\begin{equation}
			f_{\alpha-\beta^*}(\bm{\lambda}')=\frac{1}{\alpha-\beta^*}\sum_{i=\beta^*+1}^L\lambda_i=g_{\alpha,\bm{\lambda}}(\beta^*).  \label{pf-thm-opt-f(alpha-beta)}
			\end{equation}
			In view of \eqref{f(1)} and \eqref{definition-g} with $\beta=\beta^*$, it is easy to check that \eqref{pf-thm-opt-f(alpha-beta)} is also satisfied for $\beta^*=\alpha-1$.
			Without loss of generality, let $\left\{c_{\alpha-\beta^*}(\bm{u}):\bm{u}\in\Omega_{L-\beta^*}^{\alpha-\beta^*}\right\}$ be an optimal $(\alpha-\beta^*)$-resolution for $\bm{\lambda}'$. Then it follows from \eqref{pf-thm-opt-f(alpha-beta)} that
			\begin{equation*}
			\sum_{\bm{u}\in\Omega_{L-\beta^*}^{\alpha-\beta^*}}c_{\alpha-\beta^*}(\bm{u})=f_{\alpha-\beta^*}(\bm{\lambda}')=g_{\alpha,\bm{\lambda}}(\beta^*).
			\end{equation*}
			For any $\bm{v}\in\Omega_L^{\alpha}$, let
			\begin{equation}
			c_{\alpha}(\bm{v})=
			\begin{cases}
			c_{\alpha-\beta^*}(\bm{u}),&\text{ if }\bm{v}=(11\cdots1\bm{u}) \\
			&\quad \text{ for some }\bm{u}\in\Omega_{L-\beta^*}^{\alpha-\beta^*}  \\
			0,&\text{ otherwise.}
			\end{cases}  \label{resolution-u}
			\end{equation}
			Then we have 
			\begin{equation}
			\sum_{\bm{v}\in\Omega_{L}^{\alpha}}c_{\alpha}(\bm{v})=\sum_{\bm{u}\in\Omega_{L-\beta^*}^{\alpha-\beta^*}}c_{\alpha-\beta^*}(\bm{u})=g_{\alpha,\bm{\lambda}}(\beta^*).  \label{pf-thm-opt-f-u-optimal}
			\end{equation}
			Again, by \eqref{definition-beta*}, we have
			\begin{equation*}
			\frac{1}{\alpha-(\beta^*-1)}\sum_{i=\beta^*}^L\lambda_i \geq \frac{1}{\alpha-\beta^*}\sum_{i=\beta^*+1}^L\lambda_i.
			\end{equation*}
			Then
			\begin{equation*}
			\lambda_{\beta^*} \geq \frac{1}{\alpha-\beta^*}\sum_{i=\beta^*+1}^L\lambda_i=g_{\alpha,\bm{\lambda}}(\beta^*),
			\end{equation*}
			where the equality above follows from \eqref{definition-g}. Thus,
			\begin{equation*}
			\lambda_1\geq \lambda_2\geq \cdots\geq \lambda_{\beta^*} \geq g_{\alpha,\bm{\lambda}}(\beta^*).
			\end{equation*}
			For $i=1,2,\cdots,\beta^*$, since $c_{\alpha}(\bm{v})=0$ if $v_i=0$, we have
			\begin{equation}
			\sum_{\bm{v}\in\Omega_{L}^{\alpha}:v_i=1}c_{\alpha}(\bm{v})=\sum_{\bm{v}\in\Omega_{L}^{\alpha}}c_{\alpha}(\bm{v})=g_{\alpha,\bm{\lambda}}(\beta^*)\leq \lambda_i, \label{pf-thm-opt-f-u-resolution-1}
			\end{equation}
			where the second equality follows from \eqref{pf-thm-opt-f-u-optimal}. For $i=\beta^*+1,\beta^*+2,\cdots,L$,
			\begin{align}
			\sum_{\bm{v}\in\Omega_{L}^{\alpha}:v_i=1}c_{\alpha}(\bm{v})&=\hspace{-0.2cm}\mathop{\sum_{\bm{v}\in\Omega_{L}^{\alpha}:~ v_i=1,}}_{(v_1,\cdots,v_{\beta^*})=\bm{1}}c_{\alpha}(\bm{v})+ \hspace{-0.2cm}\mathop{\sum_{\bm{v}\in\Omega_{L}^{\alpha}:~ v_i=1,}}_{(v_1,\cdots,v_{\beta^*})\neq\bm{1}}c_{\alpha}(\bm{v})  \nonumber \\
			&=\sum_{\bm{u}\in\Omega_{L-\beta^*}^{\alpha-\beta^*}:~u_{i-\beta^*}=1}c_{\alpha-\beta^*}(\bm{u})+0  \nonumber \\
			&\leq \lambda_i,  \label{u-resolution-2}
			\end{align}
			since $\left\{c_{\alpha-\beta^*}(\bm{u}):\bm{u}\in\Omega_{L-\beta^*}^{\alpha-\beta^*}\right\}$ is an optimal $(\alpha-\beta^*)$-resolution for $\bm{\lambda}'$. From \eqref{pf-thm-opt-f-u-optimal}, \eqref{pf-thm-opt-f-u-resolution-1}, and \eqref{u-resolution-2}, we can see that $\{c_{\alpha}(\bm{v}):\bm{v}\in\Omega_L^{\alpha}\}$ defined by \eqref{resolution-u} is an $\alpha$-resolution for $\bm{\lambda}$ that achieves $g_{\alpha,\bm{\lambda}}(\beta^*)$. Thus, we have 
			\begin{equation*}
			f_{\alpha}(\bm{\lambda})\geq g_{\alpha,\bm{\lambda}}(\beta^*).  
			\end{equation*}
		\end{enumerate}
	\end{proof}
	
	The following lemma provides an important insight into the minimum in \eqref{definition-beta*}.
	\begin{lemma}\label{lemma-relation-beta}
		For any $\alpha\in\{2,3,\cdots,L\}$ and $0\leq \beta\leq \alpha-2$,
		\begin{enumerate}[(i)]
			\item if $g_{\alpha,\bm{\lambda}}(\beta)\geq g_{\alpha,\bm{\lambda}}(\beta+1)$, then 
			\begin{equation*}
			g_{\alpha,\bm{\lambda}}(0)\geq g_{\alpha,\bm{\lambda}}(1)\geq \cdots\geq g_{\alpha,\bm{\lambda}}(\beta+1);
			\end{equation*}
			\item if $g_{\alpha,\bm{\lambda}}(\beta)\leq g_{\alpha,\bm{\lambda}}(\beta+1)$, then 
			\begin{equation*}
			g_{\alpha,\bm{\lambda}}(\beta)\leq g_{\alpha,\bm{\lambda}}(\beta+1)\leq \cdots\leq g_{\alpha,\bm{\lambda}}(\alpha-1).
			\end{equation*}
		\end{enumerate}
	\end{lemma}
	\begin{remark}
		In \Cref{lemma-relation-beta}, if all the non-strict inequalities are replaced by strict inequalities, the lemma remains valid. This alternative version of \Cref{lemma-relation-beta} can be proved by modifying the proof below accordingly.
	\end{remark}
	\begin{remark}
		\Cref{lemma-relation-beta} reveals the {\it pseudo-convexity} \cite{pseudo-convex} of the function $g_{\alpha,\bm{\lambda}}(\beta)$.
	\end{remark}
	\begin{proof}[Proof of \Cref{lemma-relation-beta}]
		In the following, we only prove (ii). The proof for (i) can be obtained similarly.
		
		For $\alpha=2$, the lemma is immediate. For $3\leq \alpha\leq L$ and $\beta=\alpha-2$, (ii) is immediate. For $0\leq \beta\leq \alpha-3$, from the definition of $g_{\alpha,\bm{\lambda}}(\cdot)$ in~\eqref{definition-g}, the condition $g_{\alpha,\bm{\lambda}}(\beta)\leq g_{\alpha,\bm{\lambda}}(\beta+1)$ is equivalent to 
		\begin{equation*}
		\frac{1}{\alpha-\beta}\sum_{i=\beta+1}^L\lambda_i\leq \frac{1}{\alpha-(\beta+1)}\sum_{i=\beta+2}^L\lambda_i,  
		\end{equation*}
		or
		\begin{equation*}
		\lambda_{\beta+1}\leq \frac{1}{\alpha-(\beta+1)}\sum_{i=\beta+2}^L\lambda_i.
		\end{equation*}
		Then by the assumption in \eqref{assumption-lambda-1}, we have
		\begin{equation*}
		\lambda_{\beta+2}\leq \frac{1}{\alpha-(\beta+1)}\sum_{i=\beta+2}^L\lambda_i,  
		\end{equation*}
		or
		\begin{equation*}
		\lambda_{\beta+2}\leq \frac{1}{\alpha-(\beta+2)}\sum_{i=\beta+3}^L\lambda_i,
		\end{equation*}
		which is also equivalent to
		\begin{equation*}
		\frac{1}{\alpha-(\beta+1)}\sum_{i=\beta+2}^L\lambda_i\leq \frac{1}{\alpha-(\beta+2)}\sum_{i=\beta+3}^L\lambda_i.  
		\end{equation*}
		From \eqref{definition-g}, we have 
		\begin{equation*}
		g_{\alpha,\bm{\lambda}}(\beta+1)\leq g_{\alpha,\bm{\lambda}}(\beta+2).  
		\end{equation*}
		Then we see inductively that for all $\beta+1\leq \beta'\leq \alpha-2$,
		\begin{equation*}
		g_{\alpha,\bm{\lambda}}(\beta')\leq g_{\alpha,\bm{\lambda}}(\beta'+1).
		\end{equation*}
		The lemma is proved.
	\end{proof}
	
	For any $\alpha\in\{2,3,\cdots,L\}$ and any $\beta\in\{0,1,\cdots,\alpha-1\}$, we can readily see from \Cref{lemma-relation-beta} that $\beta^*_{\alpha}=\beta$ if and only if
	\begin{equation*}
	g_{\alpha, \bm{\lambda}}(0)\geq g_{\alpha,\bm{\lambda}}(1)\geq \cdots\geq g_{\alpha,\bm{\lambda}}(\beta)
	\end{equation*}
	and
	\begin{equation*}
	g_{\alpha, \bm{\lambda}}(\beta)\leq g_{\alpha,\bm{\lambda}}(\beta+1)\leq \cdots\leq g_{\alpha,\bm{\lambda}}(\alpha-1).
	\end{equation*}
	This provides a method to find the optimal value $\beta^*_{\alpha}$ conveniently. We only need to compare $g_{\alpha,\bm{\lambda}}(\beta)$ and $g_{\alpha,\bm{\lambda}}(\beta+1)$ for $\beta=0,1,\cdots,\alpha-2$ successively and stop at the first $\beta$ such that $g_{\alpha,\bm{\lambda}}(\beta)\leq g_{\alpha,\bm{\lambda}}(\beta+1)$. Then this $\beta$ gives a value of $\beta_{\alpha}^*$ that achieves the minimum in \eqref{definition-beta*}.
	
	\begin{lemma}\label{lemma-relation-alpha}
		$0=\beta_1^*\leq \beta_2^*\leq \cdots\leq \beta_L^*$.
	\end{lemma}
	\begin{proof}
		It is easy to see from \eqref{f(1)} that $\beta^*_1=0$. This implies that $\beta^*_1\leq \beta^*_2$. Now, we prove the lemma by showing that $\beta^*_{\alpha-1}\leq \beta^*_{\alpha}$ for any $3\leq \alpha\leq L$. If $\beta^*_{\alpha}\in\{\alpha-2,~\alpha-1\}$, since for a fixed $\alpha\in\cL$ we have $0\leq \beta\leq \alpha-1$, it is obvious that
		\begin{equation*}
		\beta^*_{\alpha-1}\leq \alpha-2\leq \beta^*_{\alpha}.
		\end{equation*}
		Otherwise, $\beta^*_{\alpha}\in\{0,1,\cdots,\alpha-3\}$. Since $\beta_{\alpha}^*$ achieves the minimum in \eqref{definition-beta*}, we have
		\begin{equation*}
		\frac{1}{\alpha-\beta^*_{\alpha}}\sum_{i=\beta^*_{\alpha}+1}^L\lambda_i\leq \frac{1}{\alpha-(\beta^*_{\alpha}+1)}\sum_{i=\beta^*_{\alpha}+2}^L\lambda_i,
		\end{equation*}
		which is equivalent to 
		\begin{equation*}
		\lambda_{\beta^*_{\alpha}+1}\leq \frac{1}{\alpha-(\beta^*_{\alpha}+1)}\sum_{i=\beta^*_{\alpha}+2}^L\lambda_i.
		\end{equation*}
		This implies that 
		\begin{equation*}
		\lambda_{\beta^*_{\alpha}+1}\leq \frac{1}{\alpha-(\beta^*_{\alpha}+2)}\sum_{i=\beta^*_{\alpha}+2}^L\lambda_i,
		\end{equation*}
		which is equivalent to
		\begin{equation*}
		\frac{1}{\alpha-(\beta^*_{\alpha}+1)}\sum_{i=\beta^*_{\alpha}+1}^L\lambda_i\leq \frac{1}{\alpha-(\beta^*_{\alpha}+2)}\sum_{i=\beta^*_{\alpha}+2}^L\lambda_i.
		\end{equation*}
		Thus, we have
		\begin{equation*}
		\frac{1}{(\alpha-1)-\beta^*_{\alpha}}\sum_{i=\beta^*_{\alpha}+1}^L\lambda_i\leq \frac{1}{(\alpha-1)-(\beta^*_{\alpha}+1)}\sum_{i=\beta^*_{\alpha}+2}^L\lambda_i,
		\end{equation*}
		which by \eqref{definition-g} implies that
		\begin{equation*}
		g_{\alpha-1,\bm{\lambda}}(\beta^*_{\alpha})\leq g_{\alpha-1,\bm{\lambda}}(\beta^*_{\alpha}+1).
		\end{equation*}
		By the discussion following \Cref{lemma-relation-beta}, we conclude that
		\begin{equation*}
		\beta^*_{\alpha-1}\leq \beta^*_{\alpha}.
		\end{equation*}
	\end{proof}
	
	The following lemma will be used for proving Lemma~\ref{lemma-considerable-lambda1}.

	\begin{lemma}\label{lemma-f-unchange}
		Let $\bm{\lambda}_1=(\lambda_{1,1},\lambda_{1,2},\cdots,\lambda_{1,L})$ and $\bm{\lambda}_2=(\lambda_{2,1},\lambda_{2,2},\cdots,\lambda_{2,L})$ be two vectors such that $\lambda_{1,1}>\lambda_{2,1}$ and $\lambda_{1,i}=\lambda_{2,i}$ for all $2\leq i\leq L$. For any $\alpha_0\in\cL$, if $f_{\alpha_0}(\bm{\lambda}_1)=f_{\alpha_0}(\bm{\lambda}_2)$, then $f_{\alpha}(\bm{\lambda}_1)=f_{\alpha}(\bm{\lambda}_2)$ for all $\alpha\geq \alpha_0$.
	\end{lemma}
	\begin{proof}
		For $\alpha\in\cL$, let $\beta_{\alpha}^1$ and $\beta_{\alpha}^2$ be the values (not necessarily unique) that achieve $f_{\alpha}(\bm{\lambda}_1)$ and $f_{\alpha}(\bm{\lambda}_2)$, respectively. 
		We first prove the claim that among all the possible values of $\beta_{\alpha_0}^1$ and $\beta_{\alpha_0}^2$, there exists a pair of $\left(\beta_{\alpha_0}^1,\beta_{\alpha_0}^2\right)$ such that $\beta_{\alpha_0}^1\geq 1$ and $\beta_{\alpha_0}^2\geq 1$. Consider the following four cases for all the possible values of $\left(\beta_{\alpha_0}^1,\beta_{\alpha_0}^2\right)$:
		\begin{enumerate}[i)]
			\item $\beta_{\alpha_0}^1=\beta_{\alpha_0}^2=0$;
			\item $\beta_{\alpha_0}^1\geq 1$, $\beta_{\alpha_0}^2=0$;
			\item $\beta_{\alpha_0}^1=0$, $\beta_{\alpha_0}^2\geq 1$;
			\item $\beta_{\alpha_0}^1\geq 1$, $\beta_{\alpha_0}^2\geq 1$.
		\end{enumerate}
		Since $f_{\alpha_0}(\bm{\lambda}_1)=f_{\alpha_0}(\bm{\lambda}_2)$, it is easy to see that i) and ii) are impossible.
		If iii) is true, we have
		\begin{equation*}
		\frac{1}{\alpha_0}\sum_{i=1}^L\lambda_{1,i}=\frac{1}{\alpha_0-\beta_{\alpha_0}^2}\sum_{i=\beta_{\alpha_0}^2+1}^L\lambda_{2,i}=\frac{1}{\alpha_0-\beta_{\alpha_0}^2}\sum_{i=\beta_{\alpha_0}^2+1}^L\lambda_{1,i},
		\end{equation*}
		where the second equality follows from $\beta_{\alpha_0}^2\geq 1$.	This implies that 
		\begin{equation*}
		f_{\alpha_0}(\bm{\lambda}_1)=\frac{1}{\alpha_0-\beta_{\alpha_0}^2}\sum_{i=\beta_{\alpha_0}^2+1}^L\lambda_{1,i},
		\end{equation*}
		i.e. $\left(\beta_{\alpha_0}^2,\beta_{\alpha_0}^2\right)$ is a possible pair. This proves the claim.
		For all $\alpha\geq \alpha_0$, by \Cref{lemma-relation-alpha}, we have $\beta_{\alpha}^1\geq 1$ and $\beta_{\alpha}^2\geq 1$. Then by \Cref{thm-opt-f}, we have
		\begin{equation*}
		f_{\alpha}(\bm{\lambda}_1)=f_{\alpha}(\bm{\lambda}_2)=\min_{\beta\in\{1,2,\cdots,\alpha-1\}}g_{\alpha,\bm{\lambda}_1}(\beta).
		\end{equation*}	
		The lemma is proved.
	\end{proof}
	
	Let $\bm{\lambda}^{[1]}$ be the $L$-vector with the first component being 1 and the rest being 0, i.e., 
	\begin{equation}
	\bm{\lambda}^{[1]}=(1,0,0,\cdots,). \label{def-lambda[1]}
	\end{equation}
	\begin{lemma} \label{lemma-considerable-lambda1}
		If $\lambda_1>\sum_{i=2}^L\lambda_i$, let $\bm{\lambda}'=\big(\sum_{i=2}^L\lambda_i,\lambda_2,\lambda_3,\cdots,\lambda_L\big)$. Then for all $\alpha\in\cL$,
		\begin{equation}
		f_{\alpha}(\bm{\lambda})=\bigg(\lambda_1-\sum_{i=2}^L\lambda_i\bigg)f_{\alpha}\big(\bm{\lambda}^{[1]}\big)+f_{\alpha}(\bm{\lambda}').   \label{pf-considerable-lambda1-statement}
		\end{equation}
	\end{lemma}
	
	\begin{proof}
		By \Cref{thm-opt-f}, we have 
		\begin{equation*}
		f_{2}(\bm{\lambda}')=\sum_{i=2}^L\lambda_i.   
		\end{equation*}
		The condition $\lambda_{1}>\sum_{i=2}^L\lambda_i$ implies that
		\begin{equation*}
		g_{2,\bm{\lambda}}(0)>g_{2,\bm{\lambda}}(1).
		\end{equation*}
		Thus by \Cref{thm-opt-f}, we have
		\begin{equation*}
		f_{2}(\bm{\lambda})=g_{2,\bm{\lambda}}(1)=\sum_{i=2}^L\lambda_i.  
		\end{equation*}
		Then
		\begin{equation*}
		f_{2}(\bm{\lambda})=f_{2}(\bm{\lambda}'),
		\end{equation*}
		and by \Cref{lemma-f-unchange}, we have 
		\begin{equation}
		f_{\alpha}(\bm{\lambda})=f_{\alpha}(\bm{\lambda}'),\text{ for all }2\leq \alpha\leq L.  \label{lambda-SMDC-f=f'}
		\end{equation}
		For $2\leq \alpha\leq L$, since $f_{\alpha}(\bm{\lambda}^{[1]})=0$, the equation \eqref{pf-considerable-lambda1-statement} is satisfied by virtue of \eqref{lambda-SMDC-f=f'}. For $\alpha=1$, we can check that
		\begin{eqnarray}
		f_1(\bm{\lambda})&=&\sum_{i=1}^L\lambda_i  \nonumber \\
		&=&\left(\lambda_1-\sum_{i=2}^L\lambda_i\right)\cdot 1+2\sum_{i=2}^L\lambda_i   \nonumber \\
		&=&\left(\lambda_1-\sum_{i=2}^L\lambda_i\right)f_1(\bm{\lambda}^{[1]})+f_1(\bm{\lambda}'), \nonumber
		\end{eqnarray}
		so that \eqref{pf-considerable-lambda1-statement} is also satisfied. This proves the lemma.
	\end{proof}
	
	\begin{lemma} \label{lemma-f-lambda1}
		For any $\eta\in\{1,2,\cdots,L-1\}$, 
		\begin{enumerate}[(i)]
			\item if $\lambda_1\leq \frac{1}{\eta}\sum_{i=2}^L\lambda_i$, then $f_{\alpha}(\bm{\lambda})=g_{\alpha,\bm{\lambda}}(0)$, for $\alpha=1,2,\cdots,\eta+1$;
			\item if $\lambda_1\geq\frac{1}{\eta}\sum_{i=2}^L\lambda_i$, then $f_{\alpha}(\bm{\lambda})=f_{\alpha-1}(\lambda_2,\lambda_3,\cdots,\lambda_L)$, for $\alpha=\eta+1,\eta+2,\cdots,L$.
		\end{enumerate}
	\end{lemma}
	\begin{remark}\label{remark-lemma-f-lambda1}
		If $\lambda_1=\frac{1}{\eta}\sum_{i=2}^L\lambda_i$, we have from the lemma that
		\begin{equation*}
		f_{\eta+1}(\bm{\lambda})=f_{\eta}(\lambda_2,\lambda_3, \cdots,\lambda_L)= \frac{1}{\eta+1}\sum_{i=1}^L\lambda_i.  
		\end{equation*}
		In this case,
		\begin{equation*}
		f_{\alpha}(\bm{\lambda})=
		\begin{cases}
		\frac{1}{\alpha}\sum_{i=1}^L\lambda_i,& \text{for }\alpha\leq \eta+1 \\
		f_{\alpha-1}(\lambda_2,\lambda_3,\cdots,\lambda_L),&\text{for }\alpha\geq \eta+1.
		\end{cases}  
		\end{equation*}
	\end{remark}
	\begin{proof}
		We first prove (i). For $\alpha\leq \eta+1$, it is easy to check that
		\begin{equation}
		\frac{1}{\alpha}\left(1+\frac{1}{\eta}\right)\leq \frac{1}{\alpha-1}.  \label{pf-lemma-f-lambda1-alpha<}
		\end{equation}
		Thus,
		\begin{eqnarray}
		\frac{1}{\alpha}\sum_{i=1}^{L}\lambda_i&=&\frac{1}{\alpha}\lambda_1+\frac{1}{\alpha}\sum_{i=2}^L\lambda_i  \nonumber \\
		&\leq &\frac{1}{\alpha}\bigg(\frac{1}{\eta}\sum_{i=2}^L\lambda_i\bigg)+\frac{1}{\alpha}\sum_{i=2}^L\lambda_i  \label{pf-lemma-f-lambda1-f1-1} \\
		&=&\frac{1}{\alpha}\left(1+\frac{1}{\eta}\right)\sum_{i=2}^L\lambda_i  \nonumber \\
		&\leq &\frac{1}{\alpha-1}\sum_{i=2}^L\lambda_i,  \label{pf-lemma-f-lambda1-f1-2}
		\end{eqnarray}
		where \eqref{pf-lemma-f-lambda1-f1-1} follows from the assumption that $\lambda_1\leq \frac{1}{\eta}\sum_{i=2}^L\lambda_i$ and \eqref{pf-lemma-f-lambda1-f1-2} follows from \eqref{pf-lemma-f-lambda1-alpha<}. Then by the discussion following \Cref{lemma-relation-beta}, we have 
		\begin{equation*}
		f_{\alpha}(\bm{\lambda})=\frac{1}{\alpha}\sum_{i=1}^L\lambda_i=g_{\alpha,\bm{\lambda}}(0).
		\end{equation*}
		Next, we prove (ii). For $\alpha\geq \eta+1$, it is easy to check that
		\begin{equation}
		\frac{1}{\alpha}\left(1+\frac{1}{\eta}\right)\geq \frac{1}{\alpha-1}.  \label{pf-lemma-f-lambda1-alpha>}
		\end{equation}
		Similar to the derivation of \eqref{pf-lemma-f-lambda1-f1-2}, with the assumption that $\lambda_1\geq \frac{1}{\eta}\sum_{i=2}^L\lambda_i$, \eqref{pf-lemma-f-lambda1-alpha>} implies that
		\begin{equation}
		\frac{1}{\alpha}\sum_{i=1}^{L}\lambda_i\geq \frac{1}{\alpha-1}\sum_{i=2}^L\lambda_i.  \label{pf-lemma-f-lambda1-f2-1}
		\end{equation}
		Thus, we have 
		\begin{eqnarray}
		f_{\alpha}(\bm{\lambda})&=&\min_{\beta\in\{0,1,\cdots,\alpha-1\}}\left\{\frac{1}{\alpha-\beta}\sum_{i=\beta+1}^L\lambda_i\right\}  \nonumber \\
		&=&\min_{\beta\in\{1,2,\cdots,\alpha-1\}}\left\{\frac{1}{\alpha-\beta}\sum_{i=\beta+1}^L\lambda_i\right\}  \label{pf-lemma-f-lambda1-f2-2} \\
		&=&\min_{\beta\in\{0,1,\cdots,\alpha-2\}}\left\{\frac{1}{(\alpha-1)-\beta}\sum_{i=\beta+2}^L\lambda_i\right\}  \nonumber \\
		&=&f_{\alpha-1}(\lambda_2,\lambda_3,\cdots,\lambda_L),  \nonumber 
		\end{eqnarray}
		where \eqref{pf-lemma-f-lambda1-f2-2} follows from \eqref{pf-lemma-f-lambda1-f2-1}. This proves the lemma.
	\end{proof}
	
	The following lemma implies that $f_{\alpha}(\bm{\lambda})$ is a concave function of $\bm{\lambda}\in\mathbb{R}_+^L$ for all $\alpha\in\cL$. Note that the vectors in this lemma are not necessarily ordered.
	\begin{lemma}\label{lemma-concave-f}
		For any $\alpha\in\cL$,
		\begin{equation*}
		f_{\alpha}(\mu_1\bm{\lambda}_1+\mu_2\bm{\lambda}_2)\geq \mu_1f_{\alpha}(\bm{\lambda}_1)+\mu_2f_{\alpha}(\bm{\lambda}_2)
		\end{equation*}
		for any $\bm{\lambda}_1,\bm{\lambda}_2\in\mathbb{R}_+^L$ and $\mu_1,\mu_2\geq 0$.
	\end{lemma}
	\begin{proof}
		Let $\bm{\lambda}_1=(\lambda_{1,1},\lambda_{1,2},\cdots,\lambda_{1,L})$ and $\bm{\lambda}_2=(\lambda_{2,1},\lambda_{2,2},\cdots,\lambda_{2,L})$. Let $\pi_1(\cdot),~\pi_2(\cdot)$ be two permutations of $\{1,2,\cdots,L\}$ such that 
		\begin{equation*}
		\lambda_{1,\pi_1(1)}\geq \lambda_{1,\pi_1(2)}\geq \cdots\geq \lambda_{1,\pi_1(L)}
		\end{equation*}
		and 
		\begin{equation*}
		\lambda_{2,\pi_2(1)}\geq \lambda_{2,\pi_2(2)}\geq \cdots\geq \lambda_{2,\pi_2(L)}.
		\end{equation*}
		Denote the ordered vectors by $\pi_1(\bm{\lambda}_1)$ and $\pi_2(\bm{\lambda}_2)$, respectively. For any $\beta=0,1,\cdots,\alpha-1$, it is easy to see that
		\begin{equation*}
		\frac{1}{\alpha-\beta}\sum_{i=\beta+1}^L\lambda_{1,i}\geq \frac{1}{\alpha-\beta}\sum_{i=\beta+1}^L\lambda_{1,\pi_1(i)}
		\end{equation*}
		and 
		\begin{equation*}
		\frac{1}{\alpha-\beta}\sum_{i=\beta+1}^L\lambda_{2,i}\geq \frac{1}{\alpha-\beta}\sum_{i=\beta+1}^L\lambda_{2,\pi_2(i)}.
		\end{equation*}
		Thus, we have
		\begin{align*}
		&\hspace{0.5cm}\frac{1}{\alpha-\beta}\sum_{i=\beta+1}^L(\mu_1\lambda_{1,i}+\mu_2\lambda_{2,i})  \\
		&\geq\frac{1}{\alpha-\beta}\sum_{i=\beta+1}^L(\mu_1\lambda_{1,\pi_1(i)}+\mu_2\lambda_{2,\pi_2(i)}).
		\end{align*}
		This implies that 
		\begin{equation*}
		f_{\alpha}(\mu_1\bm{\lambda}_1+\mu_2\bm{\lambda}_2)\geq f_{\alpha}(\mu_1\pi_1(\bm{\lambda}_1)+\mu_2\pi_2(\bm{\lambda}_2)).  
		\end{equation*}
		For any $\alpha\in\cL$, it is easy to check that
		\begin{equation*}
		f_{\alpha}(\pi_1(\bm{\lambda}_1))=f_{\alpha}(\bm{\lambda}_1)  
		\end{equation*}
		and
		\begin{equation*}
		f_{\alpha}(\pi_2(\bm{\lambda}_2))=f_{\alpha}(\bm{\lambda}_2).  
		\end{equation*}
		Therefore, if the lemma holds for any ordered vectors $\bm{\lambda}_1$ and $\bm{\lambda}_2$, then the lemma holds for any vectors $\bm{\lambda}_1$ and $\bm{\lambda}_2$ (not necessarily ordered), because
		\begin{eqnarray}
		f_{\alpha}(\mu_1\bm{\lambda}_1+\mu_2\bm{\lambda}_2)&\geq& f_{\alpha}(\mu_1\pi_1(\bm{\lambda}_1)+\mu_2\pi_2(\bm{\lambda}_2))  \nonumber  \\
		&\geq &\mu_1f_{\alpha}(\pi_1(\bm{\lambda}_1))+\mu_2f_{\alpha}(\pi_2(\bm{\lambda}_2))  \nonumber \\
		&=&\mu_1f_{\alpha}(\bm{\lambda}_1)+\mu_2f_{\alpha}(\bm{\lambda}_2).  \nonumber 
		\end{eqnarray}
		
		Thus without loss of generality, we assume that $\bm{\lambda}_1$ and $\bm{\lambda}_2$ are ordered. Then for any $\beta=0,1,\cdots,\alpha-1$, we have from \Cref{thm-opt-f} that
		\begin{equation*}
		\frac{1}{\alpha-\beta}\sum_{i=\beta+1}^L\lambda_{1,i}\geq f_{\alpha}(\bm{\lambda}_1)
		\end{equation*}
		and
		\begin{equation*}
		\frac{1}{\alpha-\beta}\sum_{i=\beta+1}^L\lambda_{2,i}\geq f_{\alpha}(\bm{\lambda}_2),
		\end{equation*}
		which implies
		\begin{equation*}
		\frac{1}{\alpha-\beta}\sum_{i=\beta+1}^L(\mu_1\lambda_{1,i}+\mu_2\lambda_{2,i})\geq \mu_1f_{\alpha}(\bm{\lambda}_1) +\mu_2 f_{\alpha}(\bm{\lambda}_2).
		\end{equation*}
		By taking the minimum over all $\beta=0,1,\cdots,\alpha-1$, we obtain
		\begin{align*}
		&\hspace{0.5cm}\min_{\beta\in\{0,1,\cdots,\alpha-1\}}\left\{\frac{1}{\alpha-\beta}\sum_{i=\beta+1}^L(\mu_1\lambda_{1,i}+\mu_2\lambda_{2,i})\right\}  \\
		&\geq \mu_1f_{\alpha}(\bm{\lambda}_1) +\mu_2 f_{\alpha}(\bm{\lambda}_2),
		\end{align*}
		which by \Cref{thm-opt-f} is equivalent to
		\begin{equation*}
		f_{\alpha}(\mu_1\bm{\lambda}_1+\mu_2\bm{\lambda}_2)\geq \mu_1 f_{\alpha}(\bm{\lambda}_1)+\mu_2 f_{\alpha}(\bm{\lambda}_2).  \label{pf-concave-min-is-concave}
		\end{equation*}
		This proves the lemma.
	\end{proof}
	
	\section{The Minimum Sufficient Set of Inequalities}  \label{section-hyperplane}
	Even though the superposition region $\cR_{\sup}$ (cf. \eqref{rate-contraints} and \eqref{Rsup=Rh}) can be explicitly characterized by solving $f_\alpha(\bm{\lambda})$ in \Cref{thm-opt-f}, an uncountable number of inequalities are involved. For a fixed $L$,
	among all these inequalities, 
	only a finite number of them are needed because $\cR_{\sup}$ is a polytope. In this section, we provide a method to determine this minimum sufficient set of inequalities.
	
	For any $\bm{\lambda}\in\mathbb{R}_+^L$, let $\pi(\cdot)$ be a permutation of $\{1,2,\cdots,L\}$ such that 
	\begin{equation}
	\lambda_{\pi(1)}\geq \lambda_{\pi(2)}\geq \cdots\geq \lambda_{\pi(L)}.  \label{lambda-permute-def}
	\end{equation}
	Recall that we consider only $\bm{\lambda}$'s whose minimum nonzero element is equal to 1. Let $\zeta\in\cL$ be such that 
	\begin{equation}
	\lambda_{\pi(\zeta)}=1  \label{lambda(zeta)}
	\end{equation}
	and for $j=\zeta+1,\zeta+2,\cdots,L$,
	\begin{equation}
	\lambda_{\pi(j)}=0.
	\end{equation}
	Toward listing all the inequalities defining $\cR_{\sup}$, we consider a certain finite subset of $\mathbb{R}_+^L$ defined as follows. Let $\cG_L$ be the collection of all $\bm{\lambda}\in\mathbb{R}_+^L$ such that for $j=\zeta-1,\zeta-2,\cdots,1$,
	\begin{equation}
	\lambda_{\pi(j)}\in\left\{\sum_{i=j+1}^L\lambda_{\pi(i)},\frac{1}{2}\sum_{i=j+1}^L\lambda_{\pi(i)},\cdots,\frac{1}{\theta_{j+1}+1}\sum_{i=j+1}^L\lambda_{\pi(i)}\right\},  \label{def-cG_L}
	\end{equation}
	where $\theta_{\zeta}=0$ and for $j=\zeta-2,\cdots,1$, $\theta_{j+1}$ is the integer such that
	\begin{equation}
	\lambda_{\pi(j+1)}=\frac{1}{\theta_{j+1}}\sum_{i=j+2}^L\lambda_{\pi(i)}.  \label{def-cG_L-lambda(j+1)}
	\end{equation}
	
	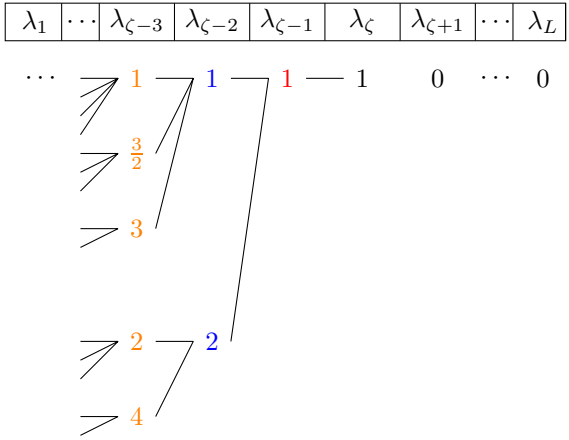
\begin{figure}[!ht]
		\centering
		\begin{tikzpicture}
		\draw (1.25,8) rectangle (8.7,8.5); 
		\draw (2.0,8)--(2.0,8.5) (2.5,8)--(2.5,8.5) (3.5,8)--(3.5,8.5) (4.5,8)--(4.5,8.5) (5.5,8)--(5.5,8.5) (6.5,8)--(6.5,8.5) (7.5,8)--(7.5,8.5) (8,8)--(8,8.5) (8.7,8)--(8.7,8.5);
		\node at (1.65,8.25) {$\lambda_1$}; \node at (2.3,8.25) {$\cdots$}; 
		\node at (3,8.25) {$\lambda_{\zeta-3}$}; \node at (4,8.25) {$\lambda_{\zeta-2}$}; \node at (5,8.25) {$\lambda_{\zeta-1}$}; \node at (6,8.25) {$\lambda_{\zeta}$};
		\node at (7,8.25) {$\lambda_{\zeta+1}$}; \node at (7.8,8.25) {$\cdots$}; \node at (8.4,8.25) {$\lambda_L$};
		
		\node at (1.75,7.5) {$\cdots$};
		\node at (6,7.5) {$1$}; \node at (7,7.5) {$0$}; \node at (7.8,7.5) {$\cdots$}; \node at (8.4,7.5) {$0$};
		
		\draw (5.25,7.5)--(5.75,7.5);
		\node at (5,7.5) {{\color{red}$1$}};
		
		\draw (4.75,7.5)--(4.25,7.5) (4.75,7.5)--(4.25,4);
		\node at (4,7.5) {\color{blue}{$1$}}; \node at (4,4) {\color{blue}{$2$}};
		
		\draw (3.75,7.5)--(3.25,7.5) (3.75,7.5)--(3.25,6.5) (3.75,7.5)--(3.25,5.5);
		\node at (3,7.5) {\color{orange}{$1$}}; \node at (3,6.5) {\color{orange}{$\frac{3}{2}$}}; \node at (3,5.5) {\color{orange}{$3$}};
		\draw (3.75,4)--(3.25,4) (3.75,4)--(3.25,3);
		\node at (3,4) {\color{orange}{$2$}}; \node at (3,3) {\color{orange}{$4$}};
		
		\draw (2.75,7.5)--(2.25,7.5) (2.75,7.5)--(2.25,7.25) (2.75,7.5)--(2.25,7.0) (2.75,7.5)--(2.25,6.75);
		\draw (2.75,6.5)--(2.25,6.5) (2.75,6.5)--(2.25,6.25) (2.75,6.5)--(2.25,6.0);
		\draw (2.75,5.5)--(2.25,5.5) (2.75,5.5)--(2.25,5.25);
		\draw (2.75,4)--(2.25,4) (2.75,4)--(2.25,3.75) (2.75,4)--(2.25,3.5);
		\draw (2.75,3)--(2.25,3) (2.75,3)--(2.25,2.75);
		\end{tikzpicture}
		\caption{Recursive generation of vectors in $\cG_L^0$}
		\label{fig-cG0-generator}
	\end{figure}
	
	Here, \eqref{lambda(zeta)}-\eqref{def-cG_L-lambda(j+1)} not only defines $\cG_L$ but in fact provides a method to exhaust all $\bm{\lambda}\in\cG_L$. For $\zeta=1$, the only possible $\bm{\lambda}$ are $\bm{\lambda}^{[1]}=(1,0,0,\cdots)$ and its permutations. For $\zeta\geq 2$, starting with $\lambda_{\pi(\zeta)}=1$, the values of $\lambda_{\pi(\zeta-1)}, \lambda_{\pi(\zeta-2)}, \cdots, \lambda_{\pi(1)}$ can be chosen recursively according to \eqref{def-cG_L}. It is easy to check that 
	\begin{equation}
	\theta_j\in\{1,2,\cdots,\theta_{j+1}+1\}\label{range-theta(j)-1}
	\end{equation}
	and
	\begin{equation}
	1\leq \theta_j\leq \zeta-j  \label{range-theta(j)-2}
	\end{equation}
	for $1\leq j\leq \zeta-1$. Furthermore, for the last element of the set in \eqref{def-cG_L} which is the smallest in the set, we have
	\begin{eqnarray}
	&&\hspace{-0.8cm}\frac{1}{\theta_{j+1}+1}\sum_{i=j+1}^L\lambda_{\pi(i)} \nonumber \\
	&=&\frac{1}{\theta_{j+1}+1}\left(\frac{1}{\theta_{j+1}}\sum_{i=j+2}^L\lambda_{\pi(i)}+\sum_{i=j+2}^L\lambda_{\pi(i)}\right)  \nonumber \\
	&=&\frac{1}{\theta_{j+1}+1}~\frac{\theta_{j+1}+1}{\theta_{j+1}}\sum_{i=j+2}^L\lambda_{\pi(i)}  \nonumber \\
	&=&\frac{1}{\theta_{j+1}}\sum_{i=j+2}^L\lambda_{\pi(i)}  \nonumber \\
	&=&\lambda_{\pi(j+1)},  \nonumber 
	\end{eqnarray}
	so that $\lambda_{\pi(j)}\geq \lambda_{\pi(j+1)}$ as required by \eqref{lambda-permute-def}. Also, we see from \eqref{def-cG_L} that 
	\begin{enumerate}
		\item for $\zeta\geq 2$, $\lambda_{\pi(\zeta-1)}=\lambda_{\pi(\zeta)}=1$;
		\item for $\zeta\geq 3$, $\lambda_{\pi(j+1)}$ is always a possible choice for $\lambda_{\pi(j)}$ for $1\leq j\leq \zeta-2$.
	\end{enumerate}
	
	Denote the cardinality of $\cG_L$ by $S_L$. Let $\cG_L^0=\left\{\bm{\lambda}\in\cG_L:\bm{\lambda} \text{ is ordered}\,\right\}$, and denote its cardinality by 
	\begin{equation}
	|\cG_L^0|=S_L^0.  \label{def-S_L-0}
	\end{equation}
	The vectors in $\cG_L^0$ are generated recursively as illustrated in \reffig{fig-cG0-generator}. 
	For the ease of notation, we let
	\begin{equation}
	\cG_L^0=\left\{\bm{\lambda}^{(1)}, \bm{\lambda}^{(2)},\cdots, \bm{\lambda}^{(S_L^0)}\right\}   \label{denote-cG_L0}
	\end{equation}
	with $\bm{\lambda}^{(1)}=\bm{\lambda}^{[1]}$ (cf. \eqref{def-lambda[1]}) and
	\begin{equation*}
	\cG_L=\cG_L^0\cup\left\{\bm{\lambda}^{(S_L^0+1)}, \bm{\lambda}^{(S_L^0+2)},\cdots, \bm{\lambda}^{(S_L)}\right\}.
	\end{equation*}
	In other words, the set $\cG_L$ is the collection of all possible permutations of the vectors in $\cG_L^0$. 
	
	For $i=1,2,\cdots,S_L$, let $\pi_i(\cdot)$ be a permutation of $\{1,2,\cdots,L\}$ such that
	\begin{equation}
	\lambda^{(i)}_{\pi_i(1)}\geq \lambda^{(i)}_{\pi_i(2)}\geq \cdots\geq \lambda^{(i)}_{\pi_i(L)}.  \label{def-pi(i)}
	\end{equation}
	For any $\bm{\lambda}\in\mathbb{R}_+^L$, let $\bm{f}(\bm{\lambda})=\big(f_1(\bm{\lambda}),f_2(\bm{\lambda}),\cdots,f_L(\bm{\lambda})\big)$. 
	The following technical lemma will be instrumental for the proof of our main theorem.
	\begin{lemma}\label{lemma-combination-of-two}
		Consider any ordered vector $\bm{\lambda}\in\mathbb{R}_+^L$ such that $\bm{\lambda}\neq\bm{\lambda}^{[1]}$. Assume there exists $c_i\geq 0,~i=1,2,\cdots,S_L$ such that 
		\begin{equation}
		\big(\bm{\lambda},\bm{f}(\bm{\lambda})\big)=\sum_{i=1}^{S_L}c_i\cdot\big(\bm{\lambda}^{(i)},\bm{f}(\bm{\lambda}^{(i)})\big).  \label{assumption-lemma-combination-two}
		\end{equation}
		Let $\cI=\big\{i\in\{1,2,\cdots,S_L\}:c_i\neq 0\big\}$. For any $\eta\in\{1,2,\cdots,\zeta-1\}$,
		\begin{enumerate}[(i)]
			\item if $\lambda_1\leq \frac{1}{\eta}\sum_{j=2}^L\lambda_j$, then $\lambda^{(i)}_{\pi_i(1)}\leq \frac{1}{\eta} \sum_{j=2}^L\lambda^{(i)}_{\pi_i(j)}$ for all $i\in\cI$;
			\item if $\lambda_1\geq \frac{1}{\eta}\sum_{j=2}^L\lambda_j$, then $\lambda^{(i)}_{\pi_i(1)}\geq \frac{1}{\eta} \sum_{j=2}^L\lambda^{(i)}_{\pi_i(j)}$ for all $i\in\cI$.
		\end{enumerate}
	\end{lemma}
	\begin{remark}
		In the above, since $\bm{\lambda}$ is ordered, we have 
		\begin{equation}
		\lambda_1\geq \frac{1}{\zeta-1}\sum_{j=2}^L\lambda_j.   \label{lambda1>=constraint}
		\end{equation}
		Therefore, when $\eta=\zeta-1$, the condition in (i) can only be satisfied with an equality, i.e., $\lambda_1=\frac{1}{\zeta-1}\sum_{j=2}^L\lambda_j$.
	\end{remark}
	\begin{proof}
		See Appendix \ref{section-proof-combination-two}.
	\end{proof}
	
	\medskip
	The assumption that $\bm{\lambda}^{(i)}\in\cG_L$ for $1\leq i\leq S_L$ is not invoked in the proof of \Cref{lemma-combination-of-two}. By taking this assumption into account, \Cref{lemma-combination-of-two} can be further strengthened with the following setup. For any ordered vector $\bm{\lambda}\in\mathbb{R}_+^L$ not equal to $\bm{\lambda}^{[1]}$, by the constraint in \eqref{lambda1>=constraint}, there exists a unique $\eta\in\{1,2,\cdots,\zeta-1\}$ such that
	\begin{equation}
	\frac{1}{\eta} \sum_{j=2}^L\lambda_j\leq \lambda_1<\frac{1}{\eta-1} \sum_{j=2}^L\lambda_j. \label{combination-two-condition}
	\end{equation}
	In the sequel, we adopt the convention that 
	\begin{equation*}
	\frac{1}{0}\cdot c=
	\begin{cases}
	\infty,&\text{if }c\neq 0 \\
	1,&\text{if }c=0.
	\end{cases}
	\end{equation*}
	Then the upper bound in \eqref{combination-two-condition} is $\infty$ when $\eta=1$.
	
	\begin{lemma}\label{lemma-combination-of-two-range}
		For any ordered vector $\bm{\lambda}\in\mathbb{R}_+^L$ such that $\bm{\lambda}\neq\bm{\lambda}^{[1]}$, assume there exists $c_i\geq 0,~i=1,2,\cdots,S_L$ such that 
		\begin{equation}
		\big(\bm{\lambda},\bm{f}(\bm{\lambda})\big)=\sum_{i=1}^{S_L}c_i\cdot\big(\bm{\lambda}^{(i)},\bm{f}(\bm{\lambda}^{(i)})\big).  \label{assumption-lemma-combination-two-range}
		\end{equation}
		Then for all $i\in\cI$,
		\begin{equation}
		\lambda^{(i)}_{\pi_i(1)}\in\left\{\frac{1}{\eta} \sum_{j=2}^L\lambda^{(i)}_{\pi_i(j)},~\frac{1}{\eta-1} \sum_{j=2}^L\lambda^{(i)}_{\pi_i(j)}\right\},  \label{lambda1-lemma-combination-of-two-range}
		\end{equation}
		where $\eta$ depends on $\bm{\lambda}$ and is defined in \eqref{combination-two-condition}.	In particular, if the lower bound in \eqref{combination-two-condition} is tight, then $\lambda^{(i)}_{\pi_i(1)}=\frac{1}{\eta} \sum_{j=2}^L\lambda^{(i)}_{\pi_i(j)}$ for all $i\in\cI$.
	\end{lemma}
	\begin{proof}
		The lemma can be easily obtained from \Cref{lemma-combination-of-two}. See details in Appendix \ref{section-proof-lemma-combination-two-range}. 
	\end{proof}
	\begin{remark}
		For all $1\leq i\leq S_L$, $\lambda^{(i)}_{\pi_i(1)}$ can in general take one of the $\theta_2+1$ values prescribed in \eqref{def-cG_L}. However, under the constraint \eqref{assumption-lemma-combination-two-range}, the above lemma asserts that for all $i\in\cI$, $\lambda^{(i)}_{\pi_i(1)}$ can only take one of the two values prescribed in \eqref{lambda1-lemma-combination-of-two-range}.
	\end{remark}

	Let $\cR^*$ be the collection of nonnegative rate tuples $\bm{R}$ such that 
	\begin{equation}
	\sum_{l=1}^L\lambda_lR_l\geq \sum_{\alpha=1}^L f_{\alpha}(\bm{\lambda})H(X_{\alpha}), \text{ for all }\bm{\lambda}\in\cG_L.  \label{def-cR*}
	\end{equation}
	The next theorem shows that $\cR^*$ provides an equivalent characterization of $\cR_{\sup}$. 
	Note that $\cR^*$ is the intersection of only a finite set of halfspaces, because the cardinality of $\cG_L$ is finite in view of its definition in \eqref{lambda(zeta)}-\eqref{def-cG_L}. 
	Thus, $\cR^*$ is more explicit than $\cR_{\text{h}}$. 
	For $L=1,2,\cdots,5$, all the rate constraints of $\cR^*$ with ordered coefficient vectors are listed in Appendix \ref{section-table}. 
	\begin{theorem} \label{thm-alter-region}
		$\cR_{\sup}=\cR^*$.
	\end{theorem}
	\begin{proof}
		We prove the theorem by showing that $\cR_{\text{h}}=\cR^*$. Since $\cG_L\subseteq\mathbb{R}_+^L$, we have $\cR_{\text{h}}\subseteq\cR^*$. To show $\cR^*\subseteq\cR_{\text{h}}$, we consider the following. Define three sets of $(2L)$-vectors by
		\begin{equation*}
		\cF_L^1=\left\{(\bm{\lambda},\bm{f}(\bm{\lambda})):\bm{\lambda}\in\mathbb{R}_+^L\right\},
		\end{equation*}
		\begin{equation*}
		\cF_L^2=\left\{(\bm{\lambda},\bm{f}(\bm{\lambda})):\bm{\lambda}\in\cG_L\right\},
		\end{equation*}
		and
		\begin{equation*}
		\cF_L^3=\left\{(\bm{\lambda},\bm{f}(\bm{\lambda})):\bm{\lambda}\in\cG_L^0\right\}.
		\end{equation*}
		Note that none of $\cF_L^2$ and $\cF_L^3$ is a vector space since they are not closed under vector addition. We prove $\cR^*\subseteq\cR_{\text{h}}$ by proving the claim that any $(\bm{\lambda},\bm{f}(\bm{\lambda}))\in\cF_L^1$ is a conic combination of the vectors in $\cF_L^2$.
		
		Without loss of generality, we consider only $\bm{\lambda}$ such that $\lambda_1\geq \lambda_2\geq \cdots\geq \lambda_L$, and show that $(\bm{\lambda},\bm{f}(\bm{\lambda}))$ for any such $\bm{\lambda}$ is a conic combination of the vectors in $\cF_L^3$. We prove the claim by induction on $L$ for $L\geq 1$. Since we consider only $\bm{\lambda}$'s whose minimum nonzero element is equal to 1, it is easy to see that $\cF_1^1=\cF_1^3=\{(1,1)\}$ and thus the claim is true for $L=1$.
		
		Assume the claim is true for $L=N$. We will show that the claim is true for $L=N+1$. This can readily be verified for $\bm{\lambda}\in\mathbb{R}_+^{N+1}$ such that $\zeta=1$. Thus, we consider only $\bm{\lambda}\in\mathbb{R}_+^{N+1}$ such that $\zeta\geq 2$. For any ordered vector $\bm{\lambda}_N=(\lambda_2,\lambda_3,\cdots,\lambda_{N+1})\in\mathbb{R}_+^N$, let $\bm{\lambda}_{N+1}=(\lambda_1,\lambda_2,\cdots,\lambda_{N+1})$ where $\lambda_1\geq \lambda_2$. By the induction hypothesis, there exist $c_i\geq 0,~i=1,2,\cdots,S_N^0$ such that
		\begin{equation}
		(\bm{\lambda}_N,\bm{f}(\bm{\lambda}_N))=\sum_{i=1}^{S_N^0} c_i\left(\bm{\lambda}^{(i)}_N,\bm{f}(\bm{\lambda}^{(i)}_N)\right),  \label{pf-thm-alter-region-induction-hypo}
		\end{equation}
		where $\bm{\lambda}^{(i)}_N,~i=1,2,\cdots,S_N^0$ are distinct elements of $\cG_N^0$. Let $\bm{\lambda}^{(i)}_N=(\lambda^{(i)}_2,\lambda^{(i)}_3,\cdots,\lambda^{(i)}_{N+1})$. Recall that $\cI=\{i\in\{1,2,\cdots,S_N\}:c_i\neq 0\}$ in \Cref{lemma-combination-of-two}. For simplicity, let $c_i=0$ for all $i\in\{S_N^0+1,S_N^0+2,\cdots,S_N\}$. For any $i\in\cI$, by \Cref{lemma-combination-of-two-range}, we have
		\begin{equation}
		\lambda^{(i)}_2\in\left\{\frac{1}{\eta'}\sum_{j=3}^{N+1}\lambda^{(i)}_j, \frac{1}{\eta'-1}\sum_{j=3}^{N+1}\lambda^{(i)}_j\right\},  \label{pf-alter-region-lambda(i2)-range}
		\end{equation}
		where $\eta'\in\{1,2,\cdots,N-1\}$ is unique and determined by
		\begin{equation}
		\frac{1}{\eta'}\sum_{j=3}^{N+1}\lambda_j\leq \lambda_2< \frac{1}{\eta'-1}\sum_{j=3}^{N+1}\lambda_j.  \label{pf-alter-region-assumption-lambda2}
		\end{equation}
		Since the second inequality in~\eqref{pf-alter-region-assumption-lambda2} is equivalent to $\lambda_2<\frac{1}{\eta'}\sum_{j=2}^{N+1}\lambda_j$, we consider the following three cases for~$\lambda_1$:
		\begin{enumerate}[~~~a)]
			\item $\frac{1}{\eta'}\sum_{j=2}^{N+1}\lambda_j<\lambda_1\leq \sum_{j=2}^{N+1}\lambda_j$;
			\item $\lambda_2\leq \lambda_1\leq \frac{1}{\eta'}\sum_{j=2}^{N+1}\lambda_j$;
			\item $\lambda_1>\sum_{j=2}^{N+1}\lambda_j$.
		\end{enumerate}	
		
		\textit{Case a):} If $\frac{1}{\eta'}\sum_{j=2}^{N+1}\lambda_j<\lambda_1\leq \sum_{j=2}^{N+1}\lambda_j$, there exists a unique $\varphi\in\{1,2,\cdots,\eta'-1\}$ such that
		\begin{equation}
		\frac{1}{\varphi+1}\sum_{j=2}^{N+1}\lambda_j< \lambda_1\leq \frac{1}{\varphi}\sum_{j=2}^{N+1}\lambda_j.  \label{pf-thm-alter-region-range(lambda1)}
		\end{equation}
		Then by \Cref{lemma-f-lambda1}, we have
		\begin{equation}
		f_{\alpha}(\bm{\lambda}_{N+1})=
		\begin{cases}
		\frac{1}{\alpha}\sum_{j=1}^{N+1}\lambda_j,&\text{for }1\leq \alpha\leq \varphi+1  \\
		f_{\alpha-1}(\bm{\lambda}_N),&\text{for }\varphi+2\leq \alpha\leq N+1.
		\end{cases}  \label{pf-thm-alter-region-f}
		\end{equation}
		For all $i\in\cI$, let 
		\begin{equation}
		\lambda^{(1,i)}_1=\frac{1}{\varphi+1}\sum_{j=2}^{N+1}\lambda^{(i)}_j  \label{pf-thm-alter-region-def-lambda(1,i)}
		\end{equation}
		and
		\begin{equation}
		\lambda^{(2,i)}_1=\frac{1}{\varphi}\sum_{j=2}^{N+1}\lambda^{(i)}_j.   \label{pf-thm-alter-region-def-lambda(2,i)}
		\end{equation}
		For $j\in\{2,3,\cdots,N+1\}$, for notational simplicity, let
		\begin{equation}
		\lambda^{(1,i)}_j=\lambda^{(2,i)}_j=\lambda^{(i)}_j.  \label{pf-thm-alter-region-def-lambda(j)}
		\end{equation}
		Let
		\begin{equation*}
		\bm{\lambda}^{(1,i)}_{N+1}=(\lambda^{(1,i)}_1,\lambda^{(1,i)}_2,\cdots,\lambda^{(1,i)}_{N+1})
		\end{equation*}
		and
		\begin{equation*}
		\bm{\lambda}^{(2,i)}_{N+1}=(\lambda^{(2,i)}_1,\lambda^{(2,i)}_2,\cdots,\lambda^{(2,i)}_{N+1}).
		\end{equation*}
		From \eqref{pf-alter-region-lambda(i2)-range}, \eqref{def-cG_L}, and the range of $\varphi$, we can check that $\bm{\lambda}^{(1,i)}_{N+1},\bm{\lambda}^{(2,i)}_{N+1}\in\cG_{N+1}^0$. By \Cref{remark-lemma-f-lambda1} following \Cref{lemma-f-lambda1}, we have
		\begin{equation}
		f_{\alpha}(\bm{\lambda}^{(1,i)}_{N+1})=
		\begin{cases}
		\frac{1}{\alpha}\sum_{j=1}^{N+1}\lambda^{(1,i)}_j,& \text{for }1\leq \alpha\leq \varphi+2 \\
		f_{\alpha-1}(\bm{\lambda}^{(i)}_N),&\text{for }\varphi+2\leq \alpha\leq N+1,
		\end{cases}  \label{pf-thm-alter-region-f1}
		\end{equation}
		and
		\begin{equation}
		f_{\alpha}(\bm{\lambda}^{(2,i)}_{N+1})=
		\begin{cases}
		\frac{1}{\alpha}\sum_{j=1}^{N+1}\lambda^{(2,i)}_j,& \text{for }1\leq \alpha\leq \varphi+1 \\
		f_{\alpha-1}(\bm{\lambda}^{(i)}_N),&\text{for }\varphi+1\leq \alpha\leq N+1.
		\end{cases}  \label{pf-thm-alter-region-f2}
		\end{equation}
		Consider the conic combination of $\lambda^{(1,i)}_1$ for $i\in\cI$,
		\begin{eqnarray}
		\sum_{i\in\cI}c_i\lambda^{(1,i)}_1&=&\frac{1}{\varphi+1}\sum_{i\in\cI}c_i\sum_{j=2}^{N+1}\lambda^{(i)}_j  \label{pf-alter-region-lambda1i-1} \\
		&=&\frac{1}{\varphi+1}\sum_{j=2}^{N+1}\bigg(\sum_{i\in\cI}c_i\lambda^{(i)}_j\bigg)  \nonumber  \\
		&=&\frac{1}{\varphi+1}\sum_{j=2}^{N+1}\lambda_j   \label{pf-alter-region-lambda1i-2} \\
		&<&\lambda_1,   \label{pf-alter-region-lambda1i-last}
		\end{eqnarray}
		where \eqref{pf-alter-region-lambda1i-1} follows from \eqref{pf-thm-alter-region-def-lambda(1,i)}, \eqref{pf-alter-region-lambda1i-2} follows from \eqref{pf-thm-alter-region-induction-hypo}, and \eqref{pf-alter-region-lambda1i-last} follows from \eqref{pf-thm-alter-region-range(lambda1)}. Similarly, from \eqref{pf-thm-alter-region-def-lambda(2,i)}, \eqref{pf-thm-alter-region-induction-hypo}, and \eqref{pf-thm-alter-region-range(lambda1)}, we have
		\begin{equation*}
		\sum_{i\in\cI}c_i\lambda^{(2,i)}_1\geq\lambda_1. 
		\end{equation*}
		Let $u_1=\sum_{i\in\cI}c_i\lambda^{(1,i)}_1$ and $u_2=\sum_{i\in\cI}c_i\lambda^{(2,i)}_1$. Then we have $u_1<\lambda_1\leq u_2$. For all $i\in\cI$, let
		\begin{equation}
		\lambda^{(i)}_1=\frac{u_2-\lambda_1}{u_2-u_1}\lambda^{(1,i)}_1+\frac{\lambda_1-u_1}{u_2-u_1}\lambda^{(2,i)}_1.  \label{pf-thm-alter-region-def-lambda1i}
		\end{equation}
		It is easy to check that 
		\begin{equation}
		\sum_{i\in\cI}c_i\lambda^{(i)}_1=\lambda_1.   \label{pf-thm-alter-region-sum-lambda1}
		\end{equation}
		Let $c_i^{(1)}=\frac{u_2-\lambda_1}{u_2-u_1}\cdot c_i$ and $c_i^{(2)}=\frac{\lambda_1-u_1}{u_2-u_1}\cdot c_i$. It is readily seen that $c_i^{(1)}$ and $c_i^{(2)}$ are nonnegative, and we can check that
		\begin{equation}
		\begin{cases}
		c_i^{(1)}+c_i^{(2)}=c_i\\
		c_i^{(1)}\lambda^{(1,i)}_1+c_i^{(2)}\lambda^{(2,i)}_1=c_i\lambda^{(i)}_1.
		\end{cases}  \label{pf-thm-alter-region-coefficient(c0)}
		\end{equation}
		Then we have from \eqref{pf-thm-alter-region-f}, \eqref{pf-thm-alter-region-sum-lambda1}, and \eqref{pf-thm-alter-region-coefficient(c0)} that
		\begin{equation}
		\bm{\lambda}_{N+1}=\sum_{i\in\cI}\left(c_i^{(1)}\bm{\lambda}^{(1,i)}_{N+1}+c_i^{(2)}\bm{\lambda}^{(2,i)}_{N+1}\right).  \label{pf-thm-alter-region-case(ia)-lambda}
		\end{equation}
		Following \eqref{pf-thm-alter-region-f}, we have for $1\leq \alpha\leq \varphi+1$ that 
		\begin{align}
		&f_{\alpha}(\bm{\lambda}_{N+1}) \nonumber \\
		&=\frac{1}{\alpha}\sum_{j=1}^{N+1}\lambda_j   \nonumber \\
		&=\frac{1}{\alpha}\sum_{j=1}^{N+1}\sum_{i\in\cI}c_i\lambda^{(i)}_j  \label{pf-thm-alter-region-f-small-2} \\
		&=\frac{1}{\alpha}\sum_{j=1}^{N+1}\left[\sum_{i\in\cI}\left(c_i^{(1)}\lambda^{(1,i)}_j+c_i^{(2)}\lambda^{(2,i)}_j\right)\right]  \label{pf-thm-alter-region-f-small-3} \\
		&=\sum_{i\in\cI}\left[c_i^{(1)}\bigg(\frac{1}{\alpha} \sum_{j=1}^{N+1}\lambda^{(1,i)}_j\bigg)+c_i^{(2)}\bigg(\frac{1}{\alpha} \sum_{j=1}^{N+1}\lambda^{(2,i)}_j\bigg)\right] \nonumber \\
		&=\sum_{i\in\cI}\left[c_i^{(1)} f_{\alpha}(\bm{\lambda}^{(1,i)}_{N+1})+c_i^{(2)} f_{\alpha}(\bm{\lambda}^{(2,i)}_{N+1})\right]   \label{pf-thm-alter-region-f-small-last}
		\end{align}
		where \eqref{pf-thm-alter-region-f-small-2} follows from \eqref{pf-thm-alter-region-induction-hypo}, \eqref{pf-thm-alter-region-f-small-3} follows from \eqref{pf-thm-alter-region-def-lambda(j)} and \eqref{pf-thm-alter-region-coefficient(c0)}, and \eqref{pf-thm-alter-region-f-small-last} follows from \eqref{pf-thm-alter-region-f1} and \eqref{pf-thm-alter-region-f2}. Similarly, for $\varphi+2\leq \alpha\leq N+1$, following \eqref{pf-thm-alter-region-f}, we have 
		\begin{align}
		f_{\alpha}(\bm{\lambda}_{N+1})&=f_{\alpha-1}(\bm{\lambda}_N)  \nonumber \\
		&=\sum_{i\in\cI}c_if_{\alpha-1}(\bm{\lambda}^{(i)}_N)    \label{pf-thm-alter-region-f-large-2} \\
		&=\sum_{i\in\cI}\left(c_i^{(1)}+c_i^{(2)}\right)f_{\alpha-1}(\bm{\lambda}^{(i)}_N)  \label{pf-thm-alter-region-f-large-3} \\
		&=\sum_{i\in\cI}\left[c_i^{(1)} f_{\alpha}(\bm{\lambda}^{(1,i)}_{N+1})+c_i^{(2)} f_{\alpha}(\bm{\lambda}^{(2,i)}_{N+1})\right], \label{pf-thm-alter-region-f-large-last}
		\end{align}
		where \eqref{pf-thm-alter-region-f-large-2} follows from \eqref{pf-thm-alter-region-induction-hypo}, \eqref{pf-thm-alter-region-f-large-3} follows from \eqref{pf-thm-alter-region-coefficient(c0)}, and \eqref{pf-thm-alter-region-f-large-last} follows from \eqref{pf-thm-alter-region-f1} and \eqref{pf-thm-alter-region-f2}. In other words, \eqref{pf-thm-alter-region-f-small-last} or \eqref{pf-thm-alter-region-f-large-last} holds for all $1\leq \alpha\leq N+1$. Summarizing the above, we have
		\begin{align}
		&\big(\bm{\lambda}_{N+1},\bm{f}(\bm{\lambda}_{N+1})\big) \nonumber \\
		&=\sum_{i\in\cI}\bigg[c_i^{(1)}\left(\bm{\lambda}^{(1,i)}_{N+1},f_{\alpha}(\bm{\lambda}^{(1,i)}_{N+1})\right)+c_i^{(2)}\left(\bm{\lambda}^{(2,i)}_{N+1},f_{\alpha}(\bm{\lambda}^{(2,i)}_{N+1})\right)\bigg],  \label{pf-thm-alter-region-case(ia)}
		\end{align}
		and thus $\big(\bm{\lambda}_{N+1},\bm{f}(\bm{\lambda}_{N+1})\big)$ is a conic combination of vectors in $\cF_{N+1}^3$. 
		
		\textit{Case b):} If $\lambda_2\leq \lambda_1\leq \frac{1}{\eta'}\sum_{j=2}^{N+1}\lambda_j$, since the condition $\frac{1}{\eta'}\sum_{j=3}^{N+1}\lambda_j\leq \lambda_2$ in \eqref{pf-alter-region-assumption-lambda2} is equivalent to
		\begin{equation*}
		\lambda_2\geq \frac{1}{\eta'+1}\sum_{j=2}^{N+1}\lambda_j,  
		\end{equation*}
		we have
		\begin{equation*}
		\frac{1}{\eta'+1}\sum_{j=2}^{N+1}\lambda_j\leq \lambda_1\leq \frac{1}{\eta'}\sum_{j=2}^{N+1}\lambda_j.
		\end{equation*}
		By \Cref{lemma-f-lambda1}, we obtain 
		\begin{equation*}
		f_{\alpha}(\bm{\lambda}_{N+1})=
		\begin{cases}
		\frac{1}{\alpha}\sum_{j=1}^{N+1}\lambda_j,&\text{for }1\leq \alpha\leq \eta'+1 \\
		f_{\alpha-1}(\bm{\lambda}_N),&\text{for }\eta'+2\leq \alpha\leq N+1.
		\end{cases}   
		\end{equation*}
		In light of \eqref{pf-alter-region-lambda(i2)-range}, let $\cI_1=\left\{i\in\cI:\lambda^{(i)}_2=\frac{1}{\eta'-1}\sum_{j=3}^{N+1}\lambda^{(i)}_j\right\}$ and $\cI_2=\left\{i\in\cI:\lambda^{(i)}_2=\frac{1}{\eta'}\sum_{j=3}^{N+1}\lambda^{(i)}_j\right\}$, where $\cI_1\cup\cI_2=\cI$. For~$i\in\cI_1$, let 
		\begin{equation}
		\lambda^{(1,i)}_1=\lambda^{(2,i)}_1=\frac{1}{\eta'}\sum_{j=2}^{N+1}\lambda^{(i)}_j.   \label{pf-thm-alter-region-def-lambda(i)-b}
		\end{equation}
		For $i\in\cI_2$, let 
		\begin{equation}
		\lambda^{(1,i)}_1=\frac{1}{\eta'+1}\sum_{j=2}^{N+1}\lambda^{(i)}_j  \label{pf-thm-alter-region-def-lambda(1,i)-b}
		\end{equation}
		and 
		\begin{equation*}
		\lambda^{(2,i)}_1=\frac{1}{\eta'}\sum_{j=2}^{N+1}\lambda^{(i)}_j.  
		\end{equation*}
		Again, from \eqref{pf-alter-region-lambda(i2)-range} and \eqref{def-cG_L}, we can check that $\bm{\lambda}^{(1,i)}_{N+1},\bm{\lambda}^{(2,i)}_{N+1}\in\cG_{N+1}^0$ for all $i\in\cI$. By \Cref{remark-lemma-f-lambda1} following \Cref{lemma-f-lambda1}, we have for $i\in\cI_1$ that 
		\begin{align*}
		&f_{\alpha}(\bm{\lambda}^{(1,i)}_{N+1})=f_{\alpha}(\bm{\lambda}^{(2,i)}_{N+1})  \\
		&=\begin{cases}
		\frac{1}{\alpha}\sum_{j=1}^{N+1}\lambda^{(1,i)}_j,& \text{for }1\leq \alpha\leq \eta'+1 \\
		f_{\alpha-1}(\bm{\lambda}^{(i)}_N),&\text{for }\eta'+1\leq \alpha\leq N+1,
		\end{cases}
		\end{align*}
		and for $i\in\cI_2$,
		\begin{equation*}
		f_{\alpha}(\bm{\lambda}^{(1,i)}_{N+1})=
		\begin{cases}
		\frac{1}{\alpha}\sum_{j=1}^{N+1}\lambda^{(2,i)}_j,& \text{for }1\leq \alpha\leq \eta'+2 \\
		f_{\alpha-1}(\bm{\lambda}^{(i)}_N),&\text{for }\eta'+2\leq \alpha\leq N+1,
		\end{cases} 
		\end{equation*}
		and
		\begin{equation*}
		f_{\alpha}(\bm{\lambda}^{(2,i)}_{N+1})=
		\begin{cases}
		\frac{1}{\alpha}\sum_{j=1}^{N+1}\lambda^{(2,i)}_j,& \text{for }1\leq \alpha\leq \eta'+1 \\
		f_{\alpha-1}(\bm{\lambda}^{(i)}_N),&\text{for }\eta'+1\leq \alpha\leq N+1.
		\end{cases} 
		\end{equation*}
		Following from \eqref{pf-thm-alter-region-def-lambda(i)-b} and \eqref{pf-thm-alter-region-def-lambda(1,i)-b}, we have 
		\begin{align}
		\sum_{i\in\cI}c_i\lambda^{(1,i)}_1&=\sum_{i\in\cI_1}c_i\frac{1}{\eta'}\sum_{j=2}^{N+1}\lambda^{(i)}_j+\sum_{i\in\cI_2}c_i\frac{1}{\eta'+1}\sum_{j=2}^{N+1}\lambda^{(i)}_j  \nonumber \\
		&=\sum_{i\in\cI_1}c_i\frac{1}{\eta'}\left(1+\frac{1}{\eta'-1}\right)\sum_{j=3}^{N+1}\lambda^{(i)}_j  \nonumber \\
		&\hspace{0.5cm}+\sum_{i\in\cI_2}c_i\frac{1}{\eta'+1}\left(1+\frac{1}{\eta'}\right)\sum_{j=3}^{N+1}\lambda^{(i)}_j   \label{pf-alter-region-lambda1i-b-1} \\
		&=\sum_{i\in\cI_1}c_i\frac{1}{\eta'-1}\sum_{j=3}^{N+1}\lambda^{(i)}_j+\sum_{i\in\cI_2}c_i\frac{1}{\eta'}\sum_{j=3}^{N+1}\lambda^{(i)}_j  \nonumber  \\
		&=\sum_{i\in\cI_1}c_i\lambda^{(i)}_2+\sum_{i\in\cI_2}c_i\lambda^{(i)}_2   \label{pf-alter-region-lambda1i-b-3} \\
		&=\lambda_2  \nonumber \\
		&\leq \lambda_1,   \nonumber 
		\end{align}
		where \eqref{pf-alter-region-lambda1i-b-1} and \eqref{pf-alter-region-lambda1i-b-3} follow from the definition of $\cI_1$ and $\cI_2$. Similar to \eqref{pf-alter-region-lambda1i-1}-\eqref{pf-alter-region-lambda1i-last}, we have
		\begin{equation*}
		\sum_{i\in\cI}c_i\lambda^{(2,i)}_1\geq \lambda_1.
		\end{equation*}
		For $i\in\cI$, similar to \eqref{pf-thm-alter-region-def-lambda1i}-\eqref{pf-thm-alter-region-coefficient(c0)}, let $c_i^{(1)}=\frac{u_2-\lambda_1}{u_2-u_1}\cdot c_i$, $c_i^{(2)}=\frac{\lambda_1-u_1}{u_2-u_1}\cdot c_i$, and 
		\begin{equation*}
		\lambda^{(i)}_1=\frac{u_2-\lambda_1}{u_2-u_1}\lambda^{(1,i)}_1+\frac{\lambda_1-u_1}{u_2-u_1}\lambda^{(2,i)}_1.
		\end{equation*}
		We can check that 
		\begin{equation*}
		\sum_{i\in\cI}c_i\lambda^{(i)}_1=\lambda_1
		\end{equation*}
		and for all $i\in\cI$,
		\begin{equation*}
		\begin{cases}
		c_i^{(1)}+c_i^{(2)}=c_i\\
		c_i^{(1)}\lambda^{(1,i)}_1+c_i^{(2)}\lambda^{(2,i)}_1=c_i\lambda^{(i)}_1.
		\end{cases} 
		\end{equation*}
		Then similar to \eqref{pf-thm-alter-region-case(ia)-lambda}-\eqref{pf-thm-alter-region-case(ia)}, we have 
		\begin{align*}
		&\big(\bm{\lambda}_{N+1},\bm{f}(\bm{\lambda}_{N+1})\big)   \\
		&=\sum_{i\in\cI}\bigg[c_i^{(1)}\left(\bm{\lambda}^{(1,i)}_{N+1},f_{\alpha}(\bm{\lambda}^{(1,i)}_{N+1})\right)+c_i^{(2)}\left(\bm{\lambda}^{(2,i)}_{N+1},f_{\alpha}(\bm{\lambda}^{(2,i)}_{N+1})\right)\bigg].  
		\end{align*}
		
		\textit{Case c):} If $\lambda_1>\sum_{j=2}^{N+1}\lambda_j$, let $\bm{\lambda}'_{N+1}=\left(\sum_{j=2}^{N+1}\lambda_j,\lambda_2,\cdots,\lambda_{N+1}\right)$ and $\bm{\lambda}^{(1)}_{N+1}$ be the $(N+1)$-vector with the first component being 1 and the rest being 0. From \Cref{lemma-considerable-lambda1}, we have 
		\begin{align*}
		&(\bm{\lambda}_{N+1},\bm{f}(\bm{\lambda}_{N+1}))  \\
		&=\bigg(\lambda_1-\sum_{j=2}^{N+1}\lambda_j\bigg)\left(\bm{\lambda}^{(1)}_{N+1},\bm{f}(\bm{\lambda}^{(1)}_{N+1})\right)+\big(\bm{\lambda}'_{N+1},\bm{f}(\bm{\lambda}'_{N+1})\big).  
		\end{align*}
		It is easy to see that 
		\begin{equation*}
		\left(\bm{\lambda}^{(1)}_{N+1},\bm{f}(\bm{\lambda}^{(1)}_{N+1})\right)\in\cF_{N+1}^3.  
		\end{equation*}
		Note that $\bm{\lambda}'_{N+1}$ satisfies the condition for Case a) provided that $\eta'\neq1$. Otherwise, it satisfies the condition for Case b). Thus we see that $\big(\bm{\lambda}'_{N+1},\bm{f}(\bm{\lambda}'_{N+1})\big)$ is always a conic combination of the vectors in $\cF_{N+1}^3$. This implies that $(\bm{\lambda}_{N+1},\bm{f}(\bm{\lambda}_{N+1}))$ is a conic combination of the vectors in $\cF_{N+1}^3$, as is to be proved. 
	\end{proof}
	
	\medskip
	For any $\bm{\lambda}\in\mathbb{R}_+^L$, let $\bm{\lambda}_{L-1}=(\lambda_2,\lambda_3,\cdots,\lambda_L)$ and $\bm{f}(\bm{\lambda}_{L-1})=\big(f_1(\bm{\lambda}_{L-1}), f_2(\bm{\lambda}_{L-1}),\cdots,f_{L-1}(\bm{\lambda}_{L-1})\big)$. The following lemma provides a method for finding a set of conic combination coefficients for $\big(\bm{\lambda}_{L-1},\bm{f}(\bm{\lambda}_{L-1})\big)$ from the conic combination for $\big(\bm{\lambda},\bm{f}(\bm{\lambda})\big)$.
	\begin{lemma}  \label{lemma-combination-reduce-dim}
		Consider any ordered vector $\bm{\lambda}\in\mathbb{R}_+^L$ such that $\bm{\lambda}\neq\bm{\lambda}^{[1]}$. Assume there exists $c_i\geq 0,~i=1,2,\cdots,S_L$ such that 
		\begin{equation}
		\big(\bm{\lambda},\bm{f}(\bm{\lambda})\big)=\sum_{i=1}^{S_L}c_i\cdot\big(\bm{\lambda}^{(i)},\bm{f}(\bm{\lambda}^{(i)})\big),  \label{assumption-lemma-combination-reduce-dim}
		\end{equation}
		Then we have
		\begin{align*}
		&\big(\bm{\lambda}_{L-1},\bm{f}(\bm{\lambda}_{L-1})\big)  \\
		&=\sum_{i=1}^{S_L}c_i\cdot \big((\lambda^{(i)}_2,\lambda^{(i)}_3,\cdots,\lambda^{(i)}_L),\bm{f}(\lambda^{(i)}_2,\lambda^{(i)}_3,\cdots,\lambda^{(i)}_L)\big).
		\end{align*}
	\end{lemma}
	\begin{proof}
		See Appendix \ref{section-proof-combination-reduce-dim}.
	\end{proof}
	
	\begin{lemma}\label{lemma-lambda(L-1)}
		For any $\bm{\lambda}^{(i)}\in\cG_L$, if $\left(\lambda^{(i)}_2,\lambda^{(i)}_3,\cdots,\lambda^{(i)}_L\right)\in\cG_{L-1}$, then $\lambda^{(i)}_1=0$ or $\lambda^{(i)}_{\pi_i(1)}$.
	\end{lemma}
	\begin{proof}
		See Appendix \ref{section-proof-lemma-lambda(L-1)}.
	\end{proof}
	
	\begin{lemma}\label{lemma-no-redundancy}
		For any $i_0\in\{1,2,\cdots,S_L\}$, there does not exist $(c_1,c_2,\cdots,c_{S_L})\in\mathbb{R}_+^{S_L}$ such that $c_{i_0}=0$ and 
		\begin{equation*}
		\big(\bm{\lambda}^{(i_0)},\bm{f}(\bm{\lambda}^{(i_0)})\big)=\sum_{i=1}^{S_L}c_i \cdot\big(\bm{\lambda}^{(i)},\bm{f}(\bm{\lambda}^{(i)})\big).
		\end{equation*}
	\end{lemma}
	\begin{proof}
		See Appendix \ref{section-proof-lemma-no-redundancy}.
	\end{proof}
	
	\medskip
	\Cref{thm-alter-region} gives a rate region $\cR^*$ that simplifies the characterization of the superposition region. The following theorem shows that there is no redundancy in the specification of $\cR^*$.
	\begin{theorem}\label{thm-no-redundancy}
		For the inequalities specifying $\cR^*$ in \eqref{def-cR*}, none of them is implied by the others.
	\end{theorem}
	\begin{proof}
		For any $i_0\in\{1,2,\cdots,S_L\}$, consider the following linear program:
		\begin{align*}
		m_p=\text{min }&\sum_{l=1}^L\lambda_l^{(i_0)}R_l   \\
		\text{s.t. }&\sum_{l=1}^L\lambda_l^{(i)}R_l\geq \sum_{\alpha=1}^Lf_\alpha(\bm{\lambda}^{(i)})H(X_\alpha),  \\
		&\hspace{0.8cm}\forall~1\leq i\leq S_L,i\neq i_0. 
		\end{align*}
		To prove \Cref{thm-no-redundancy}, it suffices to show the following: for any $i_0\in\{1,2,\cdots,S_L\}$,
		\begin{equation*}
		\sum_{\alpha=1}^Lf(\bm{\lambda}^{(i_0)})H(X_\alpha)>m_p. 
		\end{equation*}
		By strong duality, $m_p=m_d$, where $m_d$ is the optimal value of the dual problem 
		\begin{align}
			m_d=\text{max }&\sum_{1\leq i\leq S_L,i\neq i_0}c_i\left(\sum_{\alpha=1}^Lf_{\alpha}(\bm{\lambda}^{(i)})H(X_{\alpha})\right)   \nonumber \\
			\text{s.t. }&\sum_{1\leq i\leq S_L,i\neq i_0} c_i\lambda_l^{(i)}=\lambda_l^{(i_0)}, \forall~l\in\cL   \label{no-redundancy-dual-cond}\\
							&c_i\geq 0, 1\leq i\leq S_L,i\neq i_0. \nonumber
		\end{align}
		Then it suffices to show that for any $i_0\in\{1,2,\cdots,S_L\}$,
		\begin{equation}
		\sum_{\alpha=1}^Lf(\bm{\lambda}^{(i_0)})H(X_\alpha)>m_d, \label{no-redundancy-larger-md}
		\end{equation}
		for all possible values of $H(X_{\alpha})$, $\alpha\in\cL$. 
		For notational simplicity, let $c_{i_0}=0$. 
		By \Cref{lemma-concave-f}, \eqref{no-redundancy-dual-cond} implies that 
		\begin{equation}
		f_{\alpha}(\bm{\lambda}^{(i_0)})\geq \sum_{i=1}^{S_L}c_i f_{\alpha}(\bm{\lambda}^{(i)}), \text{ for all }\alpha\in\cL.   \label{no-redundancy->=}
		\end{equation}
		Upon multiplying by $H(X_{\alpha})$ and summing over all $\alpha\in\cL$, we obtain
		\begin{equation*}
		\sum_{\alpha=1}^{L}f_{\alpha}(\bm{\lambda}^{(i_0)})H(X_{\alpha})\geq \sum_{\alpha=1}^{L}\left(\sum_{i=1}^{S_L}c_if_{\alpha}(\bm{\lambda}^{(i)})\right)H(X_{\alpha}),  
		\end{equation*}
		which is equivalent to \eqref{no-redundancy-larger-md} except that the inequality above is nonstrict. 
		Thus to prove \eqref{no-redundancy-larger-md}, we only need to show that there exists at least one $\alpha\in\cL$ such that 
		\begin{equation*}
		f_{\alpha}(\bm{\lambda}^{(i_0)})>\sum_{i=1}^{S_L}c_i f_{\alpha}(\bm{\lambda}^{(i)}).  
		\end{equation*}
		Assume the contrary is true, i.e. equality holds in \eqref{no-redundancy->=} for all $\alpha\in\cL$. Then this implies
		\begin{equation}
		\big(\bm{\lambda}^{(i_0)},\bm{f}(\bm{\lambda}^{(i_0)})\big)=\sum_{i=1}^{S_L}c_i \cdot\big(\bm{\lambda}^{(i)},\bm{f}(\bm{\lambda}^{(i)})\big),
		\end{equation}
		which is a contradiction to \Cref{lemma-no-redundancy}. This completes the proof of the theorem.
	\end{proof}

	\section{Checking the Achievability of A Rate Tuple}  \label{section-check-achievability}
	\subsection{Checking Achievability} \label{section-check-achievability-1}
	
	Given the superposition coding rate region $\cR_{\sup}$ characterized by the constraints in \eqref{def-R-sup1} and \eqref{def-R-sup2}, 
	it is readily seen that a rate tuple is achievable if and only if there exist nonnegative 
	variables $r_i^{\alpha}~(i,\alpha\in\cL)$ satisfying \eqref{def-R-sup1} and \eqref{def-R-sup2}. 
	Thus, we can check the achievability of a given rate tuple $\mathbf{R}$ by determining 
	whether there exists a set of {\it feasible solutions} $r_i^{\alpha}~(i,\alpha\in\cL)$ for the optimization problem: 
	\begin{eqnarray}
	&\min& 0    \nonumber \\
	&\text{s.t.}&\sum_{\alpha=1}^{L}r_i^{\alpha}=R_i,~~\forall i\in\cL   \label{LP-condition-1} \\
	&&\sum_{i\in\cU}r_i^{\alpha}\geq H(X_{\alpha}),~~\forall\cU\subseteq\cL,~|\cU|=\alpha  \nonumber  \\
	&&r_i^{\alpha}\geq 0,~1\leq i,\alpha\leq L.   \label{LP-condition-3} 
	\end{eqnarray}
	This can be easily achieved through the MATLAB linear programming function: $$\text{x = linprog(f,A,b,Aeq,beq)}.$$
	We have run numerical tests of the ``linprog" function on a notebook computer to determine the achievability of a given rate tuple for $L\leq 20$. 
	For $L=21$, the program runs out of memory, since the size of the constraint matrices ``A" and ``Aeq" become prohibitively large. 
	
	It is also natural to check the achievability of a given rate tuple by verifying the inequalities specifying $\cR_{\sup}$ in \Cref{thm-alter-region}. 
	Even though by \Cref{thm-no-redundancy} there is no redundancy in the set of inequalities in \eqref{def-cR*} that specifies $\cR^*$, 
	by taking advantage of the symmetry of the problem, we in fact do not need to check all these inequalities. 
	The next lemma identifies the subset of these inequalities we need to check. 
	As we will see, the number of such inequalities is significantly smaller than the total number of inequalities specifying~$\cR^*$. 
	
	For any $\bm{R}\in\mathbb{R}^L_+$ and any permutation $\omega$ on $\{1,2,\cdots,L\}$, similar to the definition of $\omega(\bm{\lambda})$ at the beginning of Section III, let  
	\begin{equation*}
	\omega(\bm{R})=\left(R_{\omega(1)},R_{\omega(2)},\cdots,R_{\omega(L)}\right).
	\end{equation*}
	Due to the symmetry of the problem, for any $\bm{\lambda}\in\cG_L$, the inequality
	\begin{equation*}
	\sum_{l=1}^L\lambda_lR_l\geq \sum_{\alpha=1}^L f_{\alpha}(\bm{\lambda})H(X_{\alpha})
	\end{equation*}
	implies the inequality
	\begin{equation*}
	\sum_{l=1}^L\lambda_{\omega(l)}R_{\omega(l)}\geq \sum_{\alpha=1}^L f_{\alpha}\big(\omega(\bm{\lambda})\big)H(X_{\alpha}),
	\end{equation*}
	and vice versa. Thus, $\bm{R}$ is achievable if and only if $\omega(\bm{R})$ is achievable for all $\omega$. As such, we only need to consider rate tuples $\bm{R}\in\mathbb{R}^L_+$ such that
	\begin{equation}
	R_1\leq R_2\leq \cdots\leq R_L.  \label{def-order-R}
	\end{equation}
	
	\begin{lemma}  \label{lemma-min-sum-multiplication}
		For any nonnegative rate tuple $\bm{R}$ such that \eqref{def-order-R} is satisfied and any $\bm{\lambda}\in\mathbb{R}_+^L$, we have
		\begin{equation*}
		\sum_{i=1}^L\lambda_iR_i\geq \sum_{i=1}^L\lambda_{\pi(i)}R_i. 
		\end{equation*}
	\end{lemma}
	\begin{remark}
		The inequality in \Cref{lemma-min-sum-multiplication} is sometimes called the {\it rearrangement inequality} \cite[Chapter 5]{rearrange-inequality}.
	\end{remark}
	\begin{proof}
		See Appendix \ref{section-proof-lemma-min-sum-multiplication}.
	\end{proof}
	
	\medskip
	From \Cref{lemma-indep-order}, we can see that RHS of the inequality in \eqref{def-cR*} does not change with $\bm{\lambda}$ replaced by $\pi(\bm{\lambda})$. Thus, in order to check the achievability of a given rate tuple (assume satisfying \eqref{def-order-R}), by \Cref{lemma-min-sum-multiplication}, we only need to check those inequalities for which the coefficients are in descending order, i.e.
	\begin{equation*}
	\lambda_1\geq \lambda_2\geq \cdots\geq\lambda_L.
	\end{equation*}
	All the other inequalities are redundant for this rate tuple. Then, the number of inequalities we need to check is only $S_L^0$ (cf. \eqref{def-S_L-0}), which is bounded in the following theorem.
	
	\begin{theorem} \label{thm-polynomial-time}
		$2^{L-1}\leq S_L^0\leq L!$.
	\end{theorem}
	\begin{remark}
		We will see from the proof that both the above inequalities become strict for $L\geq 3$. 
	\end{remark}
	\begin{remark}
		On a notebook computer, it took about 8 days to list all the $S_L^0$ inequalities for all $L\leq 15$. For $L=16$, the computation involved appears to be prohibitive.
	\end{remark}
	\begin{proof}
		We can see from Appendix \ref{section-table} that $S_1^0=1$ and $S_2^0=2$. It is easy to check that the theorem is true for $L=1$ and $L=2$.
		
		For $L\geq 2$, let $\bm{\lambda}=(\lambda_1,\lambda_2, \cdots,\lambda_L)$ and $\cG_L^*=\left\{\bm{\lambda}\in\cG_{L}^0:\lambda_L=0\right\}$. For any $\bm{\lambda}\in\cG_L^*$, it is easy to check from \eqref{def-cG_L} that $\left(\lambda_1,\lambda_2,\cdots,\lambda_{L-1}\right)\in\cG_{L-1}^0$. On the other hand, for any $\left(\lambda_1, \lambda_2\right.$, $\left.\cdots, \lambda_{L-1}\right)\in\cG_{L-1}^0$, we have $\left(\lambda_1,\lambda_2,\cdots,\lambda_{L-1},0\right)\in\cG_L^*$. Thus, there is a one-to-one correspondence between $\cG_L^*$ and $\cG_{L-1}^0$, which implies that
		\begin{equation}
		|\cG_L^*|=S_{L-1}^0.   \label{pf-thm-polynomial-time-cardinality-1}
		\end{equation}
		For $k\geq2$, let $D_k=|\cG_k^0\backslash\cG_k^*|$. By \eqref{pf-thm-polynomial-time-cardinality-1} and the fact that $\cG_k^*\subseteq \cG_k^0$, we have
		\begin{equation*}
		D_k=S_k^0-S_{k-1}^0,
		\end{equation*}
		which implies that
		\begin{equation*}
		S_L^0=S_1^0+\sum_{k=2}^L D_k.   \label{pf-thm-polynomial-S(L)-1}
		\end{equation*}
		Now, we only need to calculate $D_k$ for $k\geq 2$. For any $\left(\lambda_{L-k+1},\lambda_{L-k+2},\cdots,\lambda_L\right)\in\cG_k^0\backslash\cG_k^*$, where $\lambda_L=1$, we can see from \eqref{def-cG_L} that $\left(\lambda_{L-k+2},\lambda_{L-k+3},\cdots,\lambda_L\right)\in\cG_{k-1}^0\backslash\cG_{k-1}^*$ by construction. Thus, all $\left(\lambda_{L-k+1},\lambda_{L-k+2},\cdots,\lambda_L\right)\in\cG_k^0\backslash\cG_k^*$ can be generated from $\left(\lambda_{L-k+2},\lambda_{L-k+3},\cdots,\lambda_L\right)\in\cG_{k-1}^0\backslash\cG_{k-1}^*$ with a proper choice of $\lambda_{L-k+1}$. Since $\lambda_L=1$, we have $\zeta=L$. Recall from \eqref{def-cG_L-lambda(j+1)} that $\theta_{L}=0$ and for $j=L-1,L-2,\cdots,L-k+1$, $\theta_{j}$ is the integer such that
		\begin{equation}
		\lambda_{j}=\frac{1}{\theta_{j}}\sum_{i=j+1}^L\lambda_{i}.  \label{pf-thm-polynomial-def-theta(j)}
		\end{equation}
		According to \eqref{pf-thm-polynomial-def-theta(j)}, the $k$-vector $\left(\lambda_{L-k+1},\lambda_{L-k+2}, \cdots,\lambda_L\right)\in\cG_{k}^0\backslash\cG_{k}^*$ is uniquely determined by the tuple $(\theta_{L-k+1},$ $\cdots,\theta_{L-1},\theta_{L})$. Thus $D_k$ is equal to the cardinality of the set
		\begin{align*}
		\Theta_k=&\big\{(\theta_{L-k+1},\cdots,\theta_{L-1},\theta_{L}):1\leq \theta_{j}\leq \theta_{j+1}+1  \\
		&\quad \text{ for }j=L-1,L-2,\cdots,L-k+1\big\}.  
		\end{align*}
		By straightforward counting, we can obtain
		\begin{equation}
		D_k=|\Theta_k|=\sum_{\theta_L=0}^{0}~\sum_{\theta_{L-1}=1}^{\theta_L+1}~\sum_{\theta_{L-2}=1}^{\theta_{L-1}+1}\cdots \sum_{\theta_{L-k+1}=1}^{\theta_{L-k+2}+1}~1. \label{pf-thm-polynomial-D(k)-induction}
		\end{equation}
		
		Now we bound $D_k$ according to \eqref{pf-thm-polynomial-D(k)-induction}. Observe that $\theta_L=0$ and $\theta_{L-1}=1$ always hold. Then for $k\geq 3$, \eqref{pf-thm-polynomial-D(k)-induction} can be rewritten as 
		\begin{equation*}
		D_k=\sum_{\theta_{L-2}=1}^{2}~\sum_{\theta_{L-3}=1}^{\theta_{L-2}+1}\cdots \sum_{\theta_{L-k+1}=1}^{\theta_{L-k+2}+1}~1.
		\end{equation*}
		Let 
		\begin{equation*}
		D_k^{(1)}=\sum_{\theta_{L-2}=1}^{2}~\sum_{\theta_{L-3}=1}^{2}\cdots \sum_{\theta_{L-k+1}=1}^{2}~1  
		\end{equation*}
		and 
		\begin{equation*}
		D_k^{(2)}=\sum_{\theta_{L-2}=1}^{2}~\sum_{\theta_{L-3}=1}^{3}\cdots \sum_{\theta_{L-k+1}=1}^{k-1}~1.   
		\end{equation*}
		From \eqref{range-theta(j)-2}, it is easy to check that 
		\begin{equation}
		D_k^{(1)}\leq D_k\leq D_k^{(2)},  \label{pf-thm-polynomial-D(k)}
		\end{equation}
		and we have
		\begin{equation*}
		D_k^{(1)}=2^{k-2}
		\end{equation*}
		and 
		\begin{equation*}
		D_k^{(2)}=(k-1)!~.
		\end{equation*}
		Thus, for $L\geq 3$, we have
		\begin{equation*}
		\sum_{k=3}^L D_k^{(1)}=\sum_{k=3}^L 2^{k-2}=2^{L-1}-2
		\end{equation*}
		and
		\begin{eqnarray}
		\sum_{k=3}^L D_k^{(2)}&=&\sum_{k=3}^L (k-1)!  \nonumber \\
		&\leq &(L-1)!\times (L-2)  \nonumber \\
		&\leq &L!-2.  \nonumber 
		\end{eqnarray}
		Then, by \eqref{pf-thm-polynomial-S(L)-1}, \eqref{pf-thm-polynomial-D(k)}, and the fact that $S_1^0=1$ and $D_2=1$, we have for $L\geq 3$ that
		\begin{equation*}
		2^{L-1}\leq S_L^0\leq L!~.
		\end{equation*}
		This proves the theorem.
	\end{proof}
	
	\subsection{Comparison of complexity}  \label{section-check-achievability-2}
	In this section, we compare the complexity of checking the achievability of a given rate tuple through different methods described in the last section. 
	In \Cref{fig-log-compare}, we compare the program running time of checking the LP feasibility using the MATLAB ``linprog" function
	and that of checking the achievability of a rate tuple through the inequalities with ordered coefficients. 
	\begin{figure}[!h]
		\centering
		\begin{tikzpicture}[scale=0.25,font=\scriptsize]
		\coordinate (P3) at (3, -4.059);
		\coordinate (P4) at (4, -4.059);
		\coordinate (P5) at (5, -4.059);
		\coordinate (P6) at (6, -4.059);
		\coordinate (P7) at (7, -3.837);
		\coordinate (P8) at (8, -2.737);
		\coordinate (P9) at (9, -1.515);
		\coordinate (P10) at (10, 0.536);
		\coordinate (P11) at (11, 4.322);
		\coordinate (P12) at (12, 8.205);
		\coordinate (P13) at (13, 11.983);
		\coordinate (P14) at (14, 15.757);
		\coordinate (P15) at (15, 19.556);
		
		\coordinate (Q3) at (3, -4.059);
		\coordinate (Q4) at (4, -3.837);
		\coordinate (Q5) at (5, -3.837);
		\coordinate (Q6) at (6, -3.644);
		\coordinate (Q7) at (7, -3.644);
		\coordinate (Q8) at (8, -2.737);
		\coordinate (Q9) at (9, -2.556);
		\coordinate (Q10) at (10, -2.474);
		\coordinate (Q11) at (11, -2.322);
		\coordinate (Q12) at (12, -1.690);
		\coordinate (Q13) at (13, -0.971);
		\coordinate (Q14) at (14, 0.214);
		\coordinate (Q15) at (15, 1.379);
		\coordinate (Q16) at (16, 2.392);
		\coordinate (Q17) at (17, 3.858);
		\coordinate (Q18) at (18, 5.585);
		\coordinate (Q19) at (19, 10.073);
		\coordinate (Q20) at (20, 12.295);
		
		\foreach \i in {3,...,15}
		\fill [opacity=.8] (P\i) circle (3.5pt);
		\draw[blue] (P3)--(P4)--(P5)--(P6)--(P7)--(P8)--(P9)--(P10)--(P11)--(P12)--(P13)--(P14)--(P15); 
		
		\foreach \i in {3,...,20}
		\fill [opacity=.8] (Q\i) circle (3.5pt);
		\draw[orange] (Q3)--(Q4)--(Q5)--(Q6)--(Q7)--(Q8)--(Q9)--(Q10)--(Q11)--(Q12)--(Q13)--(Q14)--(Q15)--(Q16)--(Q17)--(Q18)--(Q19)--(Q20);
		
		\foreach \x [count=\i from 3] in {3,5,...,19} {\draw [opacity=0.1] (\x,-5) -- (\x,-4.8); \node[below] at (\x,-5) {\x};}
		
		\node [right] at (8,17) {\color{blue}{$\log_2{T_1}$: via inequalities}};
		\node [right] at (13,8) {\color{orange}{$\log_2{T_2}$: via LP feasibility}};
		
		\draw [->,opacity=0.1] (0,-5)--(20,-5);
		\draw [->,opacity=0.1] (0,-5)--(0,20);
		\node [right] at (20,-5) {$L$};
		\node [right] at (0,20) {$\log_2 T$};
		\end{tikzpicture}
		\caption{\footnotesize logarithm of running time}
		\label{fig-log-compare}
	\end{figure}
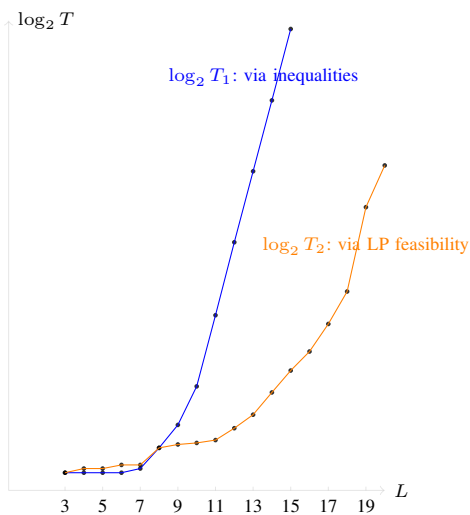
	We can see from the figure that checking the LP feasibility uses much less time than checking the achievability of a rate tuple through the inequalities with ordered coefficients. 
	We also observe that the running time of checking the inequalities with ordered coefficients grows exponentially with $L$  for $L\geq 10$, 
	even though it was shown in \Cref{thm-polynomial-time} that the number of such inequalities may grow at a rate higher than exponential in $L$. 
	
	The time of listing the inequalities with ordered coefficients and the time of constructing 
	the parameters of ``linprog" are involved in the comparison in \Cref{fig-log-compare}. 
	If we want to check the achievability of a large number of rate tuples, the more efficient way is to 
	save the inequalities with ordered coefficients and the parameters of ``linprog" in advance. 
	Then we can use the ``load" function in MATLAB to invoke these data when checking the achievability of rate tuples. 
	In this case, we start counting the program running time right after the ``load" function, and we call this the {\it pure running time}.
	In \Cref{fig-log-compare-pure}, we compare the pure running time of checking the LP feasibility using ``linprog" function
	and that of checking the achievability of a rate tuple through the inequalities with ordered coefficients. 
	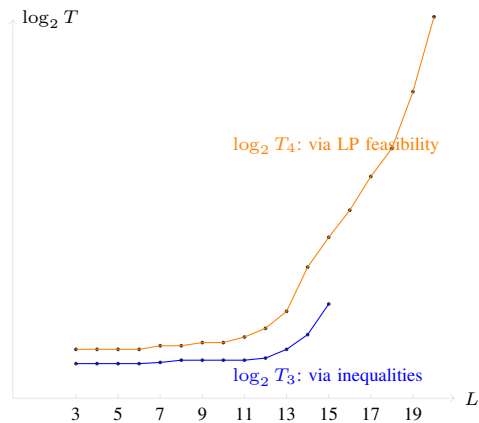
\begin{figure}[!h]
		\centering
		\begin{tikzpicture}[scale=0.28,font=\scriptsize]
		\coordinate (P3') at (3, 4.644);
		\coordinate (P4') at (4, 4.644);
		\coordinate (P5') at (5, 4.644);
		\coordinate (P6') at (6, 4.644);
		\coordinate (P7') at (7, 4.700);
		\coordinate (P8') at (8, 4.807);
		\coordinate (P9') at (9, 4.807);
		\coordinate (P10') at (10, 4.807);
		\coordinate (P11') at (11, 4.807);
		\coordinate (P12') at (12, 4.907);
		\coordinate (P13') at (13, 5.322);
		\coordinate (P14') at (14, 6.022);
		\coordinate (P15') at (15, 7.476);

		\coordinate (Q3') at (3, 5.322);
		\coordinate (Q4') at (4, 5.322);
		\coordinate (Q5') at (5, 5.322);
		\coordinate (Q6') at (6, 5.322);
		\coordinate (Q7') at (7, 5.492);
		\coordinate (Q8') at (8, 5.492);
		\coordinate (Q9') at (9, 5.644);
		\coordinate (Q10') at (10, 5.644);
		\coordinate (Q11') at (11, 5.907);
		\coordinate (Q12') at (12, 6.322);
		\coordinate (Q13') at (13, 7.129);
		\coordinate (Q14') at (14, 9.229);
		\coordinate (Q15') at (15, 10.644);
		\coordinate (Q16') at (16, 11.933);
		\coordinate (Q17') at (17, 13.535);
		\coordinate (Q18') at (18, 14.869);
		\coordinate (Q19') at (19, 17.566);
		\coordinate (Q20') at (20, 21.121);
		
		\foreach \i in {3,...,15}
		\fill [opacity=.8] (P\i') circle (2.5pt);
		\draw[blue] (P3')--(P4')--(P5')--(P6')--(P7')--(P8')--(P9')--(P10')--(P11')--(P12')--(P13')--(P14')--(P15');
		
		\foreach \i in {3,...,20}
		\fill [opacity=.8] (Q\i') circle (2.5pt);
		\draw[orange] (Q3')--(Q4')--(Q5')--(Q6')--(Q7')--(Q8')--(Q9')--(Q10')--(Q11')--(Q12')--(Q13')--(Q14')--(Q15')--(Q16')--(Q17')--(Q18')--(Q19')--(Q20');
		
		\foreach \x [count=\i from 3] in {3,5,...,19} {\draw [opacity=0.1] (\x,3) -- (\x,3.2); \node[below] at (\x,3) {\x};}
		
		\node [right] at (10,4) {\color{blue}{$\log_2{T_3}$: via inequalities}};
		\node [right] at (10,15) {\color{orange}{$\log_2{T_4}$: via LP feasibility}};
		
		\draw [->,opacity=0.1] (0,3)--(21,3);
		\draw [->,opacity=0.1] (0,3)--(0,21);
		\node [right] at (21,3) {$L$};
		\node [right] at (0,21) {$\log_2 T$};
		\end{tikzpicture}
		\caption{\footnotesize logarithm of pure running time}
		\label{fig-log-compare-pure}
	\end{figure}
	We see from the figure that checking the achievability of a rate tuple through inequalities in turn uses much less time than checking the LP feasibility. 
	This is not surprising because the time for checking the achievability through inequalities is mainly spent on listing these inequalities. 
	
	The bottleneck of checking the LP feasibility lies in that the parameters of the ``linprog" function 
	(i.e., the coefficients of the LP conditions \eqref{LP-condition-1}-\eqref{LP-condition-3}) use a mass of memory, 
	which exceeds the capacity of the hard disk on the notebook computer for $L\geq 21$. 
	The bottleneck of checking the achievability through the inequalities with ordered coefficients lies in that 
	the complexity of listing these inequalities grows exponentially with $L$, which becomes unmanageable for $L\geq 15$. 
	
	\section{Subset Entropy Inequality}  \label{section-subset-inequality}
	In \cite{yeung99}, the proof of the optimality of superposition coding was established through a subset entropy inequality, namely Theorem 3 therein. As we will point out, this subset entropy inequality is in fact a generalization of Han's inequality \cite{Han-inequality-78}. The proof of Theorem 3 in \cite{yeung99}, however, is painstaking. In this section, we present a weaker version of this theorem, namely \Cref{thm-subset-entropy-inequality} below, whose proof is considerably simpler. With our explicit characterization of $\cR_{\sup}$ in \Cref{thm-alter-region}, \Cref{thm-subset-entropy-inequality} is sufficient for proving the optimality of superposition coding. 
	
	\begin{theorem}[Subset entropy inequality]  \label{thm-subset-entropy-inequality}
		Let $L\geq2$ and for any $\bm{u}\in\{0,1\}^L$, let $H_{\bm{u}}=H\big(W_i:u_i=1\big)$. For any $\bm{\lambda}\in\cG_L$, there exists $\{c_{\alpha}(\bm{u})\},~\alpha\in\cL$, where $\{c_{\alpha}(\bm{u})\}$ is an optimal $\alpha$-resolution for $\bm{\lambda}$, such that
		\begin{equation}
		\sum_{\bm{u}\in\Omega_L^{\alpha-1}}c_{\alpha-1}(\bm{u})H_{\bm{u}}\geq \sum_{\bm{u}\in\Omega_L^{\alpha}}c_{\alpha}(\bm{u})H_{\bm{u}}  \label{subset-entropy-inequality}
		\end{equation}
		for all $\alpha=2,3,\cdots,L$.
	\end{theorem}
	\begin{remark}
		Theorem 3 in \cite{yeung99} is the same as \Cref{thm-subset-entropy-inequality} above except that the former is for all $\bm{\lambda}\in\mathbb{R}_+^L$. By the explicit characterization of $\cR_{\sup}$ in \Cref{thm-alter-region}, namely $\cR^*$, \Cref{thm-subset-entropy-inequality} is sufficient for proving the tightness of~$\cR^*$.
	\end{remark}
	\begin{remark}
		For $\alpha\in\cL$ and $\bm{u}\in\Omega_L^{\alpha}$, let
		\begin{equation*}
		\tilde{c}_{\alpha}(\bm{u})=\frac{1}{\binom{L-1}{\alpha-1}}.
		\end{equation*}
		It is not difficult to see that for all $\alpha\in\cL$, $\{\tilde{c}_{\alpha}(\bm{u}):\bm{u}\in\Omega_L^{\alpha}\}$ is the unique optimal $\alpha$-resolution for $\bm{\lambda}=\bm{1}$. Then \eqref{subset-entropy-inequality} in \Cref{thm-subset-entropy-inequality} becomes
		\begin{equation*}
		\frac{1}{\binom{L-1}{\alpha-2}}\sum_{\bm{u}\in\Omega_L^{\alpha-1}}H_{\bm{u}}\geq \frac{1}{\binom{L-1}{\alpha-1}}\sum_{\bm{u}\in\Omega_L^{\alpha}}H_{\bm{u}},
		\end{equation*}
		which is Han's inequality \cite{Han-inequality-78}. It was proved in \cite{liutie14-s-all} that both Han's inequality and the subset entropy inequality in~\cite{yeung99} can be established from the subset entropy inequality of Madiman and Tetali~\cite{Madiman-Tetali-inequality10}.
	\end{remark}
	\begin{proof}[Proof of \Cref{thm-subset-entropy-inequality}]
		By symmetry, we only have to prove the theorem for $\bm{\lambda}\in\cG_L^0$. We will prove the theorem by induction on $L$. It is easy to check that the theorem is true for $L=2$. Assume the theorem is true for $L=N-1$, we will show that the theorem is also true for $L=N$. This can be readily verified for $\bm{\lambda}\in\cG_N^0$ such that $\zeta=1$. Thus, we only need to consider $\bm{\lambda}\in\cG_N^0$ such that $\zeta\geq 2$. 
		
		For any $\bm{\lambda}_N=(\lambda_1,\lambda_2,\cdots,\lambda_N)\in\cG_N^0$, by the construction in \eqref{def-cG_L}, we have $\bm{\lambda}_{N-1}= (\lambda_2,\lambda_3, \cdots,\lambda_N) \in\cG_{N-1}^0$. For $\alpha\in\{1,2,\cdots,N-1\}$, by the induction hypothesis, let $\left\{\tilde{c}_{\alpha}(\bm{u}):\bm{u}\in\Omega_{N-1}^{\alpha}\right\}$ be an optimal $\alpha$-resolution for $\bm{\lambda}_{N-1}$ such that \eqref{subset-entropy-inequality} is satisfied for all $\alpha=2,3,\cdots,N-1$. Now we need to design a proper optimal $\alpha$-resolution $\{c_{\alpha}(\bm{w}):\bm{w}\in\Omega_N^{\alpha}\}$ for $\bm{\lambda}_N$ that satisfies \eqref{subset-entropy-inequality} for all $\alpha=2,3,\cdots,N$.
		
		From \eqref{def-cG_L}, there exists a $\theta\in\{1,2,\cdots,N-1\}$ such that
		\begin{equation}
		\lambda_1= \frac{1}{\theta}\sum_{i=2}^N\lambda_i.   \label{pf-thm-subset-lambda(1)}
		\end{equation}
		For any $\bm{u}\in\{0,1\}^{N-1}$ and $\bm{w}\in\{0,1\}^N$, let $\bm{u}=(u_2,u_3\cdots,u_N)$ and $\bm{w}=(w_1,w_2,\cdots,w_N)$. For $\alpha=1,2,\cdots,N$, we now construct an $\alpha$-resolution for $\bm{\lambda}_{N}$ in (i) and (ii) in the following.
		\begin{enumerate}[(i)]
			\item Design $\{c_{\alpha}(\bm{w})\}$ for $\alpha=\theta+1,\theta+2,\cdots,N$.
			
			For $\alpha\geq \theta+1$ and $\bm{w}\in\Omega_{N}^{\alpha}$, let 
			\begin{equation*}
			c_{\alpha}(\bm{w})=
			\begin{cases}
			\tilde{c}_{\alpha-1}(\bm{u}),&\text{if }\bm{w}=(1,\bm{u}),~\bm{u}\in\Omega_{N-1}^{\alpha-1}  \\
			0,&\text{otherwise}.
			\end{cases}
			\end{equation*}
			From \eqref{pf-thm-subset-lambda(1)}, we have
			\begin{equation}
			\lambda_1\geq \frac{1}{\alpha-1}\sum_{i=2}^N\lambda_i~~\text{ for all } \alpha=\theta+1,\theta+2,\cdots,N.   \label{pf-thm-subset-lambda1>=}
			\end{equation}
			Lemma 9 in \cite{yeung99} states that $\{c_{\alpha}(\bm{w})\}$ is an optimal $\alpha$-resolution for $\bm{\lambda}_N$ if 
			\begin{equation}
			\lambda_1> \frac{1}{\alpha-1}\sum_{i=2}^N\lambda_i~~\text{ for all } \alpha=\theta+1,\theta+2,\cdots,N.   \label{pf-thm-subset-Lemma9[2]-condition}
			\end{equation}
			We observe that the lemma can be strengthened by replacing the strict inequality in \eqref{pf-thm-subset-Lemma9[2]-condition} by a non-strict inequality (i.e., the condition in \eqref{pf-thm-subset-lambda1>=}) with essentially no change in the proof. Thus, by invoking this strengthened version of the lemma, we conclude that $\{c_{\alpha}(\bm{w})\}$ is an optimal $\alpha$-resolution for $\bm{\lambda}_N$.
			
			For all $\alpha=\theta+2,\theta+3,\cdots,N$, following the steps leading to (48) in \cite{yeung99}, we can check that 
			\begin{equation*}
			\sum_{\bm{w}\in\Omega_N^{\alpha-1}}c_{\alpha-1}(\bm{w})H_{\bm{w}}\geq \sum_{\bm{w}\in\Omega_N^{\alpha}}c_{\alpha}(\bm{w})H_{\bm{w}}.
			\end{equation*}
			
			\item Design $\{c_{\alpha}(\bm{w})\}$ for $\alpha=1,2,\cdots,\theta$.
			
			For $\alpha=\theta+1,\theta,\cdots,2$ and any optimal $\alpha$-resolution $\{c_{\alpha}(\bm{w}):\bm{w}\in\Omega_N^{\alpha}\}$ for $\bm{\lambda}_N$, we claim that there exists an optimal $(\alpha-1)$-resolution $\{c_{\alpha-1}(\bm{w}):\bm{w}\in\Omega_N^{\alpha-1}\}$ for $\bm{\lambda}_N$ such that \eqref{subset-entropy-inequality} is satisfied. Since $\lambda_1\leq \frac{1}{\alpha-1}\sum_{i=2}^N\lambda_i$, this is exactly the first case of the proof of Proposition 1 in \cite{yeung99}, which is relatively straightforward. 
		\end{enumerate}
		In (i) and (ii) above, we have constructed an optimal $\alpha$-resolution $\{c_{\alpha}(\bm{w})\}$ for any $\bm{\lambda}_N\in\cG_N^0$ that satisfies \eqref{subset-entropy-inequality} for all $\alpha=2,3,\cdots,N$. This proves the theorem.
	\end{proof}
	
	\section{Conclusion and Remarks}  \label{section-conclusion}
	In this paper, we studied the SMDC problem for which superposition coding was proved to be optimal in~\cite{yeung97,yeung99}. We enhanced their results by obtaining in closed form the minimum set of inequalities that is needed for characterizing $\cR_{\sup}$, the superposition coding rate region. We further show by the symmetry of the problem that only a much smaller subset of these inequalities needs to be verified in determining the achievability of a given rate tuple. Yet, the cardinality of this smaller set grows at least exponentially fast with~$L$, the number of levels of the coding system, thus revealing the inherent complexity of the problem. A subset entropy inequality, which plays a key role in the converse proof in \cite{yeung99}, requires a painstaking and extremely technical proof. We present a weaker version of this subset entropy inequality whose proof is considerably simpler. With our explicit characterization of the coding rate region, this weaker version of the subset entropy inequality is sufficient for proving the optimality of superposition coding. Some of our results may be extensible to the more general settings in \cite{liutie13-sSMDC,liutie14-s-all,xiaozhiqing-D-MDC15,tianchao-liutie-SMDC-Re16,A-MDC-tianchao}.
	
	While the coding rate region needs to be characterized by a set of inequalities whose size grows 
	at least exponentially with~$L$, if these inequalities are used directly for checking whether a certain rate tuple is within the coding rate
	region, then inevitably it requires at least an exponential amount of time.  However, given that these inequalities
	are not arbitrary but instead highly structured, it may still be possible to devise a polynomial-time algorithm to preform the checking.
	If such an algorithm indeed exists, then the results in this paper can well be an important handle for finding it.
	This is an interesting problem for future research.

	\section{Acknowledgements}  \label{section-ack}
	The authors wish to thank Professor Chandra Nair and the Associate Editor for the suggestions on the linear programming method for checking rate tuple achievability. 
	The authors also wish to thank Professor Chao Tian for discussions on the duality of linear programs in the proof of Theorem 3. 
	The authors also wish to thank the reviewers for their helpful and detailed comments. 
	
	\begin{appendices}
		\section{Proof of \Cref{lemma-combination-of-two}}  \label{section-proof-combination-two}
		We first prove (i). By \Cref{lemma-f-lambda1} (i), the condition $\lambda_1\leq \frac{1}{\eta}\sum_{j=2}^L\lambda_j$ implies that 
		\begin{equation}
		f_{\eta+1}(\bm{\lambda})=\frac{1}{\eta+1}\sum_{j=1}^L\lambda_j.  \label{pf-combination-two-f-theta+1}
		\end{equation}
		In the following, we prove the claim by contradiction. Assume there exists a nonempty subset $\cI_1\subseteq\cI$ such that $i\in\cI_1$ if and only if 
		\begin{equation}
		\lambda^{(i)}_{\pi_i(1)}>\frac{1}{\eta} \sum_{j=2}^L\lambda^{(i)}_{\pi_i(j)},   \label{pf-combination-two-assmuption-1}
		\end{equation}
		which is equivalent to
		\begin{equation*}
		\sum_{j=1}^L\lambda^{(i)}_{\pi_i(j)}>\left(1+\frac{1}{\eta}\right) \sum_{j=2}^L\lambda^{(i)}_{\pi_i(j)},
		\end{equation*}
		or
		\begin{equation}
		\frac{1}{\eta+1} \sum_{j=1}^L\lambda^{(i)}_{\pi_i(j)}>\frac{1}{\eta} \sum_{j=2}^L\lambda^{(i)}_{\pi_i(j)}.   \label{pf-combination-two-f(theta)-1}
		\end{equation}
		For all $i\in\cI$, by \Cref{lemma-f-lambda1} (ii), the condition in \eqref{pf-combination-two-assmuption-1} implies that 
		\begin{equation}
		f_{\eta+1}\big(\bm{\lambda}^{(i)}\big)=f_{\eta}\big(\lambda^{(i)}_{\pi_i(2)},\lambda^{(i)}_{\pi_i(3)},\cdots,\lambda^{(i)}_{\pi_i(L)}\big).  \label{pf-combination-two-f(theta)-2}
		\end{equation}
		By \Cref{thm-opt-f}, we have
		\begin{equation}
		f_{\eta}(\lambda^{(i)}_{\pi_i(2)},\lambda^{(i)}_{\pi_i(3)},\cdots,\lambda^{(i)}_{\pi_i(L)})\leq \frac{1}{\eta} \sum_{j=2}^L\lambda^{(i)}_{\pi_i(j)}.  \label{pf-combination-two-f(theta)-3}
		\end{equation}
		Then by \eqref{pf-combination-two-f(theta)-2}, \eqref{pf-combination-two-f(theta)-3}, and \eqref{pf-combination-two-f(theta)-1} we obtain
		\begin{equation}
		f_{\eta+1}(\bm{\lambda}^{(i)})<\frac{1}{\eta+1} \sum_{j=1}^L\lambda^{(i)}_{\pi_i(j)}=\frac{1}{\eta+1} \sum_{j=1}^L\lambda^{(i)}_j.  \label{pf-combination-two-f2-i0}
		\end{equation}
		For $i\in\cI\backslash\cI_1$, we have 
		\begin{equation*}
		\lambda^{(i)}_{\pi_i(1)}\leq \frac{1}{\eta} \sum_{j=2}^L\lambda^{(i)}_{\pi_i(j)},
		\end{equation*}
		which by \Cref{lemma-f-lambda1} (i) implies that 
		\begin{equation}
		f_{\eta+1}(\bm{\lambda}^{(i)})=\frac{1}{\eta+1}\sum_{j=1}^L\lambda^{(i)}_j. \label{pf-combination-two-f2-i}
		\end{equation}
		Thus, we have from \eqref{assumption-lemma-combination-two} that
		\begin{eqnarray}
		f_{\eta+1}(\bm{\lambda})&=&\sum_{i=1}^{S_L}c_i f_{\eta+1}(\bm{\lambda}^{(i)})   \nonumber \\
		&=&\sum_{i\in\cI_1}c_i f_{\eta+1}(\bm{\lambda}^{(i)})+\sum_{i\in\cI\backslash\cI_1}c_i f_{\eta+1}(\bm{\lambda}^{(i)})  \nonumber \\
		&<&\sum_{i=1}^{S_L}c_i \cdot\left(\frac{1}{\eta+1}\sum_{j=1}^L\lambda^{(i)}_j\right)  \nonumber \\
		&=&\frac{1}{\eta+1}\sum_{j=1}^L\left(\sum_{i=1}^{S_L}c_i\lambda^{(i)}_j\right)  \nonumber \\
		&=&\frac{1}{\eta+1}\sum_{j=1}^L\lambda_j, \nonumber 
		\end{eqnarray}
		where the inequality follows from \eqref{pf-combination-two-f2-i0} and \eqref{pf-combination-two-f2-i}. This is a contradiction to \eqref{pf-combination-two-f-theta+1}. Thus, the assumption in~\eqref{pf-combination-two-assmuption-1} is false and we have for all $i\in\cI$ that
		\begin{equation*}
		\lambda^{(i)}_{\pi_i(1)}\leq \frac{1}{\eta} \sum_{j=2}^L\lambda^{(i)}_{\pi_i(j)}.  
		\end{equation*}
		
		Next, we prove (ii) by contradiction. Assume there exists a nonempty subset $\cI_2\subseteq\cI$ such that $i\in\cI_2$ if and only if 
		\begin{equation}
		\lambda^{(i)}_{\pi_i(1)}<\frac{1}{\eta} \sum_{j=2}^L\lambda^{(i)}_{\pi_i(j)},  \label{pf-combination-two-assmuption-2}
		\end{equation}
		which is equivalent to
		\begin{equation*}
		\frac{1}{\eta+1}\sum_{j=1}^L\lambda^{(i)}_{\pi_i(j)} <\frac{1}{\eta}\sum_{j=2}^L\lambda^{(i)}_{\pi_i(j)},
		\end{equation*}
		or 
		\begin{equation}
		g_{\eta+1,\bm{\lambda}^{(i)}}(0)<g_{\eta+1,\bm{\lambda}^{(i)}}(1).   \label{pf-combination-two-g(0)<g(1)}
		\end{equation}
		For any $i\in\cI_2$, by \Cref{lemma-f-lambda1} (i), \eqref{pf-combination-two-assmuption-2} implies that
		\begin{equation*}
		f_{\eta+1}(\bm{\lambda}^{(i)})=g_{\eta+1,\bm{\lambda}^{(i)}}(0).
		\end{equation*}
		For any $t\in\{1,2,\cdots,\eta\}$, in light of \eqref{pf-combination-two-g(0)<g(1)}, by applying the alternative version of \Cref{lemma-relation-beta} (ii) (see the remark below \Cref{lemma-relation-beta}), we obtain
		\begin{equation*}
		f_{\eta+1}(\bm{\lambda}^{(i)})=g_{\eta+1,\bm{\lambda}^{(i)}}(0)<g_{\eta+1,\bm{\lambda}^{(i)}}(1)<\cdots<g_{\eta+1,\bm{\lambda}^{(i)}}(t).  
		\end{equation*}
		Then it follows from the definition of $\pi_i(\cdot)$ in \eqref{def-pi(i)} that
		\begin{eqnarray}
		f_{\eta+1}(\bm{\lambda}^{(i)})&<&g_{\eta+1,\bm{\lambda}^{(i)}}(t) \nonumber \\
		&=&\frac{1}{\eta+1-t}\sum_{j=t+1}^L\lambda^{(i)}_{\pi_i(j)} \nonumber \\
		&\leq& \frac{1}{\eta+1-t}\sum_{j=t+1}^L\lambda^{(i)}_j,  \label{pf-combination-two-f(theta+2)-pi}
		\end{eqnarray}
		and so
		\begin{equation}
		f_{\eta+1}(\bm{\lambda}^{(i)})<\frac{1}{\eta+1-t}\sum_{j=t+1}^L\lambda^{(i)}_j.  \label{pf-combination-two-f(theta+2)-i0}
		\end{equation}
		For all $i\in\cI\backslash\cI_2$ and any $t\in\{1,2,\cdots,\eta\}$, by \Cref{thm-opt-f}, similar to \eqref{pf-combination-two-f(theta+2)-pi}, we have
		\begin{eqnarray}
		f_{\eta+1}(\bm{\lambda}^{(i)})&\leq& g_{\eta+1,\bm{\lambda}^{(i)}}(t)=\frac{1}{\eta+1-t}\sum_{j=t+1}^L\lambda^{(i)}_{\pi_i(j)} \nonumber \\
		&\leq& \frac{1}{\eta+1-t}\sum_{j=t+1}^L\lambda^{(i)}_j,  \nonumber
		\end{eqnarray}
		and so
		\begin{equation}
		f_{\eta+1}(\bm{\lambda}^{(i)})\leq \frac{1}{\eta+1-t}\sum_{j=t+1}^L\lambda^{(i)}_j.  \label{pf-combination-two-f(theta+2)-i}
		\end{equation}
		Thus, by \eqref{assumption-lemma-combination-two}, we have for any $t\in\{1,2,\cdots,\eta\}$ that 
		\begin{eqnarray}
		f_{\eta+1}(\bm{\lambda})&=&\sum_{i=1}^{S_L}c_i f_{\eta+1}(\bm{\lambda}^{(i)})  \nonumber \\
		&<&\sum_{i=1}^{S_L}c_i \cdot\left(\frac{1}{\eta+1-t}\sum_{j=t+1}^L\lambda^{(i)}_j\right)  \nonumber \\
		&=&\frac{1}{\eta+1-t}\sum_{j=t+1}^L\lambda_j  \nonumber \\
		&=&g_{\eta+1,\bm{\lambda}}(t),  \label{pf-combination-two-contradiction}
		\end{eqnarray}
		where the inequality follows from \eqref{pf-combination-two-f(theta+2)-i0} and \eqref{pf-combination-two-f(theta+2)-i}. 
		
		The condition $\lambda_1\geq \frac{1}{\eta} \sum_{j=2}^L\lambda_j$ is equivalent to
		\begin{equation*}
		g_{\eta+1,\bm{\lambda}}(0)\geq g_{\eta+1,\bm{\lambda}}(1).
		\end{equation*}
		Then by \Cref{thm-opt-f}, we have
		\begin{equation*}
		f_{\eta+1}(\bm{\lambda})=\min_{\beta\in\{0,1,\cdots,\eta\}}g_{\eta+1,\bm{\lambda}}(\beta)=\min_{\beta\in\{1,2,\cdots,\eta\}}g_{\eta+1,\bm{\lambda}}(\beta).
		\end{equation*}
		Thus, there must exist a $t\in\{1,2,\cdots,\eta\}$ such that 
		\begin{equation*}
		f_{\eta+1}(\bm{\lambda})=g_{\eta+1,\bm{\lambda}}(t),
		\end{equation*}
		which is a contradiction to \eqref{pf-combination-two-contradiction}. Thus the assumption in \eqref{pf-combination-two-assmuption-2} is false and we have for all~$i\in\cI$ that
		\begin{equation*}
		\lambda^{(i)}_{\pi_i(1)}\geq \frac{1}{\eta} \sum_{j=2}^L\lambda^{(i)}_{\pi_i(j)}.
		\end{equation*}
		
		\section{Proof of \Cref{lemma-combination-of-two-range}}  \label{section-proof-lemma-combination-two-range}
		For all $i\in\cI$, since $\bm{\lambda}^{(i)}\in\cG_L$ in light of \eqref{def-cG_L}, we only need to prove $\lambda^{(i)}_{\pi_i(1)}\leq \frac{1}{\eta-1}\sum_{j=2}^L\lambda^{(i)}_{\pi_i(j)}$ and $\lambda^{(i)}_{\pi_i(1)}\geq \frac{1}{\eta}\sum_{j=2}^L\lambda^{(i)}_{\pi_i(j)}$. 
		
		We first prove the upper bound on $\lambda^{(i)}_{\pi_i(1)}$. For $\eta=1$ and $i=1$, we have 
		\begin{equation*}
		\lambda^{(i)}_{\pi_i(1)}=1=\frac{1}{\eta-1}\sum_{j=2}^L\lambda^{(i)}_{\pi_i(j)}.
		\end{equation*}
		For $\eta=1$ and $i\in\cI\backslash\{1\}$, it is obvious that
		\begin{equation*}
		\lambda^{(i)}_{\pi_i(1)}<\frac{1}{\eta-1}\sum_{j=2}^L\lambda^{(i)}_{\pi_i(j)}=\infty.
		\end{equation*}
		For $\eta\in\{2,3,\cdots,\zeta-1\}$, the upper bound in \eqref{combination-two-condition} can be rewritten as  
		\begin{equation*}
		\lambda_1<\frac{1}{\eta'}\sum_{j=2}^L\lambda_j,
		\end{equation*}
		where $\eta'=\eta-1$ and $\eta'\in\{1,2,\cdots,\zeta-2\}$. By \Cref{lemma-combination-of-two} (i), this implies that
		\begin{equation*}
		\lambda_{\pi_i(1)}^{(i)}\leq \frac{1}{\eta'}\sum_{j=2}^L\lambda_{\pi_i(j)}^{(i)}=\frac{1}{\eta-1}\sum_{j=2}^L\lambda_{\pi_i(j)}^{(i)}.
		\end{equation*}
		Thus, the upper bound on $\lambda^{(i)}_{\pi_i(1)}$ is proved. 
		
		Now we prove the lower bound on $\lambda^{(i)}_{\pi_i(1)}$. For $\eta\in\{1,2,\cdots,\zeta-1\}$, the lower bound in \eqref{combination-two-condition} is
		\begin{equation*}
		\lambda_1\geq \frac{1}{\eta}\sum_{j=2}^L\lambda_j,
		\end{equation*}
		and so by \Cref{lemma-combination-of-two} (ii), we have
		\begin{equation*}
		\lambda_{\pi_i(1)}^{(i)}\geq \frac{1}{\eta}\sum_{j=2}^L\lambda_{\pi_i(j)}^{(i)}.
		\end{equation*}
		If the lower bound in \eqref{combination-two-condition} is tight, it follows immediately from \Cref{lemma-combination-of-two} that for any $\eta\in\{1,2,\cdots,\zeta-1\}$,
		\begin{equation*}
		\lambda_{\pi_i(1)}^{(i)}=\frac{1}{\eta}\sum_{j=2}^L\lambda_{\pi_i(j)}^{(i)}.
		\end{equation*}
		This proves the lemma.
		
		\section{Proof of \Cref{lemma-combination-reduce-dim}}  \label{section-proof-combination-reduce-dim}
		We only need to show that for $\alpha=1,2,\cdots,L-1$,
		\begin{equation*}
		f_{\alpha}(\bm{\lambda}_{L-1})=\sum_{i=1}^{S_L}c_i f_{\alpha}(\lambda^{(i)}_2,\lambda^{(i)}_3,\cdots,\lambda^{(i)}_L).
		\end{equation*}
		By \eqref{combination-two-condition} and \Cref{lemma-combination-of-two-range}, we have for $i\in\cI$ that 
		\begin{equation}
		\lambda^{(i)}_{\pi_i(1)}\in\left\{\frac{1}{\eta}\sum_{j=2}^L\lambda^{(i)}_{\pi_i(j)},\frac{1}{\eta-1}\sum_{j=2}^L\lambda^{(i)}_{\pi_i(j)}\right\}.  \label{pf-combination-reduce-dim-lambda1i=}
		\end{equation}
		Consider the following two cases:
		\begin{enumerate}[i)]
			\item $\alpha=1,2,\cdots,\eta-1$;
			\item $\alpha=\eta,\eta+1,\cdots,L-1$.
		\end{enumerate}
		\textit{Case i):} For $\alpha=1,2,\cdots,\eta-1$, if $\eta=2$, then $\alpha$ can only be 1 and it is easy to see that
		\begin{eqnarray}
		f_{1}(\bm{\lambda}_{L-1})&=&\sum_{j=2}^L\lambda_j=\sum_{j=2}^L\sum_{i=1}^{S_L}c_i\lambda^{(i)}_j =\sum_{i=1}^{S_L}c_i\sum_{j=2}^L\lambda^{(i)}_j  \nonumber \\
		&=&\sum_{i=1}^{S_L}c_i f_{1}(\lambda^{(i)}_2,\lambda^{(i)}_3,\cdots,\lambda^{(i)}_L).  \nonumber 
		\end{eqnarray}
		If $\eta>2$, consider the following. The second inequality in \eqref{combination-two-condition} is equivalent to 
		\begin{equation}
		\frac{1}{\eta}\sum_{j=1}^L\lambda_j< \frac{1}{\eta-1}\sum_{j=2}^L\lambda_j  \label{pf-reduce-c1-1}
		\end{equation}
		or
		\begin{equation*}
		g_{\eta,\bm{\lambda}}(0)<g_{\eta,\bm{\lambda}}(1).
		\end{equation*}
		By applying the alternative version of \Cref{lemma-relation-beta} (ii) (see the remark below \Cref{lemma-relation-beta}), we obtain 
		\begin{equation*}
		g_{\eta,\bm{\lambda}}(1)<g_{\eta,\bm{\lambda}}(2),
		\end{equation*}
		which is equivalent to
		\begin{equation*}
		\frac{1}{\eta-1}\sum_{j=2}^L\lambda_j< \frac{1}{\eta-2}\sum_{j=3}^L\lambda_j
		\end{equation*}
		or
		\begin{equation}
		\lambda_2< \frac{1}{\eta-2}\sum_{j=3}^L\lambda_j.  \label{pf-reduce-c1-2}
		\end{equation}
		Then by \Cref{lemma-f-lambda1} (i), we have 
		\begin{equation}
		f_{\alpha}(\bm{\lambda}_{L-1})=g_{\alpha,\bm{\lambda}_{L-1}}(0)=\frac{1}{\alpha}\sum_{j=2}^L\lambda_j.  \label{pf-reduce-f-lambda}
		\end{equation}
		Since \eqref{pf-combination-reduce-dim-lambda1i=} implies 
		\begin{equation}
		\lambda^{(i)}_{\pi_i(1)}\leq \frac{1}{\eta-1}\sum_{j=2}^L\lambda^{(i)}_{\pi_i(j)},  \label{pf-combination-reduce-dim-lambda1i<=1}
		\end{equation}
		similar to \eqref{pf-reduce-c1-1}-\eqref{pf-reduce-c1-2} (with all $<$'s replaced by $\leq$'s), we have
		\begin{equation}
		\lambda^{(i)}_{\pi_i(2)}\leq \frac{1}{\eta-2}\sum_{j=3}^L\lambda^{(i)}_{\pi_i(j)}.  \label{pf-combination-reduce-dim-lambda2i<=}
		\end{equation}
		Thus, following \eqref{pf-combination-reduce-dim-lambda1i<=1} and \eqref{pf-combination-reduce-dim-lambda2i<=}, we have 
		\begin{eqnarray}
		\lambda^{(i)}_{\pi_i(1)}&\leq & \frac{1}{\eta-1}\sum_{j=2}^L\lambda^{(i)}_{\pi_i(j)}  \nonumber \\
		&\leq &\frac{1}{\eta-1}\left(\frac{1}{\eta-2}+1\right)\sum_{j=3}^L\lambda^{(i)}_{\pi_i(j)}  \nonumber \\
		&=&\frac{1}{\eta-2}\sum_{j=3}^L\lambda^{(i)}_{\pi_i(j)}  \nonumber \\
		&\leq &\frac{1}{\eta-2}\sum_{j\in\{2,3,\cdots,L\}\backslash\{j_0\}}\lambda^{(i)}_{\pi_i(j)}  \label{pf-combination-reduce-dim-lambda1i<=}
		\end{eqnarray}
		for any $j_0\in\{2,3,\cdots,L\}$. Let $\pi_i'(\cdot)$ be a permutation of $\{2,3,\cdots,L\}$ defined as follows:
		\begin{enumerate}[a)]
			\item if $\pi_i(1)=1$, then $\pi_i'(j)=\pi_i(j)$ for all $j\in\{2,3,\cdots,L\}$;
			\item if $\pi_i(j_0)=1$ for some $j_0\in\{2,3,\cdots,L\}$, then 
			\begin{equation}
			\pi_i'(j)=
			\begin{cases}
			\pi_i(j-1),&\text{for }j=2,3,\cdots,j_0  \\
			\pi_i(j),&\text{for }j=j_0+1,\cdots,L.  \\
			\end{cases}   \label{pf-combination-reduce-dim-pi'}
			\end{equation}
		\end{enumerate}
		It is easy to check that
		\begin{equation*}
		\lambda^{(i)}_{\pi_i'(2)}\geq \lambda^{(i)}_{\pi_i'(3)}\geq \cdots\geq \lambda^{(i)}_{\pi_i'(L)}.
		\end{equation*}
		If a) holds, then by \eqref{pf-combination-reduce-dim-lambda2i<=}, we have
		\begin{equation*}
		\lambda^{(i)}_{\pi_i'(2)}=\lambda^{(i)}_{\pi_i(2)}\leq \frac{1}{\eta-2}\sum_{j=3}^L\lambda^{(i)}_{\pi_i(j)}=\frac{1}{\eta-2}\sum_{j=3}^L\lambda^{(i)}_{\pi_i'(j)}.
		\end{equation*}
		If b) holds, then by \eqref{pf-combination-reduce-dim-lambda1i<=}, we have
		\begin{eqnarray}
		\lambda^{(i)}_{\pi_i'(2)}&=&\lambda^{(i)}_{\pi_i(1)}  \nonumber \\
		&\leq& \frac{1}{\eta-2}\sum_{j\in\{2,3,\cdots,L\}\backslash\{j_0\}}\lambda^{(i)}_{\pi_i(j)}  \nonumber \\
		&=&\frac{1}{\eta-2}\sum_{j=3}^L\lambda^{(i)}_{\pi_i'(j)}.  \nonumber 
		\end{eqnarray}
		Summarizing the two cases, we see that
		\begin{equation*}
		\lambda^{(i)}_{\pi_i'(2)}\leq \frac{1}{\eta-2}\sum_{j=3}^L\lambda^{(i)}_{\pi_i'(j)}
		\end{equation*}
		always holds. By \Cref{lemma-f-lambda1} (i), this implies that
		\begin{equation}
		f_{\alpha}(\lambda^{(i)}_2,\lambda^{(i)}_3,\cdots,\lambda^{(i)}_L)=\frac{1}{\alpha}\sum_{j=2}^L\lambda^{(i)}_{\pi_i'(j)}=\frac{1}{\alpha}\sum_{j=2}^L\lambda^{(i)}_j.  \label{pf-reduce-f-lambda(i)}
		\end{equation}
		Following \eqref{pf-reduce-f-lambda}, we have
		\begin{eqnarray}
		f_{\alpha}(\bm{\lambda}_{L-1})&=&\frac{1}{\alpha}\sum_{j=2}^L\lambda_j  \nonumber \\
		&=&\frac{1}{\alpha}\sum_{j=2}^L\left(\sum_{i=1}^{S_L}c_i\lambda^{(i)}_j\right) \label{pf-combination-reduce-dim-f(sum)-2} \\
		&=&\sum_{i=1}^{S_L}c_i\left(\frac{1}{\alpha}\sum_{j=2}^L\lambda^{(i)}_j\right) \nonumber  \\
		&=&\sum_{i=1}^{S_L}c_i f_{\alpha}(\lambda^{(i)}_2,\lambda^{(i)}_3,\cdots,\lambda^{(i)}_L), \label{pf-combination-reduce-dim-f(sum)-last}
		\end{eqnarray}
		where \eqref{pf-combination-reduce-dim-f(sum)-2} follows from \eqref{assumption-lemma-combination-reduce-dim} and \eqref{pf-combination-reduce-dim-f(sum)-last} follows from \eqref{pf-reduce-f-lambda(i)}.
		
		\textit{Case ii):} For $\alpha=\eta,\eta+1,\cdots,L-1$, by \Cref{lemma-f-lambda1} (ii), the first inequality in \eqref{combination-two-condition} implies that 
		\begin{equation}
		f_{\alpha}(\bm{\lambda}_{L-1})=f_{\alpha+1}(\bm{\lambda}).  \label{pf-combination-reduce-dim-f(lambda)}
		\end{equation}
		For any $i\in\{1,2,\cdots,S_L\}$, since \eqref{pf-combination-reduce-dim-lambda1i=} implies $\lambda^{(i)}_{\pi_i(1)}\geq \frac{1}{\eta}\sum_{j=2}^L\lambda^{(i)}_{\pi_i(j)}$, we have by \Cref{lemma-f-lambda1} (ii) that 
		\begin{equation}
		f_{\alpha+1}(\bm{\lambda}^{(i)})=f_{\alpha}\big(\lambda^{(i)}_{\pi_i(2)},\lambda^{(i)}_{\pi_i(3)},\cdots,\lambda^{(i)}_{\pi_i(L)}\big).  \label{pf-combination-reduce-dim-f(lambda-i)-1}
		\end{equation}
		From the definition of $\pi_i'(\cdot)$, it is readily seen that
		\begin{equation*}
		\lambda^{(i)}_{\pi_i(j)}\leq \lambda^{(i)}_{\pi_i'(j)}, \text{ for all }j=2,3,\cdots,L.
		\end{equation*}
		Thus, we have for any $\beta=1,2,\cdots,\alpha-1$ that
		\begin{equation*}
		\frac{1}{\alpha-\beta}\sum_{j=\beta+1}^L \lambda^{(i)}_{\pi_i(j)}\leq \frac{1}{\alpha-\beta}\sum_{j=\beta+1}^L \lambda^{(i)}_{\pi_i'(j)}.
		\end{equation*}
		By \Cref{thm-opt-f}, this implies that 
		\begin{equation*}
		f_{\alpha}\big(\lambda^{(i)}_{\pi_i(2)},\lambda^{(i)}_{\pi_i(3)},\cdots,\lambda^{(i)}_{\pi_i(L)}\big)\leq f_{\alpha}\big(\lambda^{(i)}_{\pi_i'(2)},\lambda^{(i)}_{\pi_i'(3)},\cdots,\lambda^{(i)}_{\pi_i'(L)}\big),  
		\end{equation*}
		and thus by \eqref{pf-combination-reduce-dim-f(lambda-i)-1}, we have
		\begin{equation}
		f_{\alpha+1}(\bm{\lambda}^{(i)})\leq f_{\alpha}\big(\lambda^{(i)}_2,\lambda^{(i)}_3,\cdots,\lambda^{(i)}_L\big).  \label{pf-combination-reduce-dim-f(lambda-i)-3}
		\end{equation}
		Following \eqref{pf-combination-reduce-dim-f(lambda)}, we have
		\begin{align}
		f_{\alpha}(\bm{\lambda}_{L-1})&=f_{\alpha+1}(\bm{\lambda})   \label{pf-combination-reduce-dim-case2-1} \\
		&=\sum_{i=1}^{S_L}c_i f_{\alpha+1}(\bm{\lambda}^{(i)})  \label{pf-combination-reduce-dim-case2-2} \\
		&\leq \sum_{i=1}^{S_L}c_i f_{\alpha}\big(\lambda^{(i)}_2,\lambda^{(i)}_3,\cdots,\lambda^{(i)}_L\big)   \label{pf-combination-reduce-dim-case2-3} \\
		&\leq f_{\alpha}\left(\sum_{i=1}^{S_L}c_i \cdot\big(\lambda^{(i)}_2,\lambda^{(i)}_3,\cdots,\lambda^{(i)}_L\big)\right)   \label{pf-combination-reduce-dim-case2-4} \\
		&=f_{\alpha}\left(\sum_{i=1}^{S_L}c_i\lambda^{(i)}_2, \sum_{i=1}^{S_L}c_i\lambda^{(i)}_3,\cdots, \sum_{i=1}^{S_L}c_i\lambda^{(i)}_L\right)  \nonumber \\
		&=f_{\alpha}(\lambda_2,\lambda_3,\cdots,\lambda_L)  \label{pf-combination-reduce-dim-case2-6} \\
		&=f_{\alpha}(\bm{\lambda}_{L-1}),  \label{pf-combination-reduce-dim-case2-7}
		\end{align}
		where both \eqref{pf-combination-reduce-dim-case2-2} and \eqref{pf-combination-reduce-dim-case2-6} follow from \eqref{assumption-lemma-combination-reduce-dim}, \eqref{pf-combination-reduce-dim-case2-3} follows from \eqref{pf-combination-reduce-dim-f(lambda-i)-3}, and \eqref{pf-combination-reduce-dim-case2-4} follows from \Cref{lemma-concave-f}. Upon observing that the LHS of \eqref{pf-combination-reduce-dim-case2-1} is the same as the RHS of \eqref{pf-combination-reduce-dim-case2-7}, we conclude that the inequalities in both \eqref{pf-combination-reduce-dim-case2-3} and \eqref{pf-combination-reduce-dim-case2-4} are tight, and hence 
		\begin{equation*}
		f_{\alpha}(\bm{\lambda}_{L-1})=\sum_{i=1}^{S_L}c_i f_{\alpha}\big(\lambda^{(i)}_2,\lambda^{(i)}_3,\cdots,\lambda^{(i)}_L\big).
		\end{equation*}
		The lemma is proved.
		
		\section{Proof of \Cref{lemma-lambda(L-1)}} \label{section-proof-lemma-lambda(L-1)}
		Fix $i\in\{1,2,\cdots,S_L\}$ and assume that $\left(\lambda^{(i)}_2, \lambda^{(i)}_3,\cdots,\lambda^{(i)}_L\right)\in\cG_{L-1}$. Let $\gamma_j,~j=1,2,\cdots,\zeta-1$ be the integer such that
		\begin{equation*}
		\lambda^{(i)}_{\pi_i(j)}=\frac{1}{\gamma_j}\sum_{k=j+1}^L \lambda^{(i)}_{\pi_i(k)},  
		\end{equation*} 
		and let $\gamma_{\zeta}=0$. Note that the role of $\gamma_j$ for $\bm{\lambda}^{(i)}$ is the same as the role of $\theta_j$ for $\bm{\lambda}$ (cf. \eqref{def-cG_L-lambda(j+1)}). Also note that $\zeta$ and $\gamma_j$ depend on~$i$, but since we fix $i$, this dependence is omitted to simplify notation.
		
		Let $j_0\in\{1,2,\cdots,L\}$ be such that 
		\begin{equation}
		\lambda^{(i)}_{\pi_i(j_0)}=\lambda^{(i)}_1.   \label{pf-lambda(L-1)-def-j0}
		\end{equation}
		If $j_0\geq \zeta+1$, $\lambda^{(i)}_1=0$ and thus the lemma is proved. If $j_0=1$, the lemma is immediate from \eqref{pf-lambda(L-1)-def-j0}. If $2\leq j_0\leq \zeta$, we claim that $\gamma_j=\gamma_{j_0}+(j_0-j)$ and  $\lambda^{(i)}_{\pi_i(j)}=\lambda^{(i)}_{\pi_i(j_0)}$ for all $j=j_0,j_0-1,\cdots,1$. Then the lemma follows from the claim for $j=1$. In the following, we prove the claim by induction on $j$ for $j\leq j_0$. The claim is immediate for $j=j_0$. Assume the claim is true for $j=j_0,j_0-1,\cdots,N$ for some $N\in\{2,3,\cdots,j_0\}$, and we will show that the claim is also true for $j=N-1$. By the induction hypothesis, we have 
		\begin{equation*}
		\gamma_{N}=\gamma_{j_0}+(j_0-N),  
		\end{equation*}
		and for all $j=j_0,j_0-1,\cdots,N$,
		\begin{equation}
		\lambda^{(i)}_{\pi_i(j)}=\lambda^{(i)}_{\pi_i(j_0)}=\frac{1}{\gamma_{j_0}}\sum_{k=j_0+1}^L \lambda^{(i)}_{\pi_i(k)}.  \label{pf-lambda(L-1)-induction-lambda}
		\end{equation}
		By \eqref{def-cG_L} and \eqref{range-theta(j)-1}, there exists 
		\begin{equation}
		\gamma_{N-1}\in\{1,2,\cdots,\gamma_{j_0}+(j_0-N+1)\}   \label{pf-lambda(L-1)-range-gamma(N-1)}
		\end{equation}
		such that 
		\begin{equation}
		\lambda^{(i)}_{\pi_i(N-1)}=\frac{1}{\gamma_{N-1}}\sum_{k=N}^L\lambda^{(i)}_{\pi_i(k)}.  \label{pf-lambda(L-1)-lambda(N-1)-1}
		\end{equation}
		Thus, we have
		\begin{eqnarray}
		\lambda^{(i)}_{\pi_i(N-1)}&=&\frac{1}{\gamma_{N-1}}\sum_{k=N}^L\lambda^{(i)}_{\pi_i(k)} \nonumber \\
		&=&\frac{1}{\gamma_{N-1}}\left(1+\sum_{k=N}^{j_0}\frac{1}{\gamma_{j_0}}\right)\sum_{k=j_0+1}^L\lambda^{(i)}_{\pi_i(k)} \label{pf-lambda(L-1)-lambda(N-1)-2-2} \\
		&=&\frac{\gamma_{j_0}+(j_0-N+1)}{\gamma_{j_0}\gamma_{N-1}}\sum_{k=j_0+1}^L\lambda^{(i)}_{\pi_i(k)},  \label{pf-lambda(L-1)-lambda(N-1)-2}
		\end{eqnarray}
		where \eqref{pf-lambda(L-1)-lambda(N-1)-2-2} follows from \eqref{pf-lambda(L-1)-induction-lambda}. In light of \eqref{pf-lambda(L-1)-def-j0} and $j_0\geq 2$, recall the definition of $\pi_i'(\cdot)$ in \eqref{pf-combination-reduce-dim-pi'}. With the assumption that $\left(\lambda^{(i)}_2,\lambda^{(i)}_3,\cdots,\lambda^{(i)}_L\right)\in\cG_{L-1}$, by \eqref{def-cG_L}, there exists an integer $\gamma_{N-1}'$ such that 
		\begin{equation*}
		\lambda^{(i)}_{\pi_i'(N)}=\frac{1}{\gamma_{N-1}'}\sum_{k=N+1}^L\lambda^{(i)}_{\pi_i'(k)} 
		\end{equation*}
		or
		\begin{equation}
		\lambda^{(i)}_{\pi_i(N-1)}=\frac{1}{\gamma_{N-1}'}\left(\sum_{k=N}^{j_0-1}\lambda^{(i)}_{\pi_i(k)}+\sum_{k=j_0+1}^L\lambda^{(i)}_{\pi_i(k)}\right).  \label{pf-lambda(L-1)-lambda(N-1)-3}
		\end{equation}
		Comparing the RHS of \eqref{pf-lambda(L-1)-lambda(N-1)-1} and \eqref{pf-lambda(L-1)-lambda(N-1)-3}, since $\lambda^{(i)}_{\pi_i(j_0)}\geq 1$, we have
		\begin{equation*}
		\sum_{k=N}^L\lambda^{(i)}_{\pi_i(k)}>\sum_{k=N}^{j_0-1}\lambda^{(i)}_{\pi_i(k)}+\sum_{k=j_0+1}^L\lambda^{(i)}_{\pi_i(k)}.
		\end{equation*}
		Since the LHS of \eqref{pf-lambda(L-1)-lambda(N-1)-1} and \eqref{pf-lambda(L-1)-lambda(N-1)-3} are the same, we see that $\gamma_{N-1}'<\gamma_{N-1}$, which implies that
		\begin{equation*}
		\gamma_{N-1}'\leq \gamma_{j_0}+(j_0-N).
		\end{equation*}
		Following \eqref{pf-lambda(L-1)-lambda(N-1)-3}, we have
		\begin{eqnarray}
		\lambda^{(i)}_{\pi_i(N-1)}&=&\frac{1}{\gamma_{N-1}'}\left(\sum_{k=N}^{j_0-1}\lambda^{(i)}_{\pi_i(k)}+\sum_{k=j_0+1}^L\lambda^{(i)}_{\pi_i(k)}\right) \nonumber \\
		&=&\frac{1}{\gamma_{N-1}'}\left(1+\sum_{k=N}^{j_0-1}\frac{1}{\gamma_{j_0}}\right)\sum_{k=j_0+1}^L\lambda^{(i)}_{\pi_i(k)}  \label{pf-lambda(L-1)-lambda(N-1)-4-2} \\
		&=&\frac{\gamma_{j_0}+(j_0-N)}{\gamma_{j_0}\gamma_{N-1}'}\sum_{k=j_0+1}^L\lambda^{(i)}_{\pi_i(k)},  \label{pf-lambda(L-1)-lambda(N-1)-4}
		\end{eqnarray}
		where \eqref{pf-lambda(L-1)-lambda(N-1)-4-2} follows from \eqref{pf-lambda(L-1)-induction-lambda}. Comparing \eqref{pf-lambda(L-1)-lambda(N-1)-2} and \eqref{pf-lambda(L-1)-lambda(N-1)-4}, it is easy to see that
		\begin{equation*}
		\gamma_{N-1}'=\frac{\gamma_{N-1}[\gamma_{j_0}+(j_0-N)]}{\gamma_{j_0}+(j_0-N+1)}.
		\end{equation*}
		Since $\gamma_{j_0}+(j_0-N)$ and $\gamma_{j_0}+(j_0-N+1)$ are coprime and $\gamma_{N-1}\leq \gamma_{j_0}+(j_0-N+1)$ by \eqref{pf-lambda(L-1)-range-gamma(N-1)}, we have 
		\begin{equation}
		\gamma_{N-1}=\gamma_{j_0}+(j_0-N+1).  \label{pf-lambda(L-1)-value-gamma(N-1)}
		\end{equation}
		Substituting \eqref{pf-lambda(L-1)-value-gamma(N-1)} into \eqref{pf-lambda(L-1)-lambda(N-1)-2} and invoking \eqref{pf-lambda(L-1)-induction-lambda}, we have $\lambda^{(i)}_{\pi_i(N-1)}=\lambda^{(i)}_{\pi_i(j_0)}$. This implies that the claim is true for $j=N-1$. The lemma is proved.
		
		\section{Proof of \Cref{lemma-no-redundancy}}  \label{section-proof-lemma-no-redundancy}
		Since there is only one vector in $\cG_1$, we only need to consider $L\geq 2$. If $\zeta=1$ for $\bm{\lambda}^{(i_0)}$, it is obvious that $\bm{\lambda}^{(i_0)}$ cannot be a conic combination of the other vectors in $\cG_L$. Thus, we consider only $\bm{\lambda}^{(i_0)}\in\cG_L$ such that $\zeta\geq 2$. We prove the lemma by induction on $L$ for $L\geq 2$. We first check that the claim is true for $L=2$. It is easy to see from \eqref{def-cG_L} that $\cG_2=\{(1,0), (0,1), (1,1)\}$. Then $\bm{f}\big((1,0)\big)=\bm{f}\big((0,1)\big)=(1,0)$ and $\bm{f}\big((1,1)\big)=(2,1)$. Since 
		\begin{equation*}
		(1,1)=(1,0)+(0,1)
		\end{equation*}
		whereas 
		\begin{equation*}
		f_2\big((1,1)\big)>f_2\big((1,0)\big)+f_2\big((0,1)\big),
		\end{equation*}
		we see that $\big((1,1),\bm{f}\big((1,1)\big)\big)$ cannot be a conic combination of $\big((1,0),\bm{f}\big((1,0)\big)\big)$ and $\big((0,1),\bm{f}\big((0,1)\big)\big)$. Thus, the lemma is true for $L=2$. For any $L\geq 3$, the lemma will be proved by contradiction via the following proposition, whose proof is given in Appendix \ref{section-proof-lemma-no-redundancy-claim}.
		\begin{proposition}\label{prop-induction}
			For any $L\geq 3$, if \Cref{lemma-no-redundancy} is false, then the lemma is false for $L-1$.
		\end{proposition}
		By backward induction, if \Cref{lemma-no-redundancy} is false for any $L\geq 3$, then the lemma is false for $L=2$. This is a contradiction because we already have shown that the lemma is true for $L=2$. This proves the lemma for all $L\geq 2$.
		
		\section{Proof of \Cref{prop-induction}}  \label{section-proof-lemma-no-redundancy-claim}
		Assume \Cref{lemma-no-redundancy} is false for some $L\geq 3$, i.e., for some $i_0\in\{1,2,\cdots,S_L\}$, there exists $(c_1,c_2,\cdots,c_{S_L})\in\mathbb{R}_+^{S_L}$ such that $c_{i_0}=0$ and 
		\begin{equation}
		\big(\bm{\lambda}^{(i_0)},\bm{f}(\bm{\lambda}^{(i_0)})\big)=\sum_{i=1}^{S_L}c_i \cdot\big(\bm{\lambda}^{(i)},\bm{f}(\bm{\lambda}^{(i)})\big).   \label{pf-lemma-no-redundancy-assumption}
		\end{equation}
		Assume without loss of generality that $\bm{\lambda}^{(i_0)}\in\cG_L^0$. Since we assume at the beginning of Appendix \ref{section-proof-lemma-no-redundancy} that $\zeta\geq 2$ for $\bm{\lambda}^{(i_0)}$, we can see from \eqref{def-cG_L} that $\big(\lambda^{(i_0)}_2,\lambda^{(i_0)}_3,\cdots,\lambda^{(i_0)}_L\big)\in\cG_{L-1}^0$ by construction. Let $\cG_{L-1}= \left\{\bm{\lambda}^{(1)}_{L-1}, \bm{\lambda}^{(2)}_{L-1},\cdots, \bm{\lambda}^{(S_{L-1})}_{L-1}\right\}$. Then there exists a unique $j_0\in\{1,2,\cdots,S_{L-1}\}$ such that 
		\begin{equation}
		\bm{\lambda}^{(j_0)}_{L-1}= \left(\lambda^{(i_0)}_2, \lambda^{(i_0)}_3,\cdots, \lambda^{(i_0)}_L\right).   \label{pf-thm-no-redundency-def-lambda(j0)}
		\end{equation}
		By \Cref{lemma-combination-reduce-dim}, \eqref{pf-lemma-no-redundancy-assumption} implies that
		\begin{align}
		&\big(\bm{\lambda}^{(j_0)}_{L-1},\bm{f}(\bm{\lambda}^{(j_0)}_{L-1})\big)  \nonumber  \\ &=\sum_{i=1}^{S_L}c_i\cdot\big((\lambda^{(i)}_2,\lambda^{(i)}_3,\cdots,\lambda^{(i)}_L),\bm{f}(\lambda^{(i)}_2,\lambda^{(i)}_3,\cdots,\lambda^{(i)}_L)\big). \label{pf-no-redundancy-lambda(L-1)-0}
		\end{align}
		Let $\cK_L=\{1,2,\cdots,S_L\}$, $\cI_L=\{i\in\cK_L:c_i\neq 0\}$, $\cK_L^{(j_0)}=\left\{i\in\cK_L:\left(\lambda^{(i)}_2,\lambda^{(i)}_3,\cdots,\lambda^{(i)}_L\right)=\bm{\lambda}^{(j_0)}_{L-1}\right\}$, and
		\begin{equation}
		d_{0}=\sum_{i\in\cK_L^{(j_0)}}c_i.   \label{pf-thm-no-redundency-def-d0}
		\end{equation}
		In the proof of \Cref{thm-alter-region}, we have shown that any vector in $\cF_L^1$ is a conic combination of the vectors in $\cF_L^2$. Then for any $i\in\cK_L\backslash\cK_L^{(j_0)}$, there exists $\left(t^{(i)}_1,t^{(i)}_2,\cdots,t^{(i)}_{S_{L-1}}\right) \in\mathbb{R}_+^{S_{L-1}}$ such that 
		\begin{align}
		&\big((\lambda^{(i)}_2,\lambda^{(i)}_3,\cdots,\lambda^{(i)}_L),\bm{f}(\lambda^{(i)}_2,\lambda^{(i)}_3,\cdots,\lambda^{(i)}_L)\big) \nonumber \\
		&=\sum_{j=1}^{S_{L-1}}t^{(i)}_j\big(\bm{\lambda}^{(j)}_{L-1},\bm{f}(\bm{\lambda}^{(j)}_{L-1})\big).  \label{pf-no-redundancy-lambda(L-1)-1}
		\end{align}
		Substitute \eqref{pf-thm-no-redundency-def-d0} and \eqref{pf-no-redundancy-lambda(L-1)-1} into \eqref{pf-no-redundancy-lambda(L-1)-0}, we have
		\begin{align}
		&\big(\bm{\lambda}^{(j_0)}_{L-1},\bm{f}(\bm{\lambda}^{(j_0)}_{L-1})\big) \nonumber \\
		&=d_0 \big(\bm{\lambda}^{(j_0)}_{L-1},\bm{f}(\bm{\lambda}^{(j_0)}_{L-1})\big)  \nonumber \\
		&\hspace{0.5cm}+\sum_{i\in\cK_{L}\backslash\cK_L^{(j_0)}}c_i\left[\sum_{j=1}^{S_{L-1}}t^{(i)}_{j}\big(\bm{\lambda}^{(j)}_{L-1},\bm{f}(\bm{\lambda}^{(j)}_{L-1})\big)\right]  \nonumber \\
		&=\left(d_0+\sum_{i\in\cK_{L}\backslash\cK_L^{(j_0)}}c_it^{(i)}_{j_0}\right) \big(\bm{\lambda}^{(j_0)}_{L-1},\bm{f}(\bm{\lambda}^{(j_0)}_{L-1})\big)  \nonumber \\
		&\hspace{0.5cm}+\sum_{i\in\cK_{L}\backslash\cK_L^{(j_0)}}c_i\left[\sum_{j\in\cK_{L-1}\backslash\{j_0\}}t^{(i)}_{j}\big(\bm{\lambda}^{(j)}_{L-1},\bm{f}(\bm{\lambda}^{(j)}_{L-1})\big)\right]   \nonumber \\
		&=\left(d_0+\sum_{i\in\cK_{L}\backslash\cK_L^{(j_0)}}c_it^{(i)}_{j_0}\right) \big(\bm{\lambda}^{(j_0)}_{L-1},\bm{f}(\bm{\lambda}^{(j_0)}_{L-1})\big)  \nonumber \\
		&\hspace{0.5cm}+\sum_{j\in\cK_{L-1}\backslash\{j_0\}}\left(\sum_{i\in\cK_{L}\backslash\cK_L^{(j_0)}}c_it^{(i)}_{j}\right) \big(\bm{\lambda}^{(j)}_{L-1},\bm{f}(\bm{\lambda}^{(j)}_{L-1})\big).   \nonumber 
		\end{align}
		Thus,
		\begin{align}
		&\left(1-d_0-\sum_{i\in\cK_{L}\backslash\cK_L^{(j_0)}}c_it^{(i)}_{j_0}\right)\big(\bm{\lambda}^{(j_0)}_{L-1},\bm{f}(\bm{\lambda}^{(j_0)}_{L-1})\big) \nonumber \\
		&= \sum_{j\in\cK_{L-1}\backslash\{j_0\}}\left(\sum_{i\in\cK_{L}\backslash\cK_L^{(j_0)}}c_it^{(i)}_{j}\right) \big(\bm{\lambda}^{(j)}_{L-1},\bm{f}(\bm{\lambda}^{(j)}_{L-1})\big).  \label{pf-thm-no-redundency-combination}
		\end{align}
		\begin{proposition} \label{prop-differ@1st}
			There exists $i\in\cI_L$ such that $\left(\lambda^{(i)}_2,\lambda^{(i)}_3,\cdots,\lambda^{(i)}_L\right)\neq \bm{\lambda}^{(j_0)}_{L-1}$.
		\end{proposition}
		The proof of \Cref{prop-differ@1st} is given in Appendix \ref{section-proof-differ@1st}. The proposition implies that 
		\begin{equation}
		\cI_L\cap(\cK_L\backslash\cK_L^{(j_0)})\neq \emptyset.   \label{pf-thm-no-redundency-cI-nonempty}
		\end{equation}
		For any $i\in\cI_L\cap(\cK_L\backslash\cK_L^{(j_0)})$, we can rewrite \eqref{pf-no-redundancy-lambda(L-1)-1} as follows:
		\begin{align*}
		&\big((\lambda^{(i)}_2,\lambda^{(i)}_3,\cdots,\lambda^{(i)}_L),\bm{f}(\lambda^{(i)}_2,\lambda^{(i)}_3,\cdots,\lambda^{(i)}_L)\big)  \\
		&=t^{(i)}_{j_0}\big(\bm{\lambda}^{(j_0)}_{L-1},\bm{f}(\bm{\lambda}^{(j_0)}_{L-1})\big)+ \sum_{j\in\cK_{L-1}\backslash\{j_0\}}t^{(i)}_j\big(\bm{\lambda}^{(j)}_{L-1},\bm{f}(\bm{\lambda}^{(j)}_{L-1})\big).
		\end{align*}
		Since $(\lambda^{(i)}_2,\lambda^{(i)}_3,\cdots,\lambda^{(i)}_L)\neq \bm{\lambda}^{(j_0)}_{L-1}$, there must exist $j\in\cK_{L-1}\backslash\{j_0\}$ such that
		\begin{equation}
		t^{(i)}_{j}>0.  \label{pf-thm-no-redundency-active-coef-1}
		\end{equation}
		For any $\bm{x},\bm{y}\in\mathbb{R}_+^{L-1}$, define a binary relation `$>$' by $\bm{x}>\bm{y}$ if and only if $(\bm{x}-\bm{y})\in\mathbb{R}_+^{L-1}$, i.e., $\bm{x}$ is strictly greater than $\bm{y}$ in at least one component (cf. \eqref{definition-R-set}). Then for the RHS of \eqref{pf-thm-no-redundency-combination}, we have
		\begin{eqnarray}
		&&\hspace{-0.8cm}\sum_{j\in\cK_{L-1}\backslash\{j_0\}}\left(\sum_{i\in\cK_{L}\backslash\cK_L^{(j_0)}}c_it^{(i)}_{j}\right) \big(\bm{\lambda}^{(j)}_{L-1},\bm{f}(\bm{\lambda}^{(j)}_{L-1})\big)  \nonumber \\
		&=&\sum_{i\in\cK_{L}\backslash\cK_L^{(j_0)}}c_i\left[\sum_{j\in\cK_{L-1}\backslash\{j_0\}}t^{(i)}_{j}\big(\bm{\lambda}^{(j)}_{L-1},\bm{f}(\bm{\lambda}^{(j)}_{L-1})\big)\right]   \nonumber \\
		&=&\sum_{i\in\cI_L\cap(\cK_L\backslash\cK_L^{(j_0)})}c_i\left[\sum_{j\in\cK_{L-1}\backslash\{j_0\}}t^{(i)}_{j}\big(\bm{\lambda}^{(j)}_{L-1},\bm{f}(\bm{\lambda}^{(j)}_{L-1})\big)\right]   \nonumber \\
		&>&\bm{0},   \label{pf-thm-no-redundency-active-coef}
		\end{eqnarray}
		where the inequality follows from \eqref{pf-thm-no-redundency-cI-nonempty}, \eqref{pf-thm-no-redundency-active-coef-1} and the fact that $\bm{\lambda}^{(j)}_{L-1}>\bm{0}$ for all $j\in\cK_{L-1}\backslash\{j_0\}$. Then we can see from \eqref{pf-thm-no-redundency-combination} that
		\begin{equation*}
		\left(1-d_0-\sum_{i\in\cK_{L}\backslash\cK_L^{(j_0)}}c_it^{(i)}_{j_0}\right)\big(\bm{\lambda}^{(j_0)}_{L-1},\bm{f}(\bm{\lambda}^{(j_0)}_{L-1})\big)>\bm{0},
		\end{equation*}
		which implies that
		\begin{equation*}
		1-d_0-\sum_{i\in\cK_{L}\backslash\cK_L^{(j_0)}}c_it^{(i)}_{j_0}>0.
		\end{equation*}
		For each $j\in\cK_{L-1}\backslash\{j_0\}$, let 
		\begin{equation}
		d_j=\frac{1}{1-d_0-\sum_{i\in\cK_{L}\backslash\cK_L^{(j_0)}}c_it^{(i)}_{j_0}}\sum_{i\in\cK_{L}\backslash\cK_L^{(j_0)}}c_it^{(i)}_{j}.   \label{pf-thm-no-redundency-active-coef-d}
		\end{equation}
		It is easy to see that $d_j\geq 0$ for all $j\in\cK_{L-1}\backslash\{j_0\}$. By \eqref{pf-thm-no-redundency-active-coef}, there exists $j\in\cK_{L-1}\backslash\{j_0\}$ such that $d_j>0$. Upon letting $d_j=0$ for $j=j_0$, by \eqref{pf-thm-no-redundency-combination} and \eqref{pf-thm-no-redundency-active-coef-d}, we have 
		\begin{equation*}
		\big(\bm{\lambda}^{(j_0)}_{L-1},\bm{f}(\bm{\lambda}^{(j_0)}_{L-1})\big)=\sum_{j=1}^{S_{L-1}}d_j\big(\bm{\lambda}^{(j)}_{L-1},\bm{f}(\bm{\lambda}^{(j)}_{L-1})\big).   
		\end{equation*}
		This means that \Cref{lemma-no-redundancy} is false for $L-1$. The proposition is proved.
		
		\section{Proof of \Cref{prop-differ@1st}}  \label{section-proof-differ@1st}
		Since $\bm{\lambda}^{(i_0)}\in\cG_L^0$, there exists a unique $\eta_{i_0}\in\{0,1,\cdots,L-1\}$ such that 
		\begin{equation}
		\lambda^{(i_0)}_1=\frac{1}{\eta_{i_0}}\sum_{j=2}^L\lambda^{(i_0)}_j.  \label{pf-prop-differ@1st-theta(0)}
		\end{equation}
		Recall from \eqref{pf-thm-no-redundency-def-lambda(j0)} that $\bm{\lambda}^{(j_0)}_{L-1}= \left(\lambda^{(i_0)}_2, \lambda^{(i_0)}_3,\cdots, \lambda^{(i_0)}_L\right)$. Since we assume at the beginning of Appendix \ref{section-proof-lemma-no-redundancy} that $\zeta\geq 2$ for $\bm{\lambda}^{(i_0)}$, we see that 
		\begin{equation}
		\bm{\lambda}^{(j_0)}_{L-1}\neq \bm{0},   \label{pf-prop-differ@1st-lambda(j0)}
		\end{equation}
		which implies that $\sum_{j=2}^L\lambda^{(i_0)}_j>0$. Then we have $\eta_{i_0}\neq 0$, otherwise $\lambda^{(i_0)}_1=\infty$ in \eqref{pf-prop-differ@1st-theta(0)}. Thus,~ $\eta_{i_0}\in\{1,2,\cdots,L-1\}$. By \Cref{remark-lemma-f-lambda1} following \Cref{lemma-f-lambda1}, we have
		\begin{equation}
		f_{\eta_{i_0}+1}(\bm{\lambda}^{(i_0)})=f_{\eta_{i_0}}(\bm{\lambda}^{(j_0)}_{L-1})=\frac{1}{\eta_{i_0}+1}\sum_{j=1}^L\lambda^{(i_0)}_j.   \label{pf-prop-differ@1st-f(theta+1)-i0}
		\end{equation}
		We now prove the proposition by contradiction. Assume that for all $i\in\cI_L$,
		\begin{equation}
		\left(\lambda^{(i)}_2,\lambda^{(i)}_3,\cdots,\lambda^{(i)}_L\right)=\bm{\lambda}^{(j_0)}_{L-1}.  \label{pf-prop-differ@1st-assumption-2}
		\end{equation}
		This means that for each $i\in\cI_L$,
		\begin{eqnarray}
		\bm{\lambda}^{(i)}&=&\left(\lambda^{(i)}_1,\lambda^{(i)}_2,\lambda^{(i)}_3,\cdots,\lambda^{(i)}_L\right)  \nonumber \\
		&=&\left(\lambda^{(i)}_1,\lambda^{(i_0)}_2,\lambda^{(i_0)}_3,\cdots,\lambda^{(i_0)}_L\right).  \label{pf-prop-differ@1st-lanbda(i)}
		\end{eqnarray}
		Furthermore, since $\bm{\lambda}^{(i_0)}\in\cG_L^0$, we see from \eqref{def-cG_L} that $\bm{\lambda}^{(j_0)}_{L-1}\in\cG_{L-1}^0$ by construction, so that $\left(\lambda^{(i)}_2,\lambda^{(i)}_3,\cdots,\lambda^{(i)}_L\right)\in\cG_{L-1}^0$. Then by \Cref{lemma-lambda(L-1)}, we have $\lambda^{(i)}_1=0$ or $\lambda^{(i)}_{\pi_i(1)}$.
		
		Let $\cI_L^0$ be the subset of $\cI_L$ such that $i\in\cI_L^0$ if and only if $\lambda^{(i)}_1=0$. For $i\in\cI_L^0$, it is easy to see that $\lambda^{(i)}_{\pi_i(L)}=0$. Then upon noting that 
		\begin{equation*}
		\left(\lambda^{(i)}_{\pi_i(1)},\lambda^{(i)}_{\pi_i(2)},\cdots,\lambda^{(i)}_{\pi_i(L-1)}\right)=\left(\lambda^{(i_0)}_2,\lambda^{(i_0)}_3,\cdots,\lambda^{(i_0)}_L\right),
		\end{equation*}
		by \Cref{thm-opt-f} we have 
		\begin{eqnarray}
		f_{\eta_{i_0}+1}(\bm{\lambda}^{(i)})&=&\min_{\beta\in\{0,1,\cdots,\eta_{i_0}\}}\frac{1}{(\eta_{i_0}+1)-\beta}\sum_{j=\beta+1}^{L-1}\lambda^{(i)}_{\pi_i(j)}  \nonumber \\
		&=&\min_{\beta\in\{0,1,\cdots,\eta_{i_0}\}}\frac{1}{(\eta_{i_0}+1)-\beta}\sum_{j=\beta+2}^{L}\lambda^{(i_0)}_{j}  \nonumber \\
		&=&f_{\eta_{i_0}+1}(\bm{\lambda}^{(j_0)}_{L-1})  \nonumber \\
		&<&f_{\eta_{i_0}}(\bm{\lambda}^{(j_0)}_{L-1}),  \nonumber 
		\end{eqnarray}
		where the inequality follows from Lemma 5 in \cite{yeung99}. By \eqref{pf-prop-differ@1st-f(theta+1)-i0}, this implies that 
		\begin{equation}
		f_{\eta_{i_0}+1}(\bm{\lambda}^{(i)})<\frac{1}{\eta_{i_0}+1}\sum_{j=1}^L\lambda^{(i_0)}_j.   \label{pf-prop-differ@1st-f(theta+1)-i1}
		\end{equation}
		
		For $i\in\cI_L\backslash\cI_L^0$, we have $\lambda^{(i)}_1=\lambda^{(i)}_{\pi_i(1)}$. Since $\bm{\lambda}^{(i)}\in\cG_L$, there exists a unique $\eta_i\in\{0,1,\cdots,L-1\}$ such that
		\begin{equation}
		\lambda^{(i)}_1=\frac{1}{\eta_i}\sum_{j=2}^L\lambda^{(i)}_j.  \label{pf-prop-differ@1st-theta(i)}
		\end{equation}
		From \eqref{pf-prop-differ@1st-lambda(j0)} and the assumption in \eqref{pf-prop-differ@1st-assumption-2}, we have $\left(\lambda^{(i)}_2,\lambda^{(i)}_3,\cdots,\lambda^{(i)}_L\right)\neq \bm{0}$. Then from \eqref{pf-prop-differ@1st-theta(i)}, we have $\eta_i\neq 0$, and thus $\eta_i\in\{1,2,\cdots,L-1\}$. Since $c_{i_0}=0$ in \eqref{pf-lemma-no-redundancy-assumption}, we have $\bm{\lambda}^{(i)}\neq \bm{\lambda}^{(i_0)}$ for all $i\in\cI_L$ and hence for all $i\in\cI_L\backslash\cI_L^0$. In light of \eqref{pf-prop-differ@1st-lanbda(i)}, $\bm{\lambda}^{(i)}\neq \bm{\lambda}^{(i_0)}$ implies $\lambda^{(i)}_1\neq \lambda^{(i_0)}_1$, and upon comparing \eqref{pf-prop-differ@1st-theta(0)} and \eqref{pf-prop-differ@1st-theta(i)}, we see that 
		\begin{equation}
		\eta_i\neq \eta_{i_0}.   \label{pf-prop-differ@1st-eta(neq)}
		\end{equation}
		
		Let $\cI_L^1=\{i\in\cI_L\backslash\cI_L^0:\eta_i>\eta_{i_0}\}$ and $\cI_L^2=\{i\in\cI_L\backslash\cI_L^0:\eta_i<\eta_{i_0}\}$. Then we can see from \eqref{pf-prop-differ@1st-eta(neq)} that
		\begin{equation}
		\cI_L^0\cup\cI_L^1\cup\cI_L^2=\cI_L.   \label{pf-prop-differ@1st-cI}
		\end{equation}
		For $i\in\cI_L^1$, we have $\lambda^{(i)}_1<\lambda^{(i_0)}_1$, which is equivalent to
		\begin{equation*}
		\lambda^{(i)}_1<\frac{1}{\eta_{i_0}}\sum_{j=2}^L\lambda^{(i)}_j.
		\end{equation*}
		Thus, by \Cref{lemma-f-lambda1} (i), we have
		\begin{equation}
		f_{\eta_{i_0}+1}(\bm{\lambda}^{(i)})=\frac{1}{\eta_{i_0}+1}\sum_{j=1}^L\lambda^{(i)}_j<\frac{1}{\eta_{i_0}+1}\sum_{j=1}^L\lambda^{(i_0)}_j.  \label{pf-prop-differ@1st-f(theta+1)-i-1}
		\end{equation}
		On the other hand, for $i\in\cI_L^2$, we have $\lambda^{(i)}_1>\lambda^{(i_0)}_1$, which is equivalent to
		\begin{equation*}
		\lambda^{(i)}_1>\frac{1}{\eta_{i_0}}\sum_{j=2}^L\lambda^{(i)}_j.
		\end{equation*}
		Thus, by \Cref{lemma-f-lambda1} (ii), we have
		\begin{equation*}
		f_{\eta_{i_0}+1}(\bm{\lambda}^{(i)})=f_{\eta_{i_0}}(\bm{\lambda}^{(j_0)}_{L-1}),
		\end{equation*}
		which by \eqref{pf-prop-differ@1st-f(theta+1)-i0} implies that 
		\begin{equation}
		f_{\eta_{i_0}+1}(\bm{\lambda}^{(i)})=\frac{1}{\eta_{i_0}+1}\sum_{j=1}^L\lambda^{(i_0)}_j.  \label{pf-prop-differ@1st-f(theta+1)-i-2}
		\end{equation}
		From \eqref{pf-lemma-no-redundancy-assumption} and \eqref{pf-prop-differ@1st-cI}, we have
		\begin{align}
		f_{\eta_{i_0}+1}(\bm{\lambda}^{(i_0)})&=\sum_{i=1}^{S_L}c_i f_{\eta_{i_0}+1}(\bm{\lambda}^{(i)})  \nonumber \\
		&=\sum_{i\in\cI_L}c_i f_{\eta_{i_0}+1}(\bm{\lambda}^{(i)})  \nonumber \\
		&=\sum_{i\in\cI_L^0}c_i f_{\eta_{i_0}+1}(\bm{\lambda}^{(i)})+\sum_{i\in\cI_L^1}c_i f_{\eta_{i_0}+1}(\bm{\lambda}^{(i)}) \nonumber \\
		&\hspace{0.5cm}+\sum_{i\in\cI_L^2}c_i f_{\eta_{i_0}+1}(\bm{\lambda}^{(i)}).   \label{pf-prop-differ@1st-f(theta+1)-i0-2}
		\end{align}
		Comparing \eqref{pf-prop-differ@1st-f(theta+1)-i0} for $f_{\eta_{i_0}+1}(\bm{\lambda}^{(i_0)})$ and \eqref{pf-prop-differ@1st-f(theta+1)-i1}, \eqref{pf-prop-differ@1st-f(theta+1)-i-1}, and \eqref{pf-prop-differ@1st-f(theta+1)-i-2} for $f_{\eta_{i_0}+1}(\bm{\lambda}^{(i)})$, we see that both $\cI_L^0$ and $\cI_L^1$ must be empty in order for the equality in \eqref{pf-prop-differ@1st-f(theta+1)-i0-2} to hold, and hence
		\begin{equation}
		\cI_L=\cI_L^2.  \label{pf-prop-differ@1st-cI=cI2}
		\end{equation}
		
		For any $i\in\cI_L^2$, since $\eta_i<\eta_{i_0}$ and $\eta_i\geq 1$, we see that $\eta_{i_0}\geq 2$. Thus from \eqref{pf-prop-differ@1st-theta(0)}, we have
		\begin{eqnarray}
		\frac{1}{\eta_{i_0}}\sum_{j=1}^L\lambda^{(i_0)}_j&=&\frac{1}{\eta_{i_0}}\left[\frac{1}{\eta_{i_0}}\sum_{j=2}^L\lambda^{(i_0)}_j+\sum_{j=2}^L\lambda^{(i_0)}_j\right]  \nonumber  \\
		&=&\frac{1}{\eta_{i_0}}\left(\frac{1}{\eta_{i_0}}+1\right)\sum_{j=2}^L\lambda^{(i_0)}_j  \nonumber \\
		&<&\frac{1}{\eta_{i_0}-1}\sum_{j=2}^L\lambda^{(i_0)}_j.  \nonumber 
		\end{eqnarray}
		Then by \Cref{lemma-f-lambda1} (i), \eqref{pf-prop-differ@1st-theta(0)} implies that 
		\begin{equation}
		f_{\eta_{i_0}}(\bm{\lambda}^{(i_0)})=\frac{1}{\eta_{i_0}}\sum_{j=1}^L\lambda^{(i_0)}_j<\frac{1}{\eta_{i_0}-1}\sum_{j=2}^L\lambda^{(i_0)}_j.  \label{pf-prop-differ@1st-f(eta0)-2}
		\end{equation}
		Since $\bm{\lambda}^{(i_0)}$ is ordered, by \eqref{pf-prop-differ@1st-theta(0)}, we have
		\begin{equation*}
		\lambda^{(i_0)}_2\leq \lambda^{(i_0)}_1=\frac{1}{\eta_{i_0}}\sum_{j=2}^L\lambda^{(i_0)}_j,
		\end{equation*}
		which implies that
		\begin{equation*}
		\lambda^{(i_0)}_2\leq\frac{1}{\eta_{i_0}-1}\sum_{j=3}^L\lambda^{(i_0)}_j.
		\end{equation*}
		Then by \Cref{lemma-f-lambda1} (i), we have
		\begin{equation}
		f_{\eta_{i_0}-1}(\bm{\lambda}^{(j_0)}_{L-1})=\frac{1}{\eta_{i_0}-1}\sum_{j=2}^L\lambda^{(i_0)}_j.  \label{pf-prop-differ@1st-f(eta0)-3}
		\end{equation}
		It follows from \eqref{pf-prop-differ@1st-f(eta0)-2} and \eqref{pf-prop-differ@1st-f(eta0)-3} that
		\begin{equation}
		f_{\eta_{i_0}}(\bm{\lambda}^{(i_0)})<f_{\eta_{i_0}-1}(\bm{\lambda}^{(j_0)}_{L-1}).   \label{pf-prop-differ@1st-f(i0)}
		\end{equation}
		For $i\in\cI_L^2$, we have $\eta_{i_0}\geq \eta_i+1$. Then by \Cref{lemma-f-lambda1} (ii), \eqref{pf-prop-differ@1st-theta(i)} implies that
		\begin{equation}
		f_{\eta_{i_0}}(\bm{\lambda}^{(i)})=f_{\eta_{i_0}-1}(\bm{\lambda}^{(j_0)}_{L-1}).  \label{pf-prop-differ@1st-f(i)}
		\end{equation}
		Following \eqref{pf-lemma-no-redundancy-assumption}, we have
		\begin{eqnarray}
		&&\hspace{-0.9cm}\left(\lambda^{(i_0)}_2,\lambda^{(i_0)}_3,\cdots,\lambda^{(i_0)}_L\right)  \nonumber \\
		&=&\sum_{i\in\cI_L}c_i\left(\lambda^{(i)}_2,\lambda^{(i)}_3,\cdots,\lambda^{(i)}_L\right)  \nonumber \\
		&=&\sum_{i\in\cI_L^2}c_i\left(\lambda^{(i)}_2,\lambda^{(i)}_3,\cdots,\lambda^{(i)}_L\right)  \label{pf-prop-differ@1st-convex-combin-2} \\
		&=&\bigg(\sum_{i\in\cI_L^2}c_i\bigg)\left(\lambda^{(i_0)}_2,\lambda^{(i_0)}_3,\cdots,\lambda^{(i_0)}_L\right),  \label{pf-prop-differ@1st-convex-combin-3}
		\end{eqnarray}
		where \eqref{pf-prop-differ@1st-convex-combin-2} follows from \eqref{pf-prop-differ@1st-cI=cI2} and \eqref{pf-prop-differ@1st-convex-combin-3} follows from the assumption in \eqref{pf-prop-differ@1st-assumption-2}. Thus we have
		\begin{equation}
		\sum_{i\in\cI_L^2}c_i=1.  \label{pf-prop-differ@1st-sum(ci)}
		\end{equation}
		Then from \eqref{pf-prop-differ@1st-f(i)} and \eqref{pf-prop-differ@1st-sum(ci)}, we see that 
		\begin{equation*}
		\sum_{i\in\cI_L^2}c_if_{\eta_{i_0}}(\bm{\lambda}^{(i)})=\bigg(\sum_{i\in\cI_L^2}c_i\bigg)f_{\eta_{i_0}-1}(\bm{\lambda}^{(j_0)}_{L-1})=f_{\eta_{i_0}-1}(\bm{\lambda}^{(j_0)}_{L-1}),  
		\end{equation*}
		and it follows from \eqref{pf-prop-differ@1st-f(i0)} that
		\begin{equation*}
		\sum_{i\in\cI_L^2}c_if_{\eta_{i_0}}(\bm{\lambda}^{(i)})>f_{\eta_{i_0}}(\bm{\lambda}^{(i_0)}.
		\end{equation*}
		This is a contradiction to \eqref{pf-lemma-no-redundancy-assumption}. Therefore, the assumption in \eqref{pf-prop-differ@1st-assumption-2} is false and the proposition is proved.
		
		\section{Proof of \Cref{lemma-min-sum-multiplication}}  \label{section-proof-lemma-min-sum-multiplication}
		For any permutation $\omega$ on $\{1,2,\cdots,L\}$ and any $\bm{\lambda}\in\mathbb{R}^L_+$, recall from the beginning of \Cref{section-coefficients} that
		\begin{equation*}
		\omega(\bm{\lambda})=\left(\lambda_{\omega(1)},\lambda_{\omega(2)},\cdots,\lambda_{\omega(L)}\right).
		\end{equation*}
		Then for the ordered permutation $\pi$, we have $\pi(\bm{\lambda})=\left(\lambda_{\pi(1)},\lambda_{\pi(2)},\cdots,\lambda_{\pi(L)}\right)$. 
		
		If $\bm{\lambda}=\pi(\bm{\lambda})$, the lemma is immediate. Otherwise, let $\omega_0(i)=i$ for all $i\in\cL$ so that $\omega_0(\bm{\lambda})=\bm{\lambda}$. Set $t=1$ and we sort $\bm{\lambda}$ in descending order by iteration as follows:
		\begin{enumerate}[(i)]
			\item Let $i_t=\min\{i\in\cL:\omega_{t-1}(i)\neq \pi(i)\}$. Let $k_t,j_t$ be any indexes in $\cL$ such that 
			\begin{equation}
			\pi(k_t)=\omega_{t-1}(i_t)  \label{pf-lemma-min-sum-mult-def-kt}
			\end{equation}
			and
			\begin{equation}
			\omega_{t-1}(j_t)=\pi(i_t).  \label{pf-lemma-min-sum-mult-def-jt}
			\end{equation}
			It is easy to check that $k_t>i_t$ and $j_t>i_t$, which implies 
			\begin{equation}
			\lambda_{\pi(i_t)}-\lambda_{\pi(k_t)}\geq 0   \label{pf-lemma-min-sum-mult-positive-1}
			\end{equation}
			and
			\begin{equation}
			R_{j_t}-R_{i_t}\geq 0.   \label{pf-lemma-min-sum-mult-positive-2}
			\end{equation}
			Let $\omega_t(\bm{\lambda})=(\lambda_{\omega_t(1)},\lambda_{\omega_t(2)},\cdots,\lambda_{\omega_t(L)})$ be a permutation of $\omega_{t-1}(\bm{\lambda})$ where we switch $\lambda_{\omega_{t-1}(j_t)}$ and $\lambda_{\omega_{t-1}(i_t)}$, i.e., 
			\begin{equation}
			\omega_t(i)=
			\begin{cases}
			\pi(i_t),& \text{if }i=i_t  \\
			\pi(k_t),& \text{if }i=j_t  \\
			\omega_{t-1}(i),&\text{otherwise.}
			\end{cases}   \label{pf-lemma-min-sum-mult-def-omega(t)}
			\end{equation}
			Then we have
			\begin{align}
			&\sum_{i=1}^L\lambda_{\omega_{t-1}(i)}R_i-\sum_{i=1}^L\lambda_{\omega_t(i)}R_i  \nonumber \\
			&=\left(\lambda_{\omega_{t-1}(i_t)}R_{i_t}+\lambda_{\omega_{t-1}(j_t)}R_{j_t}\right) \nonumber \\
			&\hspace{0.5cm}-\left(\lambda_{\omega_t(i_t)}R_{i_t}+\lambda_{\omega_t(j_t)}R_{j_t}\right)  \nonumber \\
			&=\left(\lambda_{\pi(k_t)}R_{i_t}+\lambda_{\pi(i_t)}R_{j_t}\right)-\left(\lambda_{\pi(i_t)}R_{i_t}+\lambda_{\pi(k_t)}R_{j_t}\right)  \nonumber \\
			&=R_{i_t}(\lambda_{\pi(k_t)}-\lambda_{\pi(i_t)})+R_{j_t}(\lambda_{\pi(i_t)}-\lambda_{\pi(k_t)})  \nonumber \\
			&=(\lambda_{\pi(i_t)}-\lambda_{\pi(k_t)})(R_{j_t}-R_{i_t})  \nonumber \\
			&\geq 0,  \nonumber 
			\end{align}
			where the second equality follows from \eqref{pf-lemma-min-sum-mult-def-kt}, \eqref{pf-lemma-min-sum-mult-def-jt} and \eqref{pf-lemma-min-sum-mult-def-omega(t)}, and the inequality follows from \eqref{pf-lemma-min-sum-mult-positive-1} and \eqref{pf-lemma-min-sum-mult-positive-2}.
			
			\item If $\omega_t(\bm{\lambda})=\pi(\bm{\lambda})$, return $T=t$ and stop. Otherwise, let $t=t+1$ and go back to step (i).
		\end{enumerate}
		At the end of the iteration, $\omega_T(\bm{\lambda})$ is sorted in the same order as $\pi(\bm{\lambda})$, and we have
		\begin{eqnarray}
		&&\hspace{-0.9cm}\sum_{i=1}^L\lambda_iR_i-\sum_{i=1}^L\lambda_{\pi(i)}R_i   \nonumber \\
		&=&\sum_{i=1}^L\lambda_{\omega_0(i)}R_i-\sum_{i=1}^L\lambda_{\omega_T(i)}R_i   \nonumber  \\
		&=&\sum_{t=1}^{T}\left(\sum_{i=1}^L\lambda_{\omega_{t-1}(i)}R_i-\sum_{i=1}^L\lambda_{\omega_{t}(i)}R_i\right)   \nonumber  \\
		&\geq& 0.  \nonumber 
		\end{eqnarray}
		This proves the lemma.
		
		\bibliographystyle{ieeetr}
		\bibliography{guotao_SMDC-region_ref}
		
		\section{Tables of Non-Redundant $\mathbf{\lambda}$} \label{section-table}
		For $L=1,2,\cdots,5$, the vectors $\bm{\lambda}\in\cG_L^0$ and the corresponding $f_{\alpha}(\bm{\lambda})$ are listed in the following tables. The parameter $\theta$ is the integer such that $\lambda_1= \frac{1}{\theta} \sum\limits_{i=2}^L\lambda_i$.
		\begin{table}[!hb]
			\centering
			\setlength\arrayrulewidth{0.7pt}        
			\def\arraystretch{1.0}
			
			\caption{non-redundant constraint for $L=1$.} \label{table-L=1}
			\begin{tabular}{|c|c|}\hline
				$\lambda$ & $f_1(\lambda)$\\ \hline
				1 & 1   \\ \hline
			\end{tabular}
		\end{table}
		\begin{table}[!hb]
			\centering
			\setlength\arrayrulewidth{0.7pt}        
			\def\arraystretch{1.0}
			
			\caption{non-redundant constraints for $L=2$.} \label{table-L=2}
			\begin{tabular}{|c|c||c|c||c|}\hline
				suffix&$\bm{\lambda}$ & $f_1(\bm{\lambda})$ & $f_2(\bm{\lambda})$ & $\theta$ \\ \hline
				- &$(1,0)$ & $1$ & $0$ & $0$  \\ \hline  
				$(1)$&$(1,1)$ & $2$ &\cellcolor{blue!10}$1$ &$1$   \\ \hline 
			\end{tabular}
		\end{table}
		\begin{table}[!hb]
			\centering
			\setlength\arrayrulewidth{0.7pt}        
			\def\arraystretch{1.0}
			
			\caption{non-redundant constraints for $L=3$.} \label{table-L=3}
			\begin{tabular}{|c|c||c|c|c||c|}\hline
				suffix&$\bm{\lambda}$ & $f_1(\bm{\lambda})$ & $f_2(\bm{\lambda})$ & $f_3(\bm{\lambda})$ &$\theta$  \\ \hline
				- &$(1,0,0)$ & $1$ & $0$ & $0$ & $0$  \\ \hline
				$(1,0)$&$(1,1,0)$ & $2$ & $1$ & $0$ & $1$  \\ \hhline{*{6}{|-}}
				\multirow{2}{*}{$(1,1)$}&$(1,1,1)$ & $3$ & $\frac{3}{2}$ &\cellcolor{red!10} $1$ & $2$  \\ \hhline{*{1}{|~}*{5}{|-}}
				&$(2,1,1)$ & $4$ &\cellcolor{red!10} $2$ &\cellcolor{red!10} $1$ & $1$  \\ \hline
			\end{tabular}
		\end{table}
		\begin{table}[!hb]
			\centering
			\setlength\arrayrulewidth{0.7pt}        
			\def\arraystretch{1.0}
			
			\caption{non-redundant constraints for $L=4$.} \label{table-L=4}
			\begin{tabular}{|c|c||c|c|c|c||c|}\hline
				suffix&$\bm{\lambda}$ & $f_1(\bm{\lambda})$ & $f_2(\bm{\lambda})$ & $f_3(\bm{\lambda})$ & $f_4(\bm{\lambda})$ & $\theta$ \\ \hline  
				- &$(1,0,0,0)$ & $1$ & $0$ & $0$ & $0$ & $0$ \\ \hline
				$(1,0,0)$&$(1,1,0,0)$ & $2$  & $1$ & $0$ & $0$ & $1$  \\ \hline
				\multirow{2}{*}{$(1,1,0)$}&$(1,1,1,0)$ & $3$ & $\frac{3}{2}$ &\cellcolor{red!10}$1$ &\cellcolor{red!10}$0$ & $2$  \\ \hhline{*{1}{|~}*{6}{|-}}
				&$(2,1,1,0)$ & $4$  &\cellcolor{red!10}$2$ &\cellcolor{red!10}$1$ &\cellcolor{red!10}$0$ & $1$  \\ \hline
				\multirow{3}{*}{$(1,1,1)$}&$(1,1,1,1)$ & $4$ & $2$ & $\frac{4}{3}$ &\cellcolor{blue!10}$1$ & $3$  \\ \hhline{*{1}{|~}*{6}{|-}}
				&$(\frac{3}{2},1,1,1)$ & $\frac{9}{2}$ & $\frac{9}{4}$ &\cellcolor{blue!10}$\frac{3}{2}$ &\cellcolor{blue!10}$1$ & $2$  \\ \hhline{*{1}{|~}*{6}{|-}}
				&$(3,1,1,1)$ & $6$ &\cellcolor{blue!10}$3$ &\cellcolor{blue!10}$\frac{3}{2}$ &\cellcolor{blue!10}$1$ & $1$  \\ \hline
				\multirow{2}{*}{$(2,1,1)$}&$(2,2,1,1)$ & $6$ & $3$ &\cellcolor{green!25}$2$ &\cellcolor{green!25}$1$ & $2$  \\ \hhline{*{1}{|~}*{6}{|-}}
				&$(4,2,1,1)$ & $8$ &\cellcolor{green!25}$4$ &\cellcolor{green!25}$2$ &\cellcolor{green!25}$1$ & $1$  \\ \hline
			\end{tabular}
		\end{table}
		\begin{table}[!hb]
			\centering
			\setlength\arrayrulewidth{0.7pt}        
			\def\arraystretch{1.0}
			
			\caption{non-redundant constraints for $L=5$.} \label{table-L=5}
			\begin{tabular}{|@{\hspace{0.05cm}}c@{\hspace{0.05cm}}|@{\hspace{0.05cm}}c@{\hspace{0.05cm}}||@{\hspace{0.05cm}}c@{\hspace{0.05cm}}|c|c|c|c||c|}\hline
				suffix&$\bm{\lambda}$ & $f_1(\bm{\lambda})$ & $f_2(\bm{\lambda})$ & $f_3(\bm{\lambda})$ & $f_4(\bm{\lambda})$ & $f_5(\bm{\lambda})$ & $\theta$ \\ \hline 
				- &$(1,0,0,0,0)$ & $1$  & $0$  & $0$  & $0$  & $0$  & $0$  \\ \hline
				$(1,0,0,0)$&$(1,1,0,0,0)$ & $2$  & $1$  & $0$  & $0$  & $0$  & $1$   \\ \hline
				\multirow{2}{*}{$(1,1,0,0)$} & $(1,1,1,0,0)$ & $3$ & $\frac{3}{2}$ &\cellcolor{green!25}$1$ & \cellcolor{green!25}$0$ &\cellcolor{green!25}$0$ & $2$ \\ \hhline{*{1}{|~}*{7}{|-}}
				&$(2,1,1,0,0)$ & $4$  &\cellcolor{green!25} $2$  &\cellcolor{green!25} $1$  &\cellcolor{green!25} $0$  &\cellcolor{green!25} $0$  & $1$  \\ \hline
				\multirow{3}{*}{$(1,1,1,0)$}& $(1,1,1,1,0)$ & $4$ & $2$  & $\frac{4}{3}$  &\cellcolor{yellow!25}$1$  &\cellcolor{yellow!25}$0$  & $3$  \\ \hhline{*{1}{|~}*{7}{|-}}
				&$(\frac{3}{2},1,1,1,0)$ & $\frac{9}{2}$ & $\frac{9}{4}$  &\cellcolor{yellow!25}$\frac{3}{2}$  &\cellcolor{yellow!25}$1$&\cellcolor{yellow!25}$0$& $2$  \\ \hhline{*{1}{|~}*{7}{|-}}
				&$(3,1,1,1,0)$ & $6$ &\cellcolor{yellow!25}$3$ &\cellcolor{yellow!25}$\frac{3}{2}$ &\cellcolor{yellow!25}$1$ &\cellcolor{yellow!25}$0$  & $1$  \\ \hline
				\multirow{2}{*}{$(2,1,1,0)$}&$(2,2,1,1,0)$ & $6$ & $3$  &\cellcolor{red!10}$2$&\cellcolor{red!10}$1$&\cellcolor{red!10}$0$  & $2$  \\ \hhline{*{1}{|~}*{7}{|-}}
				&$(4,2,1,1,0)$ & $8$ &\cellcolor{red!10}$4$&\cellcolor{red!10}$2$&\cellcolor{red!10}$1$ &\cellcolor{red!10} $0$  & $1$  \\ \hline
				\multirow{4}{*}{$(1,1,1,1)$}& $(1,1,1,1,1)$ & $5$ & $\frac{5}{2}$ & $\frac{5}{3}$ & $\frac{5}{4}$ &\cellcolor{blue!10}$1$ & $4$  \\ \hhline{*{1}{|~}*{7}{|-}}
				& $(\frac{4}{3},1,1,1,1)$ & $\frac{16}{3}$ & $\frac{8}{3}$ & $\frac{16}{9}$ &\cellcolor{blue!10}$\frac{4}{3}$ &\cellcolor{blue!10}$1$ & $3$  \\ \hhline{*{1}{|~}*{7}{|-}}
				& $(2,1,1,1,1)$ & $6$ & $3$ &\cellcolor{blue!10}$2$ &\cellcolor{blue!10}$\frac{4}{3}$ &\cellcolor{blue!10}$1$ & $2$  \\ \hhline{*{1}{|~}*{7}{|-}}
				& $(4,1,1,1,1)$ & $8$ &\cellcolor{blue!10}$4$ &\cellcolor{blue!10}$2$ &\cellcolor{blue!10}$\frac{4}{3}$ &\cellcolor{blue!10}$1$ & $1$  \\ \hline
				\multirow{3}{*}{$(\frac{3}{2},1,1,1)$}& $(\frac{3}{2},\frac{3}{2},1,1,1)$ & $6$ & $3$ & $2$ &\cellcolor{green!25}$\frac{3}{2}$ &\cellcolor{green!25}$1$ & $3$  \\ \hhline{*{1}{|~}*{7}{|-}}
				& $(\frac{9}{4},\frac{3}{2},1,1,1)$ & $\frac{27}{4}$ & $\frac{27}{8}$ &\cellcolor{green!25}$\frac{9}{4}$ &\cellcolor{green!25}$\frac{3}{2}$ &\cellcolor{green!25}$1$ & $2$  \\ \hhline{*{1}{|~}*{7}{|-}}
				& $(\frac{9}{2},\frac{3}{2},1,1,1)$ & $9$ &\cellcolor{green!25}$\frac{9}{2}$ &\cellcolor{green!25}$\frac{9}{4}$ &\cellcolor{green!25}$\frac{3}{2}$ &\cellcolor{green!25}$1$ & $1$  \\ \hline
				\multirow{2}{*}{$(3,1,1,1)$}& $(3,3,1,1,1)$ & $9$ & $\frac{9}{2}$ &\cellcolor{yellow!25}$3$ &\cellcolor{yellow!25}$\frac{3}{2}$ &\cellcolor{yellow!25}$1$ & $2$  \\ \hhline{*{1}{|~}*{7}{|-}}
				& $(6,3,1,1,1)$ & $12$ &\cellcolor{yellow!25}$6$ &\cellcolor{yellow!25}$3$ &\cellcolor{yellow!25}$\frac{3}{2}$ &\cellcolor{yellow!25}$1$ & $1$  \\ \hline
				\multirow{3}{*}{$(2,2,1,1)$}& $(2,2,2,1,1)$ & $8$ & $4$ & $\frac{8}{3}$ &\cellcolor{red!10}$2$ &\cellcolor{red!10}$1$ & $3$  \\ \hhline{*{1}{|~}*{7}{|-}}
				& $(3,2,2,1,1)$ & $9$ & $\frac{9}{2}$ &\cellcolor{red!10}$3$ &\cellcolor{red!10}$2$ &\cellcolor{red!10}$1$ & $2$  \\ \hhline{*{1}{|~}*{7}{|-}}
				& $(6,2,2,1,1)$ & $12$ &\cellcolor{red!10}$6$ &\cellcolor{red!10}$3$ &\cellcolor{red!10}$2$ &\cellcolor{red!10}$1$ & $1$  \\ \hline
				\multirow{2}{*}{$(4,2,1,1)$}& $(4,4,2,1,1)$ & $12$ & $6$ &\cellcolor{blue!10}$4$ &\cellcolor{blue!10}$2$ &\cellcolor{blue!10}$1$ & $2$  \\ \hhline{*{1}{|~}*{7}{|-}}
				& $(8,4,2,1,1)$ & $16$ &\cellcolor{blue!10}$8$ &\cellcolor{blue!10}$4$ &\cellcolor{blue!10}$2$ &\cellcolor{blue!10}$1$ & $1$  \\ \hline
			\end{tabular}
		\end{table}
		
	\end{appendices}
	

\end{document}